\documentclass[11pt]{article}
\usepackage{fullpage}
\usepackage{helvet}
\usepackage{courier}
\usepackage{float}

\usepackage[margin=1in]{geometry}
\usepackage{booktabs} % For formal tables
\usepackage[ruled]{algorithm2e} % For algorithms

\SetAlFnt{\small}
\SetAlCapFnt{\small}
\SetAlCapNameFnt{\small}
\SetAlCapHSkip{0pt}
\IncMargin{-\parindent}

\usepackage{nicefrac}
\usepackage{amsthm,amsfonts,amstext, comment, graphicx}

\usepackage[shortlabels]{enumitem}
\usepackage{bm}
\usepackage{algpseudocode}

\usepackage{hyperref}
\usepackage{color}
\usepackage{float}
\usepackage{bbm}
\floatstyle{boxed}
\newfloat{algorithm}{t}{lop}
\usepackage{tikz}
\usepackage{varwidth}
\usepackage{subfigure}

\usepackage{blkarray}

\usepackage{multirow}

\usepackage{mathtools}
\usepackage{float}
\usepackage[classfont=sanserif,langfont=roman,funcfont=italic]{complexity}
\usepackage{caption}
\usepackage{subfigure}
\usepackage{tcolorbox}
\usepackage{color-edits}
\addauthor{ss}{violet} % ss for Suho
\addauthor{sb}{violet} % sb for Simina
\addauthor{mh}{blue} % mh for Mohammad
\addauthor{rp}{olive} % rp for Reed
\addauthor{kw}{red} %kw for Kun

\newtheorem{definition}{Definition}
\newtheorem{proposition}{Proposition}

\newtheorem{remark}{Remark}

\newtheorem{claim}{Claim}
\newtheorem{theorem}{Theorem}
\newtheorem{lemma}{Lemma}

\usepackage{url}

\newtheorem{innercustomprop}{Restatement of Proposition}
\newenvironment{customprop}[1]
{\renewcommand\theinnercustomprop{#1}\innercustomprop}
{\endinnercustomprop}

\newtheorem{innercustomthm}{Restatement of Theorem}
\newenvironment{customthm}[1]
{\renewcommand\theinnercustomthm{#1}\innercustomthm}
{\endinnercustomthm}

\DeclareMathOperator{\Ex}{\mathbb{E}}% expected value
\newcommand{\explain}[1]{\tag{\textcolor{gray}{#1}}}
\newcommand{\bob}{\texttt{Bob}}

\hypersetup{colorlinks=true,linkcolor=blue,citecolor=blue,anchorcolor = blue}
\definecolor{ForestGreen}{rgb}{.05,.50,.05}
\definecolor{darkmagenta}{rgb}{0.55, 0.0, 0.55}

\newcommand{\ie}{{\em i.e.,~\xspace}}

\renewcommand{\eqref}[1]{(\ref{#1})}

\newcommand{\eps}{\varepsilon}

\DeclareMathOperator*{\argmin}{argmin}

\title{Dueling Over Dessert, Mastering the Art of  Repeated Cake Cutting\footnote{Alphabetical author ordering.  S. Br\^anzei was supported in part by US National Science Foundation
		CAREER grant CCF-2238372. This material is based upon work supported by the National Science Foundation under Grant No. DMS-1928930 and by the Alfred P. Sloan Foundation under grant G-2021-16778, while S. Br\^anzei  was in residence at the Simons Laufer Mathematical Sciences Institute (formerly MSRI) in Berkeley, California, during the Fall 2023 semester. This work is partially supported by DARPA QuICC, NSF AF:Small \#2218678, NSF AF:Small \#2114269 and Army-Research Laboratory (ARL) \#W911NF2410052.
		
}}

\author{Simina Br\^anzei\footnote{Purdue University. E-mail: simina.branzei@gmail.com.}
	\and 
	MohammadTaghi Hajiaghayi\footnote{University of Maryland. E-mail: hajiaghayi@gmail.com.}
	\and 
	Reed Phillips\footnote{Purdue University. E-mail: phill289@purdue.edu.} 
	\and
	Suho Shin\footnote{University of Maryland. E-mail: suhoshin@umd.edu.}
	\and 
	Kun Wang\footnote{Purdue University. E-mail: wang5675@purdue.edu.}
}
\date{\today}

\begin{document}

	\maketitle

	% Abstract. Note that this must come before \maketitle.
	\begin{abstract}
		We consider the setting of repeated fair division between two players, denoted Alice and Bob, with private valuations over a cake. In each round, a new cake arrives, which is identical to the ones in previous rounds. Alice  cuts the cake at a point of her choice, while  Bob  chooses the left piece or the right piece, leaving the remainder for Alice. 
		We consider two versions: \emph{sequential}, where Bob observes Alice's cut point  before choosing left/right, and \emph{simultaneous}, where he only observes her cut point after making his choice. The simultaneous version was first considered by Aumann and Maschler \cite{AumannMaschler}.

		We observe that if Bob is almost myopic and chooses his favorite piece too often, then  he can be systematically exploited by Alice through a strategy akin to a binary search. This strategy allows Alice to approximate Bob's preferences with increasing precision, thereby securing a disproportionate share of the resource over time.
		
		We analyze the limits of how much a player can exploit the other one and show that fair utility profiles are in fact achievable. Specifically, the players can enforce  the equitable utility profile  of $(1/2, 1/2)$  in the limit  on every trajectory of play, by keeping the other player's utility to approximately $1/2$ on average while guaranteeing they themselves get at least approximately $1/2$ on average. We show this theorem using  a connection with Blackwell approachability \cite{blackwell1956}.
		
		Finally, we analyze a natural dynamic known as fictitious play, where players best respond to the empirical distribution of the other player. We show that
		fictitious play converges to the equitable utility profile of $(1/2, 1/2)$ at a rate of $O(1/\sqrt{T})$.
	\end{abstract}

	% Paper body
	
	\section{Introduction}
	Cake cutting is a model of fair division~\cite{Steinhaus48}, where the cake is a metaphor for a heterogeneous divisible resource such as land, time, memory in shared computing systems, clean water, greenhouse gas emissions, fossil fuels, or other natural deposits~\cite{Pro13}. The problem is to divide the resource among multiple participants so that everyone believes the allocation is fair. There is an extensive literature on cake cutting in mathematics, political science, economics ~\cite{RW98,BT96,Moulin03} and  computer science \cite{socialchoice_book}, with a number of protocols implemented \cite{GP14}.

	Traditional approaches to cake cutting often consider single instances of division. However, many real-world scenarios require  a repeated division of resources. For instance, consider the  recurring task of allocating classroom space in educational institutions each quarter or that of repeatedly dividing computational resources (such as CPU and memory) among the members of an organization. 
	%While the participants in such scenarios may not initially know each other, they will observe each other's behavior  over time. 
	%These settings reflect the reality of many social and economic interactions and demand a model that not only addresses the fairness of a single division, but also the dynamics and strategies that emerge among the participants over repeated interactions.
	These settings reflect the reality of many social and economic interactions, necessitating a model 
	that not only addresses the fairness of a single division, but also the dynamics and strategies that emerge among  participants over repeated interactions.
	%both the fairness of individual divisions and the evolving dynamics and strategies among participants in repeated interactions

	Repeated fair division is a classic problem first considered by Aumann and Maschler  \cite{AumannMaschler}, where two  players---denoted Alice and Bob---have private valuations over the cake and interact in the following environment. Every day a new cake arrives, which is the same as the ones in previous days. Alice cuts the cake at a point of her choosing, while  Bob chooses either the left piece or the right piece, leaving the remainder to Alice. \cite{AumannMaschler} considered the simultaneous setting, where both players take their actions at the same time each day, and analyzed the payoffs achievable by Bob when he can have one of two types of valuations.

	%\smallskip 
	
	%In this paper, we provide the first substantial progress in this classic setting after 30 years.
	In this paper, we provide the first substantial progress in this classic setting. % after 30 years.
	We further analyze the simultaneous version from \cite{AumannMaschler} and also go beyond it, by considering the sequential version where  Bob has the advantage of observing Alice's chosen cut point before making his selection. Tactical  considerations remain pivotal in the sequential version, which is none other than the  repeated  \emph{Cut-and-choose} protocol with strategic players.

	A key observation in our study is the strategic vulnerability inherent in  repeated Cut-and-choose. At a high level, if Bob consistently chooses his preferred piece, then he can be systematically exploited by Alice through a strategy akin to a binary search. This strategy allows Alice to approximate Bob's preferences with increasing precision, thereby securing a disproportionate share of the resource over time. To fight back Alice's attempt to exploit him, Bob could deceive her  by being unpredictable, thus  hiding his preferences. While this  behavior has the potential to reduce  Alice's share of the cake, it   could also come at the price of  affecting Bob's own payoff guarantees in the long term. 
	% \sscomment{how about Bob trying to deceive at some rounds, instead of unpredictability? mostly because the literature on learning in repeated Stackelberg game uses the the follower tries to deceive the leader.}
	% TODO: Can say Bob can deceive Alice by being unpredictable.
	
	%\smallskip 

	Our analysis of the repeated cake cutting game formalizes the  intuition that Alice can exploit a (nearly) myopic Bob that often chooses his favorite piece.  However, the outcome where Alice gets more value is not necessarily  fair, as she is happier than Bob.  The fairness notion of \emph{equitability} embodies the idea that players should be equally happy, formally requiring  that Alice's value for her allocation should equal Bob's value for his allocation.   Achieving equitability is particularly important  in scenarios with potential  for conflict, such as splitting an inheritance. 
	
	%\smallskip 
	
	%TODO: Mention equitability; explain why $(1/2, 1/2)$ is interesting; because it's equitable; the outcome where Alice gets more value is not fair, she is happier than Bob; we don't want one player to be happier than the other; this is why we consider the equitable payoff and get the $(1/2,1/2)$. 
	
	We show that achieving  equitable outcomes  in the repeated interaction is in fact possible. Specifically, each player has a strategy that  ensures that the other player gets no more than approximately $1/2$ on average, while they themselves also get  approximately $ 1/2$ on average. This approaches the equitable utility profile of $(1/2, 1/2)$ in the limit. We obtain this result by  using a connection with Blackwell approachability~\cite{blackwell1956}. Finally, we consider a natural dynamic known as  fictitious play~\cite{fictitious_play_brown}, where the players best respond to the empirical frequency given by the past actions of the other player. We show that fictitious play   converges to the equitable utility profile of $(1/2, 1/2)$ at a rate of $O(1/\sqrt{T})$.

	\begin{comment}
	High level plan for the intro:
	
	We consider repeated cake cutting between two players.
	\begin{itemize}
	\item Cake cutting is a basic model of fair division; introduced by Steinhaus;  rich history and literature;
	\item In many real settings, players interact repeatedly; for example, a group of people may have to re-divide resources like water or classroom space in schools  every quarter; 
	\item We consider a basic model that captures this type of scenario. Our setting was studied in an example in the book of Aumann and Maschler; there Alice has a probability distribution over the types of Bobs; maybe we want to not mention that in detail but talk about in related work? 
	\item We observe that Alice can exploit Bob if he implements correctly Cut-and-choose; However, Bob could fight back by sometimes choosing his less preferred piece, which could ruin Alice's exploitation plan. 
	\item We analyze the repeated game where Alice cuts the cake at a point of her choosing each day and Bob selects the left or right piece, with Bob taking the remainder. 
	\end{itemize}
	
	\begin{itemize}
	\item Simultaneous + Sequential: Theorem about safety payoffs, in both simultaneous and sequential;
	\item Simultaneous: Fictitious play theorem; 
	\item Sequential: response to a bounded regret opponent;   
	\end{itemize}
	\end{comment}

	\subsection{Model} \label{sec:model}
	
	\subsection*{Cake cutting model for two players} We have a cake, represented by the interval $[0,1]$, and  two players $N = \{A, B\}$, where $A$ stands for Alice and $B$ for Bob. Each player $i$  has a private value density function $v_i : [0,1] \to \mathbb{R}_+$. 
	
	The value of player $i$ for each interval $[c,d] \subseteq [0,1]$ is denoted  $V_i([c,d]) = \int_{c}^{d} v_i(x) \; \textit{dx}$. The valuations are additive: $V_i\left( \bigcup_{j=1}^{m} X_j\right) = \sum_{j=1}^{m} V_i(X_j)$ for all disjoint intervals $X_1, \ldots, X_m \subseteq [0,1]$. By definition, atoms are
	worth zero and the valuations are normalized so that $V_i([0, 1]) = 1$ for each player $i$.
	We assume the densities are bounded: 
	\begin{itemize}
		\item There exist $\delta, \Delta > 0$ such that $\delta \leq v_i(x) \leq \Delta$ for all $x \in [0,1]$.
	\end{itemize}
	
	%are hungry, meaning that  $v_i(x) > 0$ for all $x \in [0, 1]$.
	%\simina{Add lower and upper bound on densities; where theorems will be more general, add a remark.}
	
	A \emph{piece} of cake is a finite union of disjoint intervals. A piece is connected if it’s a single
	interval. An allocation $Z = (Z_A, Z_B)$ is a partition of the cake among the players such that each player $i$ receives piece $Z_i$, the pieces are disjoint, and $\bigcup_{i \in N} Z_i = [0,1]$.
	The valuation (aka utility or payoff) of player $i$ at an allocation $Z$ is $V_i(Z_i)$.
	An allocation $Z$ is \emph{equitable} if the players are equally happy with their pieces, meaning $V_A(Z_A) = V_B(Z_B)$. 
	
	\begin{comment}
	%=============================================
	
	All the discrete cake cutting protocols operate in a query model known as the Robertson-Webb model (see, e.g., the book  \cite{RW98}), which was explicitly stated in~\cite{WS07}. In this model, the protocol communicates with the players using the following types of queries:
	\begin{itemize}
	\item{}$\emph{\textit{Cut}}_i(\alpha)$: Player $i$ cuts the cake at a point $y$ where $V_{i}([0,y]) = \alpha$, where $\alpha \in [0,1]$ is chosen arbitrarily by the mediator \footnote{Ties are resolved deterministically, using for example the leftmost point with this property.}. The point $y$ becomes a {\em cut point}.
	\item{}$\emph{\textit{Eval}}_i(y)$:  Player $i$ returns $V_{i}([0,y])$, where $y$ is a previously made cut point.
	\end{itemize}
	
	A protocol asks the players a sequence of cut and evaluate queries, at the end of which it outputs an allocation demarcated by cut points from its execution (i.e. cuts discovered through queries).
	Note that the value of a piece $[x,y]$ can be determined with two Eval queries, $Eval_i(x)$ and $Eval_i(y)$.
	
	We will focus on the \emph{Cut-and-Choose protocol}: Alice  cuts the cake in two pieces that it values equally; Bob then chooses the piece that it prefers,
	leaving Alice with the remaining piece. 
	
	\end{comment}
	
	%==========================================================
	
	%Suppose Alice cuts the cake at a point of her choosing, while  Bob will choose the left piece or the right piece, leaving the remainder to Alice. 
	Let $m_A$ be Alice's midpoint of the cake and $m_B$ Bob's midpoint. Alice's \emph{Stackelberg value}, denoted $u_A^*$, is the utility she receives when she cuts the cake at  $m_B$ and Bob chooses his favorite piece, breaking ties in Alice's favor. 
	%Bob's Stackelberg value, denoted $u_B^*$, is the utility Bob receives in the reversed game where he cuts the cake at Alice's midpoint and Alice chooses her favorite piece, breaking ties in Bob's favor. 
	The midpoints and Alice's Stackelberg value are depicted in Figure~\ref{fig:midpoints_stackelberg}.

	\begin{figure}[h!]
		\centering
		\includegraphics[scale=0.5]{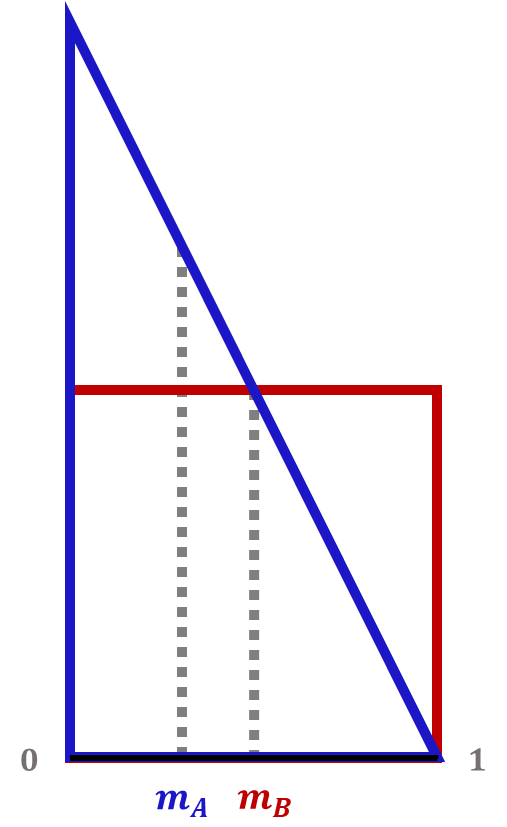} \qquad   \;    
		\includegraphics[scale=0.5]{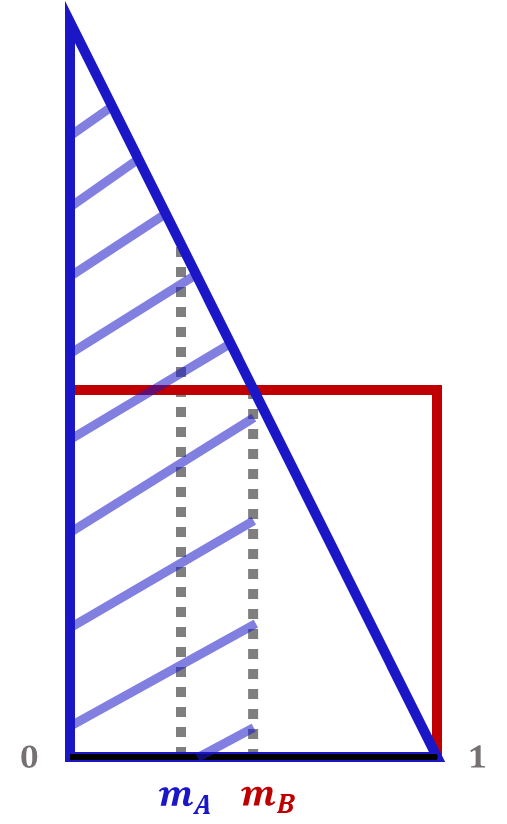}
		\caption{Illustration of densities for Alice and Bob, with blue and red, respectively. Figure (a) shows Alice's midpoint at $m_A$ and Bob's midpoint at $m_B$. Figure (b) shows Alice's Stackelberg value, captured by the blue shaded area.}
		\label{fig:midpoints_stackelberg}
	\end{figure}
	
	\begin{comment}
	\begin{figure}[h!]
	\centering
	\subfigure[]{
	\includegraphics[scale=0.42]{cake_midpoints_nice.png}
	}\qquad  \qquad \qquad  \qquad 
	\subfigure[]{
	\includegraphics[scale=0.42]{cake_midpoints_nice 6.png}
	}%\qquad  \qquad \qquad  \qquad 
	% \subfigure[]{
	%\includegraphics[scale=0.39]{cake_midpoints_nice 5.png}
	%}
	\caption{Illustration of densities for Alice and Bob, with blue and red, respectively. Figure (a) shows Alice's midpoint at $m_A$ and Bob's midpoint at $m_B$. Figure (b) shows Alice's Stackelberg value, captured by the blue shaded area.}
	\label{fig:midpoints_stackelberg}
	\end{figure}
	\end{comment}

	\subsection*{Repeated cake cutting} Each round $t = 1, 2, \ldots, T$, the next steps take place:
	\begin{itemize}
		\item A new cake arrives, which is identical to the ones in previous rounds.
		\item Alice cuts the cake at a point $a_t \in [0,1]$ of her choice.
		%\item Bob chooses one of the two resulting pieces  $\bigl\{[0,a_t], [a_t, 1]\bigr\}$, while Alice takes the remainder.
		\item Bob chooses  either the left piece  or the right piece, then Alice takes the remainder.
	\end{itemize}
	
	We consider two versions: \emph{sequential}, where Bob observes Alice's cut point $a_t$ before choosing left/right \footnote{The sequential version is the repeated \emph{Cut-and-choose} protocol, where the players may not necessarily be honest.}, and \emph{simultaneous}, where he only observes her cut point after making his choice. 
	
	%Each player only knows their own valuation, but can observe the actions taken by the other player.
	
	A pure strategy is a map from the history observed by a player to the next action to play. A mixed strategy is a probability distribution over pure strategies.

	\subsection{Our Results} \label{sec:our_results}
	
	Our results will examine how  players fare in the repeated game over $T$ rounds.
	Given a history $H$, Alice's Stackelberg regret is 
	\begin{align} \label{eq:regret_stackelberg_informal}
		\text{Reg}_{A}(H) = \sum_{t=1}^{T} \left[ u_A^* -  u_A^t(H) \right],
	\end{align}
	where $u_A^*$ is Alice's Stackelberg value and $u_A^t(H)$ is Alice's utility in round $t$ under history $H$.
	
	Suppose Alice uses a mixed strategy $S_A$ and Bob uses a mixed strategy $S_B$. Then $S_A$ ensures Alice's Stackelberg regret is at most $\gamma$ against $S_B$ if 
	$
	\text{Reg}_{A}(H) \leq \gamma 
	$ for all $T$-round histories $H$ that could have arisen under the strategies $(S_A, S_B)$.
	Precise definitions for strategies and regret  can be found in Section \ref{sec:prelim}.

	\subsubsection*{\textbf{Alice exploiting Bob}}
	We start with the following observation about the sequential setting. If Bob chooses his favorite piece in each round, then Alice can exploit him in the long run by running binary search on his midpoint until identifying it within a small error and then cutting near it for the rest of time. This will lead to Alice getting essentially her Stackelberg value in all but $O(\log{T})$ rounds, while Bob will  get  $1/2$ in all but $O(\log{T})$ rounds.
	\begin{proposition}\label{prop:myopic-bob}
		If Bob plays myopically in the sequential setting, then Alice has a strategy that ensures her  Stackelberg regret is $O(\log{T})$.
	\end{proposition}

	This  exploitation phenomenon holds more generally: if Bob's strategy has bounded regret with respect to the standard of selecting his preferred piece in every round {in hindsight}, then Alice can almost get her Stackelberg value in each round. Her Stackelberg regret is a function of Bob's regret guarantee, as quantified in the next theorem. 
	%\sscomment{I think we need to discuss whether there exists such an algorithm for Bob. \Ie how Bob can (efficiently) implement a no-regret learning in our setup.}
	%\simina{We are not sure what you mean here, but I think it is clear that Bob can simply behave according to the benchmark of choosing the best piece in each round.}
	\begin{theorem}[Exploiting a nearly myopic Bob] \label{thm:response_bounded_regret_Bob}
		Let $\alpha \in [0,1)$. Suppose Bob plays a strategy that ensures  his regret %---compared to the standard of selecting the optimal piece in every round---
		is %that guarantee his regret  with respect to the benchmark of choosing the best piece in each round is 
		$O(T^{\alpha})$ in the sequential setting. Let $\mathcal{B}^{\alpha}$ denote the set of all such Bob strategies. %Suppose Bob plays a strategy in $ \mathcal{B}^{\alpha}$.
		%, which ensures his  regret is $O(T^{\alpha})$ with respect to his best pure quantized response  in hindsight, denoted $S_B^{*}$.   
		%a fraction $\lambda \in [0,1]$ of the surplus. If  $\lambda \in [0,1)$,  then Alice has a strategy $S_A$ with negative regret that does not require knowledge of $\alpha$ or $\lambda$. 
		%If $\lambda = 1$:
		%   \begin{description}
		%\item[$\bullet$] If  $\lambda \in [0,1)$,  Alice has a strategy $S_A$ that ensures her regret is negative and does not require knowledge of $\alpha$ or $\lambda$. 
		%\item[$\bullet$ $\lambda = 1$] (her Stackelberg value):  
		\begin{description}
			\item[$\bullet$]  If Alice knows $\alpha$, she has a strategy $S_A = S_A(\alpha)$ that ensures her Stackelberg regret   is ${O}\Bigl(T^{\frac{\alpha+1}{2}} \log T\Bigr)$. Moreover, Alice's Stackelberg regret is  $\Omega\left(T^{\frac{\alpha+1}{2}}\right)$ for some Bob strategy in $\mathcal{B}^{\alpha}$.
			\item[$\bullet$]  If Alice does not know $\alpha$, she has a strategy $S_A$ that ensures her Stackelberg regret is $O\Bigl(\frac{T}{\log{T}}\Bigr)$. %but at least   $\Omega(T/f(T))$  for some Bob strategy in $\mathcal{R}_B^{\alpha}$, where $f$ grows slower than any polynomial. 
			{This is essentially optimal: if $S_A$ guarantees Alice Stackelberg regret %Alice guarantees herself regret 
				$O(T^{\beta})$ against all Bob strategies in $\mathcal{B}^{\alpha}$ for some $\beta \in [0,1)$, then $S_A$ has Stackelberg regret $\Omega(T)$ for some Bob strategy in $\mathcal{B}^{\beta}$.}
		\end{description}
	\end{theorem}

	In contrast, in the simultaneous setting, Alice may not approach her Stackelberg value on \emph{every} trajectory of play. In order to get her Stackelberg value in any given round, Alice needs to cut near Bob's midpoint and Bob needs to pick the piece he prefers, say $R$. 
	However, if Bob deterministically commits to picking $R$, he will be completely exploited by an Alice who cuts at $1$, breaking any reasonable regret guarantee he might have.
	Indeed, any Bob with a deterministic strategy (possibly using different actions over the rounds) has a corresponding Alice who can completely exploit him.
	Therefore, any Bob strategy with a good regret guarantee would behave randomly, making it impossible for Alice to reliably get her Stackelberg value on every trajectory. For this reason, we focus on the sequential setting when studying how Alice can exploit Bob.\footnote{One may explore a weaker regret benchmark for Bob of always picking the better of left or right in hindsight, which we believe would be an interesting direction for future work.}
	
	% But in the simultaneous setting, Bob cannot see Alice's cut point, so he has no incentive to play along. 
	% In fact, he is \emph{disincentivized} from picking the correct piece deterministically, since Alice could have changed her cut point and he has no way to know. 
	% Therefore, we focus on the sequential setting when studying how Alice can exploit Bob.
	%\sscomment{This para looks informal.} \simina{It is informal since we don't have a theorem, so this explains the intuition. Is there anything specific you want to change? I made some light edits} 

	%In contrast,  the opportunity for Alice to exploit Bob in the simultaneous setting is reduced. Since the actions are simultaneous, Bob can randomize by taking  the left or right piece with equal probability, which leaves  Alice without a strategy that can guarantee her Stackelberg value. 
	%\sscomment{I might be confused: isn't this also true for sequential setting? Simultaneous setting is obviously advantageous for Bob as Bob can commit to an algorithm that simply ignores Alice's action, so this sentence seems not correct. Also, best piece for each round in hindsight gives Bob at least ex-post half each round, whereas randomizing (in the simultaneous setting) only gives half in an ex-ante manner, which (may) increase Alice's payoff (but only for this certain randomized strategy).}
	%The best guarantee that Bob can get  (regardless of Alice's strategy) is also weaker in the simultaneous setting compared to the sequential one: utility $1/2$ in expectation, rather than on every trajectory.
	
	\subsubsection*{\textbf{Equitable payoffs.}}
	% \subsubsection*{\textbf{\em Equitable payoffs.}}
	Motivated by   Theorem~\ref{thm:response_bounded_regret_Bob}, we examine the general limits of how much a player can exploit the other one and whether fair  outcomes are achievable,  in both the sequential and simultaneous settings. 
	
	Given a history $H$, player $i$ is said to get an average payoff of $\gamma$  if
	\begin{align} \label{eq:intro_average_payoff_gamma_def}
		\left(\frac{1}{T} \right)\sum_{t=1}^{T} {u_i^t(H)} = \gamma
		\,.
	\end{align}
	The left hand side of \eqref{eq:intro_average_payoff_gamma_def} is not expected utility, but rather the observed total utility averaged over $T$ rounds.
	
	We say a utility profile $(u_A, u_B)$ is \emph{equitable} if $u_A = u_B$. In the single round setting, $u_A$ and $u_B$ will naturally represent the utilities of the players at an allocation. In the repeated setting, $u_A$ and $u_B$ will represent the time-average utilities of the players.
	
	The next theorems show that the  equitable utility profile of $(1/2, 1/2)$ can be approached on every trajectory of play, by having a player keep  the other's utility to approximately $1/2$ on average while ensuring they themselves get at least approximately $1/2$ on average. %\sscomment{(in the worst case)?}.
	%This  bounds how much a player can exploit the other one. Equitable utility profiles are relevant in settings with heightened conflict, such as acrimonious property division. 

	\begin{theorem}[Alice enforcing equitable payoffs; informal] \label{thm:Alice_safety_payoffs_summary}
		%Suppose Bob uses a mixed quantized response strategy $S_B \in \mathcal{Q}_B$. 
		In both the sequential and simultaneous settings, Alice has a pure strategy $S_A$, such that for every  Bob strategy $S_B$:
		\begin{itemize}
			\item   on every trajectory of play, Alice's  average payoff is  at least  $1/2 - o(1)$, while Bob's average payoff is at most $1/2 + o(1)$. More precisely, %for all $t \in [T]$:
			\begin{align}
				\frac{u_A}{T} \geq \frac{1}{2} -  \Theta \left( \frac{1}{\sqrt{T}}\right)  \qquad \mbox{and} \qquad  
				\frac{u_B}{T} \leq \frac{1}{2} + \Theta \left( \frac{1}{\ln{T}}\right), 
				\notag
			\end{align}
			where $u_i$ is the cumulative payoff of player $i$ over the time horizon $T$.
		\end{itemize}
	\end{theorem}
	%  \frac{5\Delta+11}{\ln(2t/5)}

	A key ingredient in the proof of Theorem~\ref{thm:Alice_safety_payoffs_summary} is  a connection with Blackwell's approachability theorem~\cite{blackwell1956}.  Generally  speaking, Blackwell approachability can be used by a player to limit the payoff of the other player in a certain region of the utility profile. However, the main challenge is that there are uncountably many types of Bob and so Alice cannot apply the strategy from Blackwell directly. Instead, Alice's strategy constructs a countably infinite set of representatives, which allows us to adapt Blackwell's argument to this setting.
	%\sscomment{I thought this is about theorem 3 - and then realized was about theorem 2. would it be better to write this down directly below theorem 2?}}

	\medskip 
	We show a symmetric theorem for Bob in the sequential setting, while in the simultaneous setting Bob's guarantee  only holds in expectation.
	%\sscomment{I couldn't find where the theorem statement says about "this is the best possible". I think this is relevant to Bob's randomized strategy and my previous comments?}
	
	\begin{theorem}[Bob enforcing equitable payoffs; informal] \label{thm:safety_payoffs_summary}
		%Suppose Bob uses a mixed quantized response strategy $S_B \in \mathcal{Q}_B$. 
		\phantom{s}
		%Bob has a strategy $S_B$, such that for every Alice strategy $S_A$:
		\begin{itemize}
			\item \emph{In the sequential setting:}  Bob has a pure strategy $S_B$, such that for every Alice strategy $S_A$, on every trajectory of play, Bob's average payoff is at least $1/2 - o(1)$, while Alice's average payoff is at most  $1/2+o(1)$. More precisely,
			\begin{align}
				\frac{u_B}{T} \geq \frac{1}{2} - \frac{1}{\sqrt{T}} \qquad \mbox{and} \qquad  \frac{u_A}{T} \leq \frac{1}{2} + \Theta\left(\frac{1}{\sqrt{T}}\right) \,. \notag 
			\end{align}
			\item \emph{In the simultaneous setting:}  Bob has a mixed strategy $S_B$, such that for every Alice strategy $S_A$,  both players have average payoff $1/2$ in expectation. % Bob's average payoff is $1/2$ in expectation, while Alice's average payoff is $1/2$ in expectation.
		\end{itemize}  
	\end{theorem}
	
	%TODO: Couyld put here a sentence about learner vs adversarial environment; worst case, can use minimax; then there is beyond worst case ; see discussion with Suho;

	\subsubsection*{\textbf{Fictitious play}}
	Fictitious play is a classic learning rule where 
	at each round, each player  best responds to the empirical frequency of play of the other player. Fictitious play was introduced in \cite{fictitious_play_brown}.  Convergence to Nash equilibria has been shown for zero-sum games \cite{robinson51} and special cases of general-sum games~\cite{Nachbar90,MS96a,MS96b}.
	
	In the cake cutting model, learning rules such as  fictitious play are  more meaningful in the simultaneous setting, where there is uncertainty for both players due to the simultaneous actions. 
	%Thus  best responding to the empirical frequency of play of the other player is a way to handle the uncertainty. 
	The precise definition of the fictitious play dynamic is in Section~\ref{sec:fictitious_play}, while an example of trajectories for an instance  with random valuations and uniform random tie-breaking can be found in Figure~\ref{fig:fictitious_play}.
	
	\begin{figure}[h!]
		\centering
		\subfigure[Alice's utility.]{
			\includegraphics[scale=0.6]{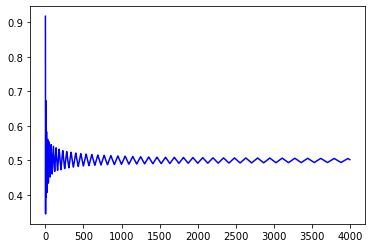}
		}  
		\subfigure[Bob's utility.]{
			\includegraphics[scale=0.6]{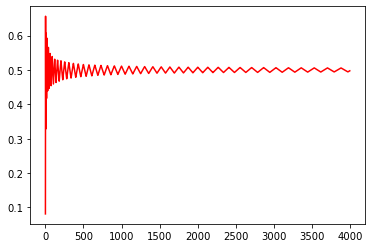}
		}
		\caption{Illustration of Alice's and Bob's average payoff in a randomly generated instance of valuations. The X axis shows the time and the Y axis shows the average payoff up to that round.}
		\label{fig:fictitious_play}
	\end{figure} 
	
	The  convergence properties of fictitious play can be characterized as follows. %is that fictitious play converges to an equitable utility profile where both players get $(1/2, 1/2)$ at a rate of $O\left(\frac{1}{\sqrt{T}}\right)$.  
	
	\begin{theorem}[Fictitious Play; informal] 
		\label{thm_dynamic}
		When both Alice and Bob run fictitious play, the  average payoff  of each player converges to $1/2$ at  a rate of $O\left(\frac{1}{\sqrt{T}}\right)$.
	\end{theorem}

	\subsubsection*{\bf Roadmap to the paper.} % The model is presented in Section~\ref{sec:model}. A summary of our  results is included in Section~\ref{sec:our_results}. 
	Related work is surveyed in Section~\ref{sec:related_work}. Preliminaries, which state the necessary  notation, can be found in Section \ref{sec:prelim}. 
	
	An overview of how Alice  can exploit a nearly myopic Bob can be found in Section~\ref{sec:alice_exploiting_bob_main}, with formal proofs in Appendix~\ref{app:alice_exploiting_bob}. An overview of how  players can enforce equitable payoffs can be found in Section~\ref{sec:equitable_payoffs_main}, with formal proofs in Appendix~\ref{app:safety_payoffs}. Fictitious play  can be found in Section~\ref{sec:fictitious_play}, with formal proofs in Appendix~\ref{app:fictitious_play}.
	
	\section{Related work} \label{sec:related_work}
	
	We overview related work in several areas, including fair division,  repeated games and their connections with learning. 
	\paragraph{Cake cutting and fairness notions.}
	The cake cutting model was introduced in~\cite{Steinhaus48} to capture  the allocation of a heterogeneous resource among agents with complex preferences. There are $n$ players, each with a private value density over the cake, which is represented as the interval $[0,1]$. The goal is to find a ``fair'' allocation of the cake among the $n$ players, where the fairness notion could include   proportionality, envy-freeness \footnote{An allocation is envy-free if no player prefers the piece of another player to their own.}, and equitability. The related task of  necklace splitting~\cite{Alon_necklace} requires dividing the cake in $t$ pieces (not necessarily contiguous)  such that each player has value $1/t$ for each of the pieces. Special cases include the  consensus halving problem, where $t=2$,  and finding perfect allocations, where $t=n$.  For surveys, see \cite{RW98,BT96,Moulin03,socialchoice_book,Pro13}. % with multiple protocols  implemented \cite{GP14}.

	\paragraph{Existence of fair allocations.} Proportional allocations with contiguous pieces can be computed efficiently for  any number of players~\cite{EP84,DS61}. However, the existence of more stringent fairness notions such as envy-freenes is typically ensured by a fixed point theorem. An envy-free allocation with contiguous pieces exists for any instance with $n$ of players via an application of Sperner's lemma or Brouwer's fixed point theorem~\cite{edward1999rental,Stromquist80}, while the existence of perfect partitions is shown via  Borsuk-Ulam~\cite{Alon_necklace}. 
	
	\paragraph{Complexity of cake cutting.}	 There is a standard query model for cake cutting, called the Robertson-Webb (RW) query model~\cite{woeginger2007complexity}, where  a mediator (algorithm) helps the players find a fair division by asking them enough queries about their preferences  until  it has  sufficient information to output  a fair allocation.
	
	Proportional allocations with contiguous pieces can be computed with  $O(n \log{n})$ queries
	\cite{EP84}, with matching lower bounds for finding proportional (not necessarily contiguous) allocations  were given in  \cite{woeginger2007complexity,EP06b}. 
	
	%The query complexity of other fairness notions is partially understood. 
	For the query complexity of exact envy-free cake cutting (possibly with disconnected pieces), a lower bound of $\Omega(n^2)$ was given by \cite{Pro09} and an upper bound of $O\bigl(n^{n^{n^{n^{n^{n}}}}}\bigr)$ by \cite{AM16}. ~\cite{AFMPV18} designed a simpler algorithm  for $4$ agents. The problem of computing envy-free allocations with high probability was solved by \cite{Cheze_whp_ef}.

	Equitable, envy-free, or perfect allocations with \emph{minimum number of cuts} cannot be computed exactly by  RW protocols \cite{Stromquist08}, thus $\varepsilon$-fairness becomes the goal. The query complexity 
	of finding $\varepsilon$-fair (envy-free, equitable, perfect) allocations was analyzed in \cite{CDP13,PW17,BN17,BN19,HR23}. 
	%%%Longer version of related work on query complexity , commenting it out::: An upper bound on the query complexity of $\varepsilon$-equitability was given by \cite{CDP13} and a lower bound by \cite{PW17}.  \cite{BN17} analyzed the query complexity of  $\varepsilon$-envy-freeness among $n$ players, $\varepsilon$-perfection among $n=2$ players, and $\varepsilon$-equitability for $n=2$ players.  \cite{HR23}  showed  an improved bound for finding an $\varepsilon$-envy-free allocation among $n=4$ players with monotone valuations and found the problem becomes   PPAD-hard with non-monotone valuations, also giving   communication complexity lower bounds. The communication complexity of $\varepsilon$-fair (envy-free, equitable, perfect) cake cutting was   studied in \cite{BN19}  among $n=2,3$ players.
	\cite{DQS12} studied the  complexity of cake cutting in the white box model, showing that  $\varepsilon$-envy-free cake cutting with contiguous pieces  is PPAD-complete when the valuations  are  unrestricted (i.e.  not necessarily induced by value densities,  allowing externalities). See  \cite{GoldbergHS20,FHHH21,Halevi18} for more on the complexity of cake cutting. %	The query complexity of cake cutting in one round, i.e. in the simultaneous setting, was studied in \cite{BBKP14}.  

	The complexity of consensus halving was studied, e.g., in  ~\cite{DeligkasFH21,GoldbergHIMS20,FHSZ20} and   the complexity of necklace splitting in~\cite{AlonGraur,filos2021topological}.
	
	\paragraph{Incentives in cake cutting.} A body of work studied truthful cake cutting  in the RW query model~\cite{mossel_tamuz,BM15} and in the direct revelation model~\cite{truth_justice_and_cake_cutting,BU2023103904,bei2022truthful,Tao22}. The equilibria of cake cutting protocols with strategic agents were considered in~\cite{NY08,branzei2013equilibrium}, with an algorithmic model of generalized cut-and-choose  protocols proposed  in \cite{BranzeiCKP16}. \cite{IaruGoldberg21} designed branch-choice protocols   as a simpler yet expressive alternative.
	
	\paragraph{Multiple divisible/indivisible  goods and chores. New models.}

	The complexity of finding fair allocations in settings with multiple  indivisible  was considered in \cite{OPS21_few_queries,plaut2020almost,PlautR19,MS21_closing_gaps,ChaudhuryKMS21,BiloCFIMPVZ19} among many others. 
	For more on the fair allocation of indivisible goods and applicable solution concepts, see, e.g., \cite{AmanatidisBFV22,Procaccia_EFX_survey,chaudhury2020efx,chaudhury2021improving,procaccia2014fair}.    For the allocation of divisible or indivisible bads, see, e.g., \cite{kulkarni2021indivisible,chaudhury2021competitive}. 
	
	\cite{ghodsi2011dominant} studied fairness in settings inspired by cloud computing, where there are  multiple divisible goods (e.g. CPU and memory) and the users have to run jobs with different resource requirements. \cite{ghodsi2011dominant} introduced the dominant resource fairness (DRF) mechanism and showed it has strong fairness and incentive properties. \cite{parkes2015beyond} further studied the properties of the DRF mechanism  for indivisible goods. 
	\cite{kandasamy2020online} studied the allocation of multiple goods in settings where the players do not know their own resource requirements, with the goal of designing mechanisms that guarantee efficiency, fairness, and strategy-proofness.

	% There has been a resurgence  of interest in richer models.
	Cake cutting with separation was studied in~\cite{ElkindSS21}, fair division of a graph or graphical cake cutting in \cite{graph_fair_division,Bei2019DividingAG,deligkas2022complexity},  multi-layered cakes in ~\cite{IgarashiM21},   cake cutting where some parts are good and others bad  in ~\cite{cake_burned}, and  when the whole cake is a bad  (e.g. chore) in \cite{ijcai2018p31,DehghaniFHY18,hajiaghayi2023almost}.  Cake cutting in two dimensions was studied in \cite{segal2017fair} and cake cutting in practice in \cite{KyropoulouOS22}.
	The allocation of multiple homogeneous divisible goods was studied in \cite{CaragiannisGP022}.

	\paragraph{Dynamic fair division.} Closest to our setting is the analysis in the book of \cite{AumannMaschler} (page 243), where two players are dividing a cake with a cherry. Alice (the cutter)  has a uniform density and so she does not care for the cherry, while  Bob (the chooser) may or may not like the cherry. Alice and Bob declare their actions simultaneously and Alice is only allowed to cut in one of two locations. Additionally, Alice has a prior over the type of Bob she is facing. The question is how the players should behave in the repeated game.  \cite{AumannMaschler} analyze the set of payoffs approachable for Bob using Blackwell approachability. The difference from our setting is that we allow arbitrary value densities for the players and  do not assume priors. Additionally, we also consider the sequential version. 
	
	Online cake cutting was  considered by \cite{walsh2011online}, in the setting where agents can arrive and depart over time and the goal is to still ensure some form of fairness. Dynamic fair division with multiple divisible goods was studied in \cite{kash2014no}, with  agents arriving over time but  not departing and  the allocation algorithm taking irrevocable decisions (i.e., the resources can never be taken back once given to an agent). 
	
	\cite{friedman2015dynamic} considered the allocation of a single unit of a divisible resource when the arrivals and departures of the agents are uncertain. The question  is how to maintain  fairness and Pareto efficiency in a way that minimizes the disruptions to existing allocations.  \cite{benade2022dynamic} consider dynamic fair division with partial information, where  $T$ goods become available over a sequence of rounds,  every good  must be allocated immediately and irrevocably before the next one arrives, and the valuations are drawn from an underlying distribution. 
	
	\paragraph{Learning in repeated Stackelberg games.}
	The Stackelberg game (competition) was first 
	introduced by~\cite{stackelberg1934marktform} to understand the first  mover advantage of firms when entering a market. The Stackelberg equilibrium concept has  received significant interest in economics and computer science,  with  real world applications such as security games  \cite{tambe2011security,balcan2015commitment}, online strategic classification \cite{dong2018strategic}, and online principal agent problems  \cite{hajiaghayi2023regret}. Our model can be seen  as each player facing an  online learning version of a repeated Stackelberg game.  
	% ~\cite{tambe2011security,balcan2015commitment,dong2018strategic,hajiaghayi2023regret}.
	
	\cite{kleinberg2003value} considered a seller's problem of designing an efficient repeated posted price mechanism to buy identical goods when it interacts with a sequence of myopic buyers.
	% \cite{balcan2015commitment} studied a repeated version of Stackelberg security game in which at each round a defender commits to a deployment of security resources and then an attacker attacks to maximize his utility.
	% \cite{hajiaghayi2023regret} suggested a repeated principal agent problem without monetary transfer, \ie delegated choice problem, to obtain sublinear Stackelberg regret against a (semi) myopic agent.
	% \cite{dong2018strategic} introduced a repeated version of the strategic classification by~\cite{hardt2016strategic} in which at each round a learner deploys a classifier and an agent arrives with possibly modifying her features to deceive the classifier.
	\cite{gan2019imitative,birmpas2020optimally,zhao2023online} considered a repeated Stackelberg game to study how the follower or leader can exploit the opponent in a general game with arbitrary payoffs.
	Their techniques, however, do not apply to our model as they typically consider the setting of one player knowing the entire payoff matrix trying to deceive the other player given various behavioral assumptions.

	\paragraph{Exploiting no-regret agents.} Several works have considered the extent to which one player can exploit the knowledge that the other player has a strategy with sublinear regret. The goal is often to approach the Stackelberg value, the maximum payoff that the exploiter could get by selecting an action first and allowing the opponent to best-respond. In simultaneous games, \cite{deng2019strategizing} showed that it is possible for the exploiter to get arbitrarily close to their Stackelberg value, assuming knowledge of the other player's payoff function. \cite{haghtalab2022learning} showed that, for certain types of sequential games, an exploiting leader can approach their Stackelberg value in the limit. Our Theorem \ref{thm:response_bounded_regret_Bob} is a similar statement in our setting, but we bound the exploited agent's behavior with an explicit regret guarantee rather than using discounted future payoffs. Moreover, our setting is not captured by the types of games they consider.
	
	\paragraph{Fictitious play.} Fictitious play was introduced in \cite{fictitious_play_brown}. Convergence to Nash equilibria has been shown for zero-sum games \cite{robinson51} and special cases of general-sum games~\cite{Nachbar90,MS96a,MS96b}.

	None of these results directly apply to our setting, but the most relevant is \cite{Berger05}, which covered non-degenerate  $2 \times n$ games (i.e. where every action has a unique best response). Our ``$2 \times \infty$" game is degenerate, as Bob does not have a unique best response to Alice cutting at $m_B$. Few existing works apply fictitious play to settings where the players have continuous action spaces. An example is \cite{Perkins14}, which showed that a variant (stochastic fictitious play) does converge in two-player zero-sum games with continuous action spaces. 
	
	%The rate of convergence of fictitious play has also been studied. 
	\cite{karlin59} conjectured that fictitious play converges at a rate of $O(T^{-1/2})$. \cite{brandt2013rate} found small games where the convergence rate is $O(T^{-1/2})$, but with very large constants in the $O()$.  \cite{Daskalakis14} disproved Karlin's conjecture, showing that there exist  games in which convergence takes place   at a rate of $\Omega(T^{-1/n})$ using adversarial tie-breaking rules. \cite{Panageas23} found more examples of games in which fictitious play  converges exponentially slowly in the number $n$ of actions that each player has. 
	
	\cite{Harris98} showed that fictitious play converges at a rate of $O(T^{-1})$ in $2 \times 2$ zero-sum games. \cite{Abernethy21} considered settings with diagonal payoff matrices and non-adversarial tie-breaking rules and showed convergence rates of $O(T^{-1/2})$. The result in \cite{Abernethy21} does not imply a rate of  convergence in our setting because  requiring the payoff matrix to be diagonal would correspond to Alice only being allowed to cut at $0$ or $1$. This assumption is not as natural in our setting. In fact, if Alice can only cut at $0$ or $1$ the game becomes zero-sum. Furthermore, we allow arbitrary tie-breaking rules.
	
	\paragraph{Strategic experimentation}
	In the general strategic experimentation model, there are $n$ players and $k$ arms. In every round, each player pulls an arm,  receives the resulting reward,  and obtains    feedback about the other player. There are two types of feedback models: perfect monitoring, where the feedback 
	consists of both the choice of the other player \emph{and} their reward; and imperfect monitoring, where the feedback consists of  the choice of the other player \emph{but not} their reward.
	
	%and imperfect monitoring. Perfect monitoring feedback consists of both the choice of the other player \emph{and} their reward. Imperfect monitoring  feedback consists of  the choice of the other player \emph{but not} their reward.
	
	\cite{bolton1999strategic}  study strategic experimentation with two players, two arms,  and perfect monitoring  in continuous time, where one of the arms is ``safe'' and emits steady rewards, while the other arm is ``risky'' and is  governed by a stochastic process. The main effects observed in symmetric equilibria are a free
	rider effect and an encouragement effect, where a player may explore more in order to
	encourage further exploration from others.
	%introduced by~\cite{bolton1999strategic}, there are two players facing a two-armed bandit problem. Precisely, there exists a ``safe'' arm with known outcome, and a ``risky'' arm governed by a stochastic process. Each player selects an arm at each round and they observe each other's action and corresponding reward as well, which is often called as complete feedback model. \cite{bolton1999strategic} characterized a set of stationary Markov equilibria in terms of the free-rider effect.
	\cite{keller2005strategic} considered the same problem for exponential bandits, and obtained a unique symmetric Markov equilibria followed by various asymmetric ones.
	We refer to~\cite{horner2017learning} for a  survey.

	\cite{aoyagi1998mutual,aoyagi2011corrigendum} study strategic experimentation with imperfect monitoring when there are two players and two arms with discrete priors,  showing that  the players eventually settle on the same arm in any equilibrium. This is a version of the agreement theorem by \cite{aumann2016agreeing} for the multi-player multi-armed bandit model,  stating  that rational players cannot agree to disagree.
	\cite{rosenberg2007social,rosenberg2013games}  also study the model with imperfect monitoring, but where  the decision to switch from the risky arm to the safe one is irreversible.
	
	The model with imperfect monitoring has similarities to our setting.
	Our model also involves two players engaged in a strategic repeated game such that each player observes the action of the other player but not their reward. The main difference is that the payoffs of the players are misaligned, the action spaces are orthogonal, and the payoffs are deterministic.
	It is also an interesting open direction for future work to explore what happens if the payoffs in our model are stochastic.
	%In short, strategic experimentation invokes the strategic tension of free-riding the exploration, however, our setup involves players try to strategically learn the other's payoff function to exploit them, while hiding their own payoff function.
	
	\paragraph{Social learning}
	Initiated by~\cite{banerjee1992simple,welch1992sequential}, a long line of literature studied social learning (or herd behavior) to investigate agents who learn over time in a shared environment.
	Unlike the strategic experimentation model, the agents do not strategize against each other, but only observe the past actions and possibly rewards therein.
	The objective is to identify whether the social learning succeeds or fails if the agents receive private signals~\cite{smith2000pathological}, have behavioral biases~\cite{banihashem2023bandit}, or have a certain  network structure~\cite{bala1998learning}.
	Incentivized exploration by~\cite{kremer2014implementing,che2018recommender} considers mechanism design  in a similar problem, to induce the society to behave in a desired manner. % but we do not discuss more details due to the large difference.
	
	% In sequential social learning by~\cite{banerjee1992simple,welch1992sequential,smith2000pathological}, agents observe a private signal of the nature and made their decision based on the previous decisions but not the outcomes therein.
	% Unlike the 
	% Specifically, the agents are rather myopic and do not strategize 
	
	% In this model, the agents do not strategize each other
	
	% Initiated by~\cite{bolton1999strategic} and further developed by~\cite{keller2005strategic}, strategic experimentation examines the behavior of enduring agents who learn by observing the actions and outcomes of others. 
	% In their model, two players face a multi-armed bandit problem with two arms. 
	% Precisely, there exists a "safe" option with fully understood outcomes, and a "risky" option governed by a stochastic process.
	% The players alternately 
	
	% The strategic players try to commit to the optimal arm in hindsight, where they prefer to free-ride on exploration by others.
	
	% The original paper by~\cite{bolton1999strategic} considered a complete feedback model in which each player

	\section{Preliminaries} \label{sec:prelim}
	%\simina{This was part of the model before, make sure it's okay to move all of it here. We may need to state the theorems more informally  in our results since the formal definitions of strategies/histories are here. Another option is to put all of this in the Model.}
	
	In this section we formally define the notation needed for the proofs. All our notation applies to both the sequential and simultaneous settings, unless otherwise stated.
	
	\subsection*{History} Recall $T$ is the number of rounds. For each round $t \in [T]$, 
	\begin{itemize}
		\item let $a_t \in [0,1]$ be Alice's cut at time $t$ and $b_t \in \{L, R\}$ be Bob's choice at time $t$, where $L$ stands for the left piece $[0, a_t]$ and $R$ for the right piece $[a_t, 1]$. 
		\item let  $A_t = (a_1, \ldots, a_t)$ be the history of cuts until the end of round $t$ and $B_t = (b_1, \ldots, b_t)$ the history of choices made by Bob until the end of round $t$.
	\end{itemize}
	A history $H = (A_T, B_T)$ will denote an entire trajectory of play.
	
	\subsection*{Strategies} 
	Let $\mathcal{P}$ be the space of integrable value densities over $[0,1]$. 
	%probability measures over $[0,1]$. 
	A pure strategy for Alice  at time $t$ is a function 
	\[ S_A^t : [0,1]^{t-1} \times \{L, R\}^{t-1} \times \mathcal{P} \times \mathbb{N} \to [0,1],
	\] such that $S_A^t(A_{t-1}, B_{t-1}, v_A, T)$ is the next cut point   made by Alice as a function of the history $A_{t-1}$ of Alice's cuts, the history $B_{t-1}$ of Bob's choices, Alice's valuation $v_A$, and the  horizon $T$.
	
	For Bob, we define pure strategies separately for the sequential and simultaneous settings due to the different feedback that he gets:
	\begin{itemize}
		\item \textit{Sequential setting.}
		%$S_A^t : A_{t-1} \times B_{t-1} \times v_A \to [0,1]$, which maps the history to the next cut point.
		A pure strategy for Bob at time $t$ is a function 
		\[ S_B^t:  [0,1]^{t} \times \{L, R\}^{t-1} \times \mathcal{P} \times \mathbb{N} \to \{L, R\}\,.
		\]
		That  is, Bob observes Alice's cut point  and then responds.
		\item \textit{Simultaneous setting.}  
		%$S_A^t : A_{t-1} \times B_{t-1} \times v_A \to [0,1]$, which maps the history to the next cut point.
		A pure strategy for Bob at time $t$ is a function 
		\[
		S_B^t:  [0,1]^{t-1} \times \{L, R\}^{t-1} \times \mathcal{P} \times \mathbb{N} \to \{L, R\}\,.
		\]
		Thus here Bob chooses $L$/$R$ before observing Alice's cut point at time $t$.
	\end{itemize}
	
	A pure strategy for Alice over the entire time horizon $T$ is denoted $S_A = (S_A^1, \ldots, S_A^T)$ and tells Alice what cut to make at each time $t$. % given her valuation and the history of cuts and choices made so far.
	A pure strategy for Bob over the entire time horizon $T$ is denoted $S_B = (S_B^1, \ldots, S_B^T)$ and tells Bob whether to play $L$/$R$ at each time $t$.
	%\smallskip 
	
	A mixed strategy is a probability distribution over the set of pure strategies.
	\footnote{In fact, this is equivalent to the behavior strategy in which the player assigns a probability distribution given a history, thanks to Kuhn's theorem~\cite{kuhn1950extensive,kuhn11953extensive}. The original version of Kuhn's theorem is restricted to games with finite action space, but can be extended to any action space that is isomorphic to unit interval by~\cite{aumann1961mixed,dresher2016advances}, which contains our setting.}
	
	\subsection*{Rewards and utilities} 
	%The expected utility of a player is computed using the player’s beliefs about the private information of the other player. 
	Suppose Alice has mixed strategy $S_A$ and Bob has mixed strategy $S_B$. 
	Let $u_A^t$ and $u_B^t$ be the random variables for the utility (payoff) experienced by Alice  and Bob, respectively, at round $t$.  The utility of player $i \in \{A,B\}$ is denoted \[
	u_i = u_i(S_A, S_B) =  \sum_{t=1}^{T} u_i^t\,.
	\]
	The utility of player $i$ from round $t_1$ to $t_2$ is $u_i(t_1, t_2) = \sum_{t=t_1}^{t_2} u_i^t$.
	
	The expected utility of player $i$  is 
	$\Ex[u_i] = \sum_{t=1}^{T} \Ex[u_i^t], 
	$
	where the expectation is taken over the randomness of the strategies $S_A$ and $S_B$.
	
	Given a history $H$, let $u_i^t(H)$ be player $i$'s utility in round $t$ under  $H$ and let
	$u_i(H)= \sum_{t=1}^T u_i^t(H)$ be player $i$'s cumulative utility  under $H$.
	
	%NOTE: if We consider an alternative strategy $S_A'$ for Alice, then Alice gets utility $u_i' = u_i(S_A', S_B)$. 
	\subsection*{Midpoints and Stackelberg value} Let $m_A \in [0,1]$ be Alice's midpoint of the cake, with $
	V_A([0,m_A]) =  1 /2$,  and $m_B \in [0,1]$ be Bob's midpoint, with $V_B([0,m_B])  =1 /2$. Since the densities are bounded from below, the midpoint of each player is uniquely defined. 
	
	Alice's Stackelberg value, denoted $u_A^*$, is the utility Alice gets when she cuts at  $m_B$ and Bob chooses his favorite piece breaking ties in favor of Alice (i.e. taking the piece she prefers less).
	
	%Previous version of the above paragraph: Alice's Stackelberg value, denoted $u_A^*$, is the utility that Alice gets when she cuts the cake  at $m_B$ and Bob takes his favorite piece, breaking  ties in favor of Alice.

	\section{Alice exploiting Bob} \label{sec:alice_exploiting_bob_main}
	
	In this section we give an overview of Theorem~\ref{thm:response_bounded_regret_Bob}, which considers the sequential setting and quantifies the extent to which Alice can exploit a Bob that has sub-linear regret  with respect to the benchmark of choosing the best piece in each round. The formal proof of Theorem~\ref{thm:response_bounded_regret_Bob} can be found in Appendix~\ref{app:alice_exploiting_bob}. The proof of Proposition~\ref{prop:myopic-bob} is included in Appendix~\ref{app:alice_exploiting_bob} as well.
	
	\medskip 
	
	We start by defining the  notion of Stackelberg regret~\cite{dong2018strategic,haghtalab2022learning}.
	
	\begin{definition}[Stackelberg regret]
		Given a history $H$, Alice's Stackelberg regret is  
		\begin{align} \label{eq:regret_stackelberg}
			\text{Reg}_{A}(H) = \sum_{t=1}^{T} \left[ u_A^* -  u_A^t(H) \right],
		\end{align}
		recalling that $u_A^*$ is Alice's Stackelberg value and $u_A^t(H)$ is Alice's utility in round $t$ under the \\history $H$.
		%\sscomment{we can skip the latter explanation for now? (same for the definition below)} \simina{See answer in chat; the reader may not remember the notation.}
	\end{definition}
	
	% This definition is equivalent to that of \cite{dong2018strategic}.

	\medskip 
	
	For Bob, we consider the basic notion of static regret, where Bob compares his payoff to what would have happened if Alice's actions remained the same but he chose  the best piece in each round. 
	
	\begin{definition}[Regret] 
		Given a history $H$, Bob's  regret is 
		\[
		\textit{Reg}_B(H) = \sum_{t=1}^{T} \Bigl[ \max 
		\Bigl\{V_B([0, a_t]), V_B([a_t, 1])\Bigr\} - u_B^t(H)\Bigr],
		\]
		recalling that $u_B^t(H)$ is Bob's utility in round $t$ under the history $H$.  
	\end{definition}
	%\sscomment{I'm not sure if the term myopic is needed as the notation regret itself conventionally accounts for the best action(s) in hindsight, given the opponent's (adversary's) action.} \simina{Agreed, changed.}
	
	\paragraph{Regret guarantees.} Suppose Alice uses a mixed strategy $S_A$ and Bob uses a mixed strategy $S_B$. We say that 
	\begin{itemize} 
		\item Alice's strategy $S_A$ ensures Alice's Stackelberg regret is at most $\gamma $ against $S_B$ if 
		$
		\text{Reg}_{A}(H) \leq \gamma 
		$ for all $T$-round histories $H$ that could have arisen under the strategy pair $(S_A, S_B)$.
		\item Bob's strategy $S_B$ ensures Bob's regret is at most $ \gamma $ if $\text{Reg}_B(H) \leq \gamma $ for all $T$-round histories $H$ that could have arisen under the strategy pair $(S_A, S_B)$.
	\end{itemize}
	More broadly, a Bob strategy $S_B$ has regret $\gamma$ if 
	\begin{itemize} 
		\item $\text{Reg}_B(H) \leq \gamma$ for all $T$-round histories $H$ that could have arisen under strategy pairs $(S_A, S_B)$, for all  Alice strategies $S_A$.
	\end{itemize}
	
	\medskip 
	
	Next we provide a proof sketch for Theorem~\ref{thm:response_bounded_regret_Bob}, which is divided in the next two propositions, corresponding to the cases where Alice knows and does not know $\alpha$.
	
	\begin{proposition}  \label{prop:exploiting_Bob_part1}
		Let $\alpha \in [0,1)$. Suppose  Bob plays a strategy that ensures  his regret is   $O(T^{\alpha})$ and let  $\mathcal{B}^{\alpha}$ denote the set of all such Bob strategies. Assume Alice knows $\alpha$. Then she  has a strategy $S_A = S_A(\alpha)$ that ensures her Stackelberg regret   is  ${O}\Bigl(T^{\frac{\alpha+1}{2}} \log T\Bigr)$. This is essentially optimal: Alice's Stackelberg regret is $\Omega\Bigl(T^{\frac{\alpha+1}{2}}\Bigr)$ for some Bob strategy in $\mathcal{B}^{\alpha}$. 
	\end{proposition}
	\begin{proof}[Proof sketch]
		We sketch both the upper and lower bounds.
		
		\paragraph{Sketch for the upper bound.} 
		Let $S_B$ denote Bob's strategy, which guarantees his regret is  $O(T^{\alpha})$. Suppose Alice knows $\alpha$. Then Alice initializes an interval $I = [0,1]$ and uses the next strategy.
		Iteratively, for $i = 0, 1, \ldots, $:
		\begin{enumerate}
			\item  Alice discretizes the interval $I = [u,w]$ in a constant number of  sub-intervals (set to $6$) of equal value to her, by cutting at  points $a_{i,j}$ for $j \in [5]$ such that $u < a_{i,1} < a_{i,2} < \ldots < a_{i,5} < w$. Denote $a_{i,0} = u$ and $a_{i,6} =w$. 
			\begin{figure}[h!] 
				\centering 
				\includegraphics[scale=0.9]{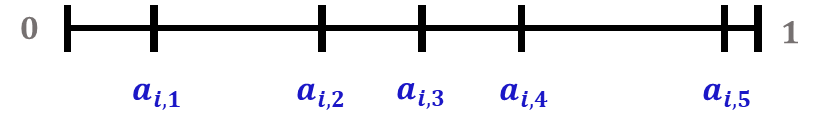}
				\caption{Illustration of step $(1)$ for $i=0$. Alice divides the interval $[0,1]$ in $6$ disjoint intervals that have equal value to her, demarcated by points $a_{0,0} = 0 < a_{0,1} < a_{0,2} < a_{0,3} < a_{0,4} < a_{0,5} < 1  = a_{0,6}$.}
				\label{fig:illustration_Alice_cuts_in_6_pieces_step1}
			\end{figure}
			\item Alice selects a number $\eta$, which will be set ``large enough'' as a function of  $T$ and $\alpha$.  In the next $5 \eta$ rounds, Alice cuts an equal number of times at each  point $a_{i,j}$ for $j \in [5]$. That is: 
			\begin{itemize}
				\item In each of the next $\eta$ rounds, Alice cuts at $a_{i,1}$ and observes Bob's choices there, computing the majority  answer as $c_{i,1} = L$ if Bob picked the left piece more times than the right piece, and $c_{i,1} = R$ otherwise.
				\item The next $\eta$ rounds Alice switches to cutting at $a_{0,2}$, and so on.  
			\end{itemize} 
			
			In this fashion, Alice computes $c_{0,j}$ as Bob's majority answer corresponding to cut point $a_{0,j}$ for all $j \in [5]$. Also, by default  $c_{0,0} = R$ and $c_{0,6} = L$. Step 2 is illustrated in Figure~\ref{fig:illustration_Alice_cuts_in_6_pieces_step2}.
			\begin{figure}[h!] 
				\centering 
				\includegraphics[scale=0.9]{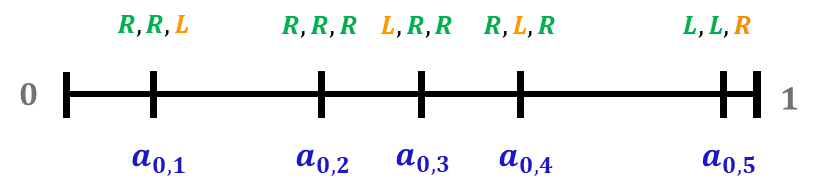}
				\caption{Illustration of step $(2)$ for $i=0$. Suppose $\eta = 3$. Alice cuts $3$ times at each of the points $a_{0,j}$ and observes Bob's choices, which are marked near each such cut point. By default, Alice knows what the answer would be if she cut at $0$ or $1$, so those are set to $R$ and $L$, respectively. The truthful answers (reflecting Bob's favorite piece  according to his actual valuation) are marked with green, while the lying answers are marked with orange.}
				\label{fig:illustration_Alice_cuts_in_6_pieces_step2}
			\end{figure}
			\item The points $a_{0,j}$ for $j \in \{0,\ldots, 6\}$ are arranged on a line  and each is labelled $L$ or $R$, with the leftmost point $a_{0,0} = 0$ labelled $R$ and the rightmost point $a_{0,6} = 1$ labelled $L$. Then there is an index $j \in \{0, \ldots, 5\}$ such that $c_{0,j} = R$ and $c_{0,j+1} = L$. 
			
			Alice computes a smaller interval $I_1$,  essentially consisting of $[a_{0,j}, a_{0, j+1}]$ and some extra space around it to make sure that $I_1$ contains Bob's midpoint:
			\begin{itemize}
				\item If $j \in \{1, \ldots, 4\}$, then set $I_1 = [a_{0, j-1}, a_{0, j+2}]$.
				\item If $j = 0$, then set $I_1 = [a_{0,0}, a_{0,3}]$. 
				\item If $j = 5$, then set $I_1 = [a_{0,3}, a_{0,6}]$.
			\end{itemize}
			\medskip 
			
			Then Alice iterates steps $(1-3)$ on the interval $I_1$.  Step 3 is illustrated in Figure~\ref{fig:illustration_Alice_cuts_in_6_pieces_step3}.
			
			\begin{figure}[h!] 
				\centering 
				\includegraphics[scale=0.9]{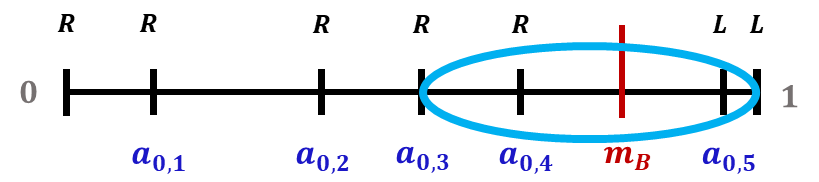}
				\caption{Illustration of step $(3)$ for $i=0$. Alice labels each point $a_{0,j}$ with the majority answer there. Then she identifies the index $j$ such that the point $a_{0,j}$ is labelled  $R$ and the point $a_{0,j+1}$ is labelled $L$. At this point she is assured that either the interval $[a_{0,j}, a_{0, j+1}]$ or one of the adjacent ones contains Bob's midpoint. Alice sets $I_1 = [a_{0,3}, a_{0,6}]$ and recurses on it.}
				\label{fig:illustration_Alice_cuts_in_6_pieces_step3}
			\end{figure}
		\end{enumerate} 
		The full proof explains why the  index $j$ from step $3$ is unique and why it is in fact necessary to include a slightly larger interval than $[a_{i,j}, a_{i,j+1}]$ in the recursion step, due to Bob potentially having lied if his midpoint was very close to a boundary of $[a_{i,j}, a_{i,j+1}]$ but on the other side.
		
		\paragraph{Sketch for the lower bound.} The  lower bound of $\Omega\left(T^{\frac{\alpha+1}{2}}\right)$ relies on the observation that rounds where Alice cuts near $m_B$ and Bob picks his less-preferred piece cost Bob very little but cost Alice a lot. More precisely, suppose $m_A < m_B$ and Alice cuts at $m_B - \varepsilon$. Then compared to his regret bound, Bob loses $\Theta(\varepsilon)$ if he picks the wrong piece. On the other hand, Alice loses $\Theta(m_B - m_A) = \Theta(1)$ compared to her Stackelberg value.
		
		Bob can use this asymmetry by acting as if his midpoint were $\Theta\left(T^{\frac{\alpha-1}{2}}\right)$ closer to $m_A$ than it really is. Lying $\Theta\left(T^{\frac{\alpha+1}{2}}\right)$ times costs Bob only $\Theta(T^{\alpha})$ regret, but costs Alice $\Theta\left(T^{\frac{\alpha+1}{2}}\right)$ regret. To avoid accumulating more regret than this, Bob can afterwards revert to picking his truly preferred piece; the damage to Alice's payoff has already been done.
	\end{proof}

	\begin{proposition} \label{prop:exploiting_Bob_part2}
		Let $\alpha \in [0,1)$. Suppose Bob plays a strategy that ensures  his regret is   $O(T^{\alpha})$. Let $\mathcal{B}^{\alpha}$ denote the set of all such Bob strategies.  
		If Alice does not know $\alpha$, she has a strategy $S_A$ that ensures her Stackelberg regret is $O\left(\frac{T}{\log{T}}\right)$.
		
		\noindent {This is essentially optimal: if $S_A$ guarantees Alice Stackelberg regret %Alice guarantees herself regret 
			$O(T^{\beta})$ against all Bob strategies in $\mathcal{B}^{\alpha}$ for some $\beta \in [0,1)$, then $S_A$ has Stackelberg regret $\Omega(T)$ for some Bob strategy in $\mathcal{B}^{\beta}$.}
	\end{proposition}
	\begin{proof}[Proof sketch]
		Alice's strategy that achieves $O(T/\log T)$ regret follows the same template as her strategy from Proposition \ref{prop:exploiting_Bob_part1}. The only difference is that she sets $\eta$ differently (and much larger) to cover any possible regret bound Bob could have.
		
		The main idea of the lower bound is that, if Alice does not know the value of $\alpha$ in Bob's regret bound, she cannot know when she has true information about Bob's preferences. We exploit this by having a Bob with $O(T^{\beta})$ regret behave exactly like one with $O(T^{\alpha})$ regret but a different midpoint. Then Bob can hide his deception from an Alice with $O(T^{\beta})$ regret because he can tolerate more regret than Alice.
	\end{proof}
	
	Theorem~\ref{thm:response_bounded_regret_Bob} is implied by 
	Proposition~\ref{prop:exploiting_Bob_part1} and Proposition~\ref{prop:exploiting_Bob_part2}.
	
	\medskip 
	
	We briefly remark that the players' value densities must be bounded for Theorem \ref{thm:response_bounded_regret_Bob} to hold; see Remark \ref{no-bound-bob-beat-alice} in Appendix \ref{app:response_to_bounded_regret_bob} for a counterexample with unbounded value densities.
	
	\section{Equitable payoffs} \label{sec:equitable_payoffs_main}
	
	In this section, we give an overview of the proofs for Theorem \ref{thm:Alice_safety_payoffs_summary}, which shows how Alice can enforce equitable payoffs, and Theorem \ref{thm:safety_payoffs_summary}, which shows how Bob can enforce equitable payoffs. The formal proofs can be found in Appendix~\ref{app:safety_payoffs}.
	
	\subsection{Alice enforcing equitable payoffs} \label{sec:alice_enforcing_equitable_payoffs}
	
	We include a formal statement of  Theorem~\ref{thm:Alice_safety_payoffs_summary} next. It contains   the precise constants and  shows that, in fact, the guarantee that Alice can enforce holds at every point in time $t$, as opposed to just at the end.
	
	\begin{customthm}{\ref{thm:Alice_safety_payoffs_summary}}[Alice enforcing equitable payoffs; formal]
		In both the sequential and simultaneous settings, Alice has a pure strategy $S_A$, such that for every  Bob strategy $S_B$:
		\begin{itemize}
			\item   on every trajectory of play, Alice's  average payoff is  at least  $1/2 - o(1)$, while Bob's average payoff is at most $1/2 + o(1)$. More precisely, for all $t \in \{3, \ldots, T\}$:
			\begin{align}
				\frac{u_B(1, t)}{t} \leq \frac{1}{2} + \frac{5\Delta+11}{\ln(2t/5)} \qquad \mbox{and} \qquad
				\frac{u_A(1, t)}{t} \geq \frac{1}{2} - \frac{4}{\sqrt{t-1}}, \notag
			\end{align} 
			recalling that $\Delta$ is the upper bound on the players' value densities. 
		\end{itemize}
		%
		%    In both the sequential and simultaneous settings, there exists a strategy $S_A$ for Alice which ensures that, for every Bob strategy $S_B$, Bob valuation $v_B$, and all $t \geq 3$:
		%
		%   \begin{align}
		%        \frac{u_B(1, t)}{t} \leq \frac{1}{2} + \frac{5\Delta+11}{\ln(2t/5)} \qquad \mbox{and} \qquad
		%       \frac{u_A(1, t)}{t} \geq \frac{1}{2} - \frac{4}{\sqrt{t-1}} \,. \notag
		%    \end{align}
		% That is, payoffs to Bob are upper bounded by $1/2$ on average, while Alice gets at least $1/2$ on average. 
		Moreover, even if Bob's value density is unbounded, his average payoff will still converge to ${1}/{2}$.
	\end{customthm}
	
	\begin{proof}[Proof sketch]
		Alice's strategy is based on Blackwell approachability   \cite{blackwell1956}. The challenge is that Alice has to be prepared for an uncountably infinite variety of Bob's valuation functions, while the original construction in \cite{blackwell1956} works when the number of player types is finite. Another difference is that Alice's action space is also infinite, which turns out to be necessary.
		
		We get around the infinite-Bob issue in two steps. First, Alice's strategy defines a countably infinite set $\overline{\mathcal{V}}$ as a stand-in for the full variety of Bobs. We design $\overline{\mathcal{V}}$ to include arbitrarily good approximations to any valuation function.
		
		Second, we replace Blackwell's original finite-dimensional space with a countably-infinite-dimensional one, where the elements of $\overline{\mathcal{V}}$ are the axes. We define an inner product on this space and use it to adopt Blackwell's argument. Briefly, Alice's strategy tracks the average payoff to each type of Bob in $\overline{\mathcal{V}}$ and defines $\mathcal{S}$ to be the region of the space where all of them have payoffs at most $1/2$. In each round, she constructs a cut point which moves the Bobs' average payoff closer to $\mathcal{S}$, and in the limit traps them in $\mathcal{S}$.
		
		Under this strategy, Alice's payoff guarantee is mostly a byproduct of Bob's. If Bob and Alice have the same value density, then $u_A + u_B = 1$, so bounding Bob's payoff to $1/2$ also bounds Alice's to $1/2$. We achieve the substantially better bound on Alice's payoff by explicitly including her value density $v_A$ in the set  $\overline{\mathcal{V}}$ of Bobs, thus eliminating any approximation error.
	\end{proof}
	
	\subsection{Bob enforcing equitable payoffs} \label{sec:bob_enforcing_equitable_payoffs}
	
	We include a formal statement of  Theorem~\ref{thm:safety_payoffs_summary} next, with the precise constants.
	
	\begin{customthm}{\ref{thm:safety_payoffs_summary}}[Bob enforcing equitable payoffs; formal]
		\phantom{s}
		\phantom{s}
		%Bob has a strategy $S_B$, such that for every Alice strategy $S_A$:
		\begin{itemize}
			\item \emph{In the sequential setting:}  Bob has a pure strategy $S_B$, such that for every Alice strategy $S_A$, on every trajectory of play, Bob's average payoff is at least $1/2 - o(1)$, while Alice's average payoff is at most  $1/2+o(1)$. More precisely,
			\begin{align}
				\frac{u_B}{T} \geq \frac{1}{2} - \frac{1}{\sqrt{T}} \qquad \mbox{and} \qquad  \frac{u_A}{T} \leq \frac{1}{2} + \left(\frac{\Delta}{2\delta}+2\right) \frac{1}{\sqrt{T}}, \notag 
			\end{align}
			recalling that $\delta$ and $\Delta$ are, respectively, the lower and upper bounds on the players' value densities. 
			\item \emph{In the simultaneous setting:}  Bob has a mixed strategy $S_B$, such that for every Alice strategy $S_A$,  both players have average payoff $1/2$ in expectation. % Bob's average payoff is $1/2$ in expectation, while Alice's average payoff is $1/2$ in expectation.
		\end{itemize}  
	\end{customthm}

	\begin{proof}[Proof sketch]
		We cover the simultaneous setting first because it informs the sequential setting.
		\paragraph{Simultaneous setting.} Bob's algorithm is extremely simple: in each round, randomly select $L$ or $R$ with equal probability. The expected payoffs to each player follow  immediately.
		\paragraph{Sequential setting.} This strategy can be thought of as a derandomized version of the simultaneous-setting strategy. The simplest way to derandomize it would be to strictly alternate between $L$ and $R$, but if Bob runs that strategy Alice can easily exploit it.
		
		Instead, Bob mentally partitions the cake into $\sqrt{T}$ intervals $I_1, \ldots, I_{\sqrt{T}}$ of equal value to him. He then treats each interval $I_i$ as a separate cake, alternating between $L$ and $R$ for the rounds Alice cuts in $I_i$. Alice can still exploit this strategy on a single interval $I_i$, but doing so can only give her an average payoff of $1/2 + O(V_A(I_i)) = 1/2 + O(1/\sqrt{T})$. The proof shows that this bound applies for any Alice strategy. 
	\end{proof}
	
	\section{Fictitious play} \label{sec:fictitious_play}
	
	In this section we include a proof sketch of Theorem~\ref{thm_dynamic}, which analyzes the  fictitious play dynamic. The formal proof can be found in Appendix~\ref{app:fictitious_play}.
	
	\medskip 
	
	To define fictitious play, we introduce the  empirical frequency  and  empirical distribution of play:
	\begin{itemize} 
		\item The empirical frequency of Alice’s play
		up to (but not including) time $t$ is: 
		\[
		\phi_A^t(x) = \sum_{\tau = 1}^{t-1} \mathbbm{1}_{\{a_{\tau} = x\}} \; \;  \forall x \in [0,1]\,.
		\] 
		\item The empirical frequency of Bob's play   up to (but not including) time $t$ is:
		\[ \phi_B^t(x) = \sum_{\tau = 1}^{t-1} \mathbbm{1}_{\{b_{\tau} = x\}} \; \;  \forall x \in \{L,R\}\,.
		\]
	\end{itemize}
	The \emph{empirical distribution} of player $i$'s play up
	to (but not including) time $t$ is: $ 
	\mathfrak{p}_i^t(x) = \frac{\phi_i^t(x)}{(t-1)}
	$, where $x \in [0,1]$ for Alice and $x \in \{L,R\}$ for Bob.
	
	\begin{definition}(Fictitious play) \label{def:fictitious_play_dynamic}
		In round $t=1$, each player simultaneously selects an arbitrary action. 
		In every round $t=2,\ldots, T$, each player simultaneously best responds to the empirical distribution of the other player up to time $t$.
		If there are multiple best responses, the player chooses one arbitrarily. 
	\end{definition}

	Our main result in this section is  proving that the average payoffs under fictitious play converge to $(1/2, 1/2)$ and   quantifying the  rate of convergence. The precise statement is included next.
	\begin{customthm}{\ref{thm_dynamic}}
		When both Alice and Bob run fictitious play, regardless of tie-breaking rules, their average payoff will converge to $1/2$ at a rate of $O(1/\sqrt{T})$. Formally:
		\begin{align}
			& \left| \frac{u_A}{T} - \frac{1}{2} \right| \leq \frac{2\sqrt{10}}{\sqrt{T}} \qquad \mbox{and} \qquad  \left| \frac{u_B}{T} - \frac{1}{2} \right| \leq \frac{\sqrt{10}}{\sqrt{T}} \qquad  \forall T  \geq 5\,. \notag
		\end{align}
	\end{customthm}
	
	\begin{proof}[Proof sketch]
		To analyze the fictitious play dynamic, we define for each $t=0, \ldots, T$:
		\begin{itemize}
			\item $\alpha_t = r_t - \ell_t$, where $r_t$ is the number of times Bob picked $R$ up to round $t$ and $\ell_t$ is the number of times he picked $L$
			\item $\beta_t = \sum_{\tau=1}^t \bigl(2V_B([0, a_{\tau}]) - 1\bigr)$.
		\end{itemize}

		\begin{figure}[h!]
			\centering
			\subfigure[The sequence $\alpha_t$.]{
				\includegraphics[scale=0.6]{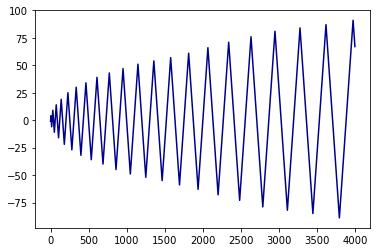}
			}\;  
			\subfigure[The sequence $\beta_t$.]{
				\includegraphics[scale=0.6]{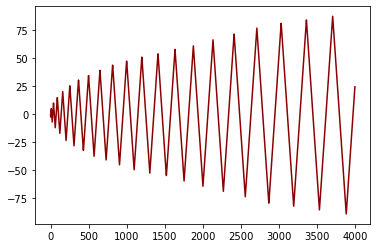}
			}\;
			\subfigure[The sequence $\rho_t$.]{
				\includegraphics[scale=0.6]{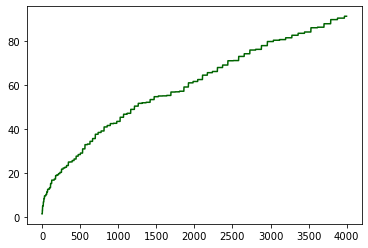}
			}
			\caption{Illustration of the  sequences $\{\alpha_t\}_{t=1}^{\infty}$, $\{\beta_t\}_{t=1}^{\infty}$, and $\{\rho_t\}_{t=1}^{\infty}$ for the instance with trajectories shown in  Figure~\ref{fig:fictitious_play}. The X axis shows the round number $t = 1, \ldots T$ and the Y axis shows the value of the variable being plotted.}
			\label{fig:fictitious_play_variables}
		\end{figure}

		\begin{figure}[h!]
			\centering
			\subfigure[The spiral over the first $100$ rounds.]{
				\includegraphics[scale=0.58]{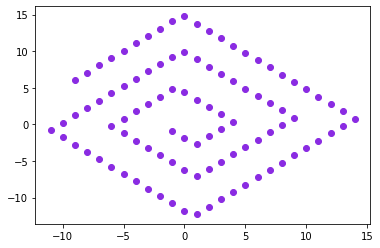}
			}\;  
			\subfigure[The spiral over the first $1000$ rounds.]{
				\includegraphics[scale=0.58]{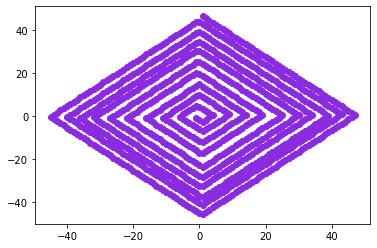}
			}
			\caption{Scatter plot  of the sequence $(\alpha_t, \beta_t)_{t \geq 1}$, illustrating the spiral for the instance with trajectories shown in  Figure~\ref{fig:fictitious_play}, where the sequences $\alpha_t$ and $\beta_t$ are illustrated separately  in Figure~\ref{fig:fictitious_play_variables}.}
			\label{fig:fictitious_play_spiral}
		\end{figure} 
		
		The quantities $\alpha_t$ and $\beta_t$ control what happens under fictitious play: Alice's decision in round $t+1$ is based on $\alpha_t$ and Bob's decision in round $t+1$ is based on $\beta_t$. These decisions in turn affect $\alpha_{t+1}$ and $\beta_{t+1}$, forming a dynamical system. In general, these dynamics result in a counterclockwise spiral through $\alpha$-$\beta$ space. 
		Figure~\ref{fig:fictitious_play_variables} illustrates the sequences $\alpha_t$ and  $\beta_t$ for the instance in  Figure~\ref{fig:fictitious_play}.

		We define  $\rho_t = |\alpha_t| + |\beta_t|$  and  formalize this spiral,  by showing that the sequence $\{\rho\}_{t=0}^{T}$ is non-decreasing and analyzing the change in $(\alpha_t, \beta_t)$ from round to round. Figure~\ref{fig:fictitious_play_variables} illustrates the parameter $\rho_t$ over time. Figure~\ref{fig:fictitious_play_spiral} illustrates the spiral (associated with  the same trajectory as in Figure~\ref{fig:fictitious_play} and \ref{fig:fictitious_play_variables}), where the spiral is visualized as a scatter plot of the sequence $(\alpha_t, \beta_t)_{t \geq 1}$.

		We first use these dynamics to bound Bob's payoff. Bob's payoff can be almost directly read off due to changes in $\beta_t$ closely matching changes in Bob's payoff. Bob's total payoff to round $t$ turns out to be of the order $t/2 \pm \rho_t$, so bounding the rate at which the spiral expands also bounds Bob's payoff.
		
		We then use the dynamics to bound the total payoff to Alice and Bob. Alice can only cut in the interior of the cake when $\alpha_t = 0$, which happens less and less often as the spiral expands. The players' total payoff when Alice cuts at one end of the cake is $1$, so across $T$ rounds we show the sum of cumulative payoffs of the players is of the order $T \pm \Theta(\sqrt{T})$.  
		
		Combining the bound on the total payoff with the bound on Bob's payoff gives a bound for Alice's payoff.
	\end{proof}
	
	\section{Concluding remarks} \label{sec:concluding_remarks_main}
	
	There are several directions for future work. One direction is to consider a wider class of regret benchmarks and understand how the choice of benchmark influences the outcomes reached. Moreover,  what payoff profiles are attained when the players use randomized algorithms  such as exponential weights to update their strategies? Finally, it would also make sense to consider  settings where the cake has both good  and bad parts.
	
	%The literature on strategic experimentation suggests a model where the rewards (i.e. cake) itself may be random. For instance, suppose the cake is the same each day \emph{except} a portion of it may have been  burned in the oven or not, where burning takes place independently with some probability $p$.
	%
	%Finally, what are appropriate generalizations for $n \geq 2$ players in settings such as cake cutting, multiple divisible or indivisible goods/chores? When can one equip the players with strategies  that ensure fair allocations with high social welfare are reached?
	
	%\sscomment{as a reader, I understood  that fictitious play is natural for simultaneous setting, but wondered what would be the parallel for sequential setting. We may need to more talk about what happens in the sequential setting. it might be enough to talk about differences and further challenges?} \simina{Okay, so we can say here that for sequential setting it would make sense to have sort of soft best response / MWU; future work; we couild also say that the sequential one seems to converge too; }
	
	%We note that the fact that Alice and Bob have bounded value density is necessary; see Remark \ref{no-bound-bob-beat-alice} in Appendix \ref{app:response_to_bounded_regret_bob} for a counterexample with unbounded density.
	
	\bibliographystyle{alpha}

	\bibliography{ref}

	\appendix 
	
	%\newpage 
	\section{Appendix: Alice exploiting Bob} \label{app:alice_exploiting_bob}
	
	In this section we present the proofs for Proposition~\ref{prop:myopic-bob}, showing how Alice can exploit a myopic Bob that  chooses his favorite piece in each round, and  for Theorem~\ref{thm:response_bounded_regret_Bob}, where Bob is nearly myopic. 
	
	\subsection{Appendix: Exploiting a Myopic Bob}
	%We present the proof of Proposition~\ref{prop:myopic-bob}.
	
	\begin{customprop}{\ref{prop:myopic-bob}}\label{prop:myopic-bob-2}
		If Bob plays myopically in the sequential setting, then Alice has a strategy that ensures her  Stackelberg regret is $O(\log{T})$.
	\end{customprop}

	\begin{proof}
		%\simina{There is some inconsistency in notation; can we leave $u_t$ type of notation to utilities? See notation on first two pages, let's use it.}
		%\sscomment{fixed}
		We consider an explore-then-commit type of algorithm for Alice.
		In the exploration phase, Alice does binary search to find  Bob's midpoint (within accuracy of ${1}/{T}$). In the commitment (exploitation) phase, Alice repeatedly cuts at Bob's approximate midpoint. This leads to Alice getting nearly her Stackelberg value in nearly every round.
		Figure~\ref{fig:Alice_searching_midpoint_example_with_myopic_Bob}  shows a visualization of a cake instance with Alice and Bob's midpoints, respectively, with Alice's search process.

		\bigskip 
		\bigskip 
		
		\begin{figure}[h!]
			\centering 
			\includegraphics[scale=0.6]{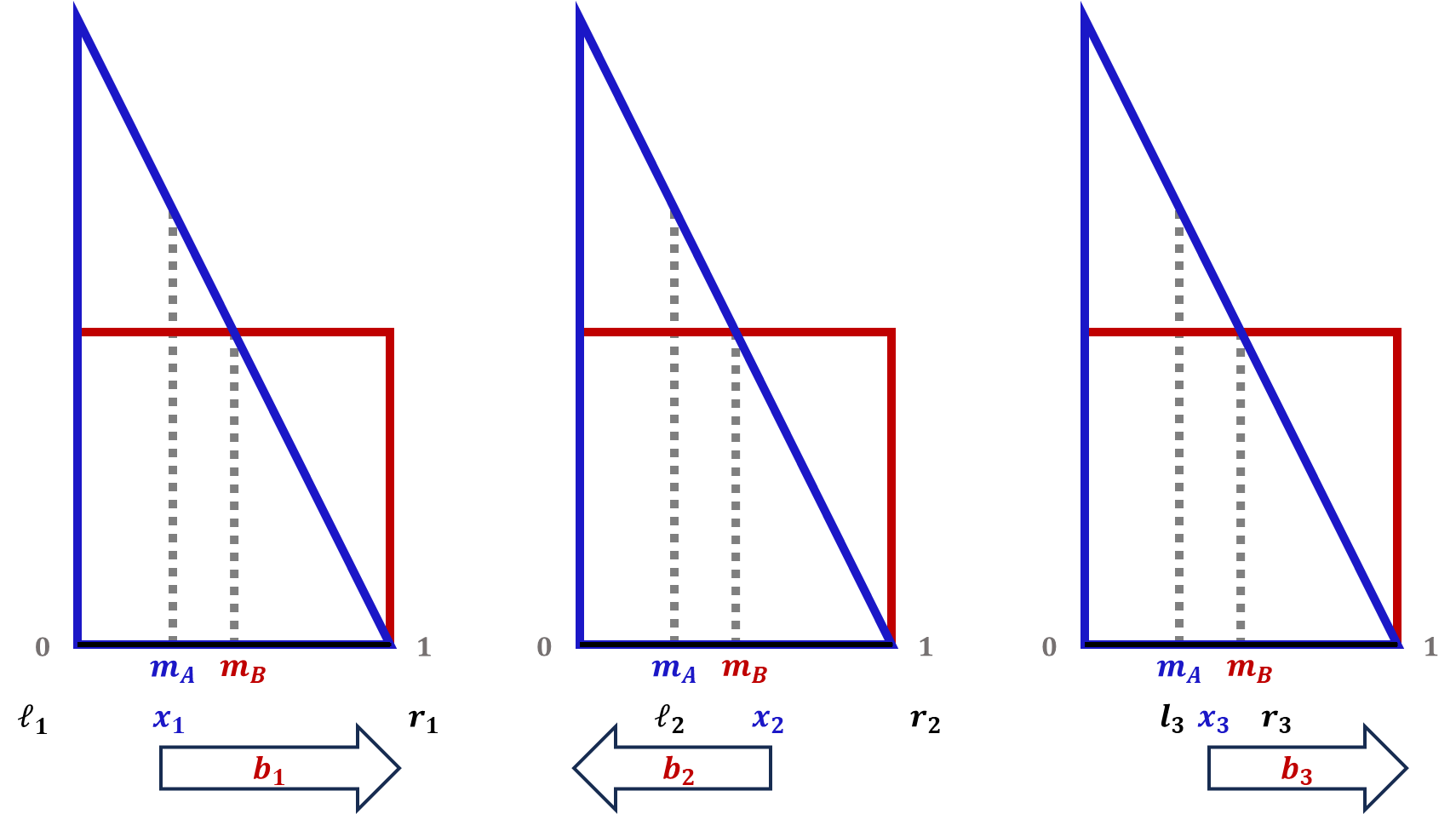}
			\caption{Alice's algorithm against myopic Bob in the exploration phase. Alice's density is shown with blue and her midpoint is $m_A$, while  Bob's density is shown with red and his midpoint is $m_B$. The algorithm initialized $\ell_1 = 0$ and $r_1 =1$ and then re-computes them iteratively depending on Bob's answers. The constructed interval $[\ell_t,r_t]$ shrinks exponentially and becomes closer to $m_B$ as the time $t$ increases.}
			\label{fig:Alice_searching_midpoint_example_with_myopic_Bob}
		\end{figure}

		Alice's algorithm is described precisely in Figure~\ref{alg:alice_strategy_myopic_bob}.

		% \definecolor{mycolor}{rgb}{0.78, 0.80, 0.80}
		\definecolor{mycolor}{rgb}{0.96,0.96,0.96}
		\begin{tcolorbox}[colback=mycolor,colframe=black]
			% \begin{center}
			{\bf Alice's algorithm when Bob is myopic:}
			
			\paragraph{Initialization} Set $\ell_1 = 0,r_1 = 1$ and $\tau = \ln(T)$.
			\paragraph{Exploration} For $t=1,2,\ldots, \tau:$
			\begin{itemize}
				\item Cut at a point $x_t \in [\ell_t, r_t]$ such that $V_A([\ell_t, x_t]) = V_A([x_t, r_t])$. Then observe Bob's action $b_t$.
				\item If $b_t = L$, then set $(\ell_{t+1},r_{t+1}) = (\ell_t, x_t)$.
				\item Else if $b_t = R$, then set $(\ell_{t+1},r_{t+1}) = (x_t, r_t)$.
			\end{itemize}
			\paragraph{Exploitation} For $ t = \tau+1, \ldots, T$:
			\begin{itemize}
				\item If $m_A \leq \ell_\tau$, then cut  at $\ell_\tau - {1}/{T}$.
				\item If $m_A \geq r_\tau$, then cut  at $r_\tau + {1}/{T}$.
			\end{itemize}
			% \end{center}
			\captionof{figure}{Algorithm A.1}\label{alg:alice_strategy_myopic_bob}
		\end{tcolorbox}
		
		We show that Algorithm A.1 in Figure~\ref{alg:alice_strategy_myopic_bob} with $\tau = \Theta(\ln T)$ gives the desired regret bound in several steps.
		\smallskip 
		
		\paragraph{Bob's midpoint lies in the interval $[r_t, \ell_t]$ for all $t$}
		To this end, we first claim that Bob's midpoint lies in $[\ell_t,r_t]$ at each round $t$. We proceed by induction on $t$.
		The base case is $t=1$ clearly holds since $\ell_1 = 0$ and $r_1 = 1$, so $m_B \in [\ell_1, r_1]$.
		Suppose that $m_B \in [\ell_t,r_t]$ for some $t \geq 1$.
		Given that Alice cuts at $x_t \in (\ell_t, r_t)$, if $b_t = L$, this implies that $m_B \in [\ell_t, x_t]$.
		Hence $m_B \in [\ell_t, x_t] = [\ell_{t+1}, r_{t+1}]$.
		This argument also holds  when Bob chooses $R$.
		Thus by induction, we conclude that $m_B \in [\ell_t,r_t]$ for every $t \in [T]$.
		
		\paragraph{Alice's midpoint satisfies $m_A \notin (\ell_t,r_t)$ for every $t \ge 2$}
		Now, during the execution of the algorithm, we will next show that Alice's midpoint $m_A$ satisfies either of $m_A \le \ell_t$ and $m_A \ge r_t$, \ie $m_A \notin (\ell_t,r_t)$ for every $t \ge 2$.
		To see this, recall that in the first round, Alice cuts $x_1 = m_A$.
		If $b_1 = L$, then $\ell_2 = 0$ and $r_2 = m_A$.
		In this case, $[\ell_2,r_2] = [0,m_A]$, so $m_A \notin (\ell_2,r_2)$.
		Afterwards, it still holds since the intervals only shrink, \ie $(\ell_{t+1},r_{t+1}) \subset (\ell_t,r_t)$.
		Similarly, consider the case that $b_1 = R$.
		Then, we have $\ell_2 =m_A$ and $r_2 = 1$.
		Thus $[\ell_2,r_2] = [m_A, 1]$, which implies that $m_A \notin (\ell_2,r_2)$.
		Again since the intervals $(\ell_t,r_t)$ only shrink, we conclude that $m_A \notin (\ell_t,r_t)$ for every $t \ge 2$.
		
		\paragraph{Interval exponentially shrinks}
		We have $V_A([\ell_t, r_t]) = V_A([\ell_{t-1}, r_{t-1}])/2$ for every $t \in [T]$, as we shrink the interval by cutting a point that equalizes Alice's value for both parts within the interval.
		This implies that $V_A([\ell_\tau, r_\tau]) = 2^{-\tau +1}$.
		
		\paragraph{Bounding exploitation phase regret}
		In the exploitation phase, due to the observation above, we have two cases: (i) $m_A \le \ell_\tau$ and (ii) $m_A \ge r_\tau$.
		We will prove that in either case, Alice's single-round regret in the exploitation phase is at most $2^{-\tau+1} + {\Delta}/{T}$.
		\begin{itemize}
			\item For the first case of $m_A \le \ell_\tau$, Alice keeps cutting at $\ell_\tau - {1}/{T}$ for the rest of rounds as per the algorithm's description.
			Then, Bob will myopically choose $R$ and Alice will obtain $V_A([0,\ell_\tau - 1/T])$.
			In this case, we have that $m_A \le m_B$ since $m_B \in [\ell_\tau, r_\tau]$.
			Then, Alice's single-round regret in the exploitation phase is bounded by
			\begin{align*}
				V_A([0,m_B]) - V_A([0,\ell_\tau-1/T]) &=
				V_A([\ell_\tau, m_B]) + V_A([\ell_{\tau}-1/T, \ell_{\tau}]) \\
				&\le
				V_A([\ell_\tau, r_\tau]) + \frac{\Delta}{T} \\
				&= 2^{-\tau +1} + \frac{\Delta}{T}.
			\end{align*}
			\item Otherwise suppose $m_A \ge r_\tau$.
			According to the algorithm, Alice keeps cutting $r_\tau+{1}/{T}$ for all the rest of the rounds, and Bob will respond with $L$. 
			In this case we have $m_A \ge m_B$ since $m_B \in [\ell_\tau, r_\tau]$.
			Similarly, Alice's single-round regret can be upper-bounded by
			\begin{align*}
				V_A([m_B,1]) - V_A([r_\tau+1/T,1])
				&=
				V_A([m_B, r_\tau]) + V_A([r_\tau, r_{\tau} + 1/T]) \\
				&\le
				V_A([\ell_\tau, r_\tau]) + \frac{\Delta}{T} \\
				&= 2^{-\tau +1} + \frac{\Delta}{T}.
			\end{align*}
		\end{itemize}
		Hence in both cases, Alice's single-round regret in the exploitation phase is at most  $2^{-\tau +1} + {\Delta}/{T}$.
		
		\paragraph{Final regret bound}
		Overall, by simply upper-bounding Alice's single-round regret in the exploration phase by $1$, we obtain the following upper bound for the total regret:
		\begin{align*}
			\tau \cdot 1 + (T- \tau) \cdot \left(2^{-\tau + 1} + \frac{\Delta}{T}\right).
		\end{align*}
		Plugging $\tau = \ln(T)$, we obtain the regret bound of $O(\ln T)$, which completes the proof.\footnote{We do not optimize over $\tau$.}
	\end{proof}
	%We remark that if Bob behaves myopically, Alice does not even need to know her own valuation.
	%This can be done by discretizing the entire interval with equal valuation \rpcomment{If Alice doesn't know her own valuation, how does she do this?}, and run standard some no-regret algorithms by considering each interval as an arm in the multi-armed bandit problem.

	\subsection{Appendix: Exploiting a Nearly Myopic Bob} \label{app:response_to_bounded_regret_bob}
	
	In this section we prove Theorem~\ref{thm:response_bounded_regret_Bob}, which explains the payoffs achievable by Alice when Bob has a strategy with sub-linear regret with. We restate it here for reference.
	\begin{customthm}{\ref{thm:response_bounded_regret_Bob}}[Exploiting a nearly myopic Bob] 
		Let $\alpha \in [0,1)$. Suppose Bob plays a strategy that ensures  his regret %---compared to the standard of selecting the optimal piece in every round---
		is %that guarantee his regret  with respect to the benchmark of choosing the best piece in each round is 
		$O(T^{\alpha})$. Let $\mathcal{B}^{\alpha}$ denote the set of all such Bob strategies. %Suppose Bob plays a strategy in $ \mathcal{B}^{\alpha}$.
		%, which ensures his  regret is $O(T^{\alpha})$ with respect to his best pure quantized response  in hindsight, denoted $S_B^{*}$.   
		%a fraction $\lambda \in [0,1]$ of the surplus. If  $\lambda \in [0,1)$,  then Alice has a strategy $S_A$ with negative regret that does not require knowledge of $\alpha$ or $\lambda$. 
		%If $\lambda = 1$:
		%   \begin{description}
		%\item[$\bullet$] If  $\lambda \in [0,1)$,  Alice has a strategy $S_A$ that ensures her regret is negative and does not require knowledge of $\alpha$ or $\lambda$. 
		%\item[$\bullet$ $\lambda = 1$] (her Stackelberg value):  
		\begin{description}
			\item[$\; \;\; \bullet$]  If Alice knows $\alpha$, she has a strategy $S_A = S_A(\alpha)$ that ensures her Stackelberg regret   is ${O}\Bigl(T^{\frac{\alpha+1}{2}} \log T\Bigr)$. Moreover, Alice's Stackelberg regret is  $\Omega\left(T^{\frac{\alpha+1}{2}}\right)$ for some Bob strategy in $\mathcal{B}^{\alpha}$.
			\item[$\; \;\; \bullet$]  If Alice does not know $\alpha$, she has a strategy $S_A$ that ensures her Stackelberg regret is $O\Bigl(\frac{T}{\log{T}}\Bigr)$. %but at least   $\Omega(T/f(T))$  for some Bob strategy in $\mathcal{R}_B^{\alpha}$, where $f$ grows slower than any polynomial. 
			{This is essentially optimal: if $S_A$ guarantees Alice Stackelberg regret %Alice guarantees herself regret 
				$O(T^{\beta})$ against all Bob strategies in $\mathcal{B}^{\alpha}$ for some $\beta \in [0,1)$, then $S_A$ has Stackelberg regret $\Omega(T)$ for some Bob strategy in $\mathcal{B}^{\beta}$.}
		\end{description}
	\end{customthm}
	
	\begin{proof}[Proof of Theorem~\ref{thm:response_bounded_regret_Bob}]
		The known-$\alpha$ upper bound of ${O}\left(T^{\frac{\alpha+1}{2}} \log T\right)$ follows from invoking Proposition~\ref{alice-beat-no-regret-tighter-update} with $f(T) = T^{\alpha}$. The $\Omega\left(T^{\frac{\alpha+1}{2}}\right)$ lower bound is Proposition~\ref{prop:alice-known-alpha-lower-bound}.
		The  lower bound for the case where $\alpha$ is unknown follows by Lemma~\ref{lem:unknown-alpha-alice-regret-lower-bound}. The upper bound follows by Proposition~\ref{alice-beat-no-regret-tighter-update} with $f(T) = \frac{T}{(\log T)^4}$.
	\end{proof}
	
	Both  upper bounds follow the same template, which is captured by the following proposition.
	
	\begin{proposition}
		\label{alice-beat-no-regret-tighter-update}
		%Suppose Alice's value density is bounded above by a constant $D$ and Bob's value density is bounded below by a constant $\delta > 0$. 
		Suppose Bob's strategy has regret $O(f(T))$, for $f(T) \in o\left(\frac{T}{(\log T)^2}\right)$ and $f(T) \geq 1$. If Alice knows $f$, then she has a strategy that guarantees her  Stackelberg regret is $O\left(\sqrt{T \cdot f(T)} \log T\right)$. In particular, if 
		\begin{itemize}
			\item Bob's strategy has regret at most $r f(T)$, for some $r > 0$; and  
			\item $T$ is large enough so that $T >  \exp\left(\frac{4r\Delta}{\delta}\right)$ and $f(T) < \frac{T}{(\ln T)^2}$; 
		\end{itemize}
		then Alice's payoff satisfies:
		\[
		u_A \geq T \cdot u_A^* - \left(\frac{5}{\ln 2} + 6\right) \sqrt{f(T) \cdot T} \ln T\,.
		\]
	\end{proposition}
	
	Before proving the proposition, we present the algorithm that Alice will run  to beat a Bob with a regret guarantee of $O(f(T))$.
	\definecolor{mycolor}{rgb}{0.96,0.96,0.96}
	\begin{tcolorbox}[colback=mycolor,colframe=black]
		% \begin{center}
		{\bf \em Alice's strategy when Bob's regret is at most $r \cdot f(T)$, for some constant $r > 0$:} \\
		
		\textit{Informational assumption: Alice needs to know $f$ and $T$, but not $r$.} %$O(f(T))$:} \\
		
		\paragraph{\bf Initialization} Set $x_0 = 0, y_0 = 1$,  $\eta = \left\lceil \sqrt{f(T) \cdot T}\right\rceil$, and $n = \left\lfloor -\log_2 \left(3\sqrt{f(T)/T} \ln T\right)\right\rfloor$.
		
		\paragraph{\bf Exploration} For  $i = 0, 1, \ldots, n-1$: %\rpcomment{Or ``for the first $5n\eta$ rounds", if that's clearer}
		\begin{itemize}
			\item  {\em \bf Step 1:}  Set $a_{i,0} = x_i$ and $a_{i,6} = y_i$. For  $j \in [5]$, let $a_{i, j}$ be the point with \[ V_A([x_i, a_{i, j}]) = \frac{j}{6} V_A([x_i, y_i])\,.
			\] 
			\item {\em \bf Step 2:} For  $j \in [5]$:
			\begin{itemize} 
				\item  Cut at $a_{i, j}$ for $\eta$ rounds. Define $c_{i, j} = L$ if the majority of Bob's answers were left when the cut point was   $a_{i, j}$, and  $c_{i, j} = R$ otherwise.
			\end{itemize}
			\item {\em \bf Step 3:} 
			\begin{description}
				\item[$\;\;$ \em (3.a):] If $c_{i, j} = L$ $\forall j \in [5]$, then set $x_{i+1} = x_i$ and $y_{i+1} = a_{i, 3}$.
				\item[$\;\;$ \em (3.b):]  If $c_{i, j}=R$  $\forall j \in [5]$, then set $x_{i+1} = a_{i, 3}$ and $y_{i+1} = y_i$.
				\item[$\;\;$ \em (3.c):] Else, there exists a unique $k \in [4]$ such that $c_{i, j} = R$ for all $j \leq k$ and $c_{i, j} = L$ for all $j > k$. Set $x_{i+1} = a_{i, k-1}$ and $y_{i+1} = a_{i, k+2}$.
			\end{description}
		\end{itemize}
		\paragraph{\bf Exploitation} For the rest of the rounds, Alice cuts at $\chi$ based on the following cases:
		\begin{align*}
			\chi = 
			\begin{cases}
				x_n & m_A < x_n \\
				y_n & m_A > y_n \\
				m_A & m_A \in [x_n, y_n].
			\end{cases}    
		\end{align*}
		% \end{center}
		\captionof{figure}{Algorithm A.2}\label{alg:alice_strategy_nearly_myopic_bob}
	\end{tcolorbox}
	
	\begin{proof}[Proof of Proposition~\ref{alice-beat-no-regret-tighter-update}]
		Overall, Alice will use an explore-then-commit style of algorithm:
		\begin{itemize}
			\item     In the exploration phase, Alice will conduct a variant of binary search to locate Bob's midpoint $m_B$ within an accuracy of $O\left(\sqrt{f(T)/T} \log T\right)$.
			\item     In the exploitation phase, Alice will cut near the estimated midpoint for the rest of the rounds.
		\end{itemize}
		The main difficulty that Alice encounters is to precisely locate Bob's midpoint in the exploration phase, since Bob can fool Alice if she cuts sufficiently close to his midpoint.
		We overcome this challenge by having Alice's algorithm stay  far enough  from $m_B$ so that Bob is forced to answer  truthfully most of the time.
		
		\paragraph{Notation.}    Let $w = \sqrt{f(T)/T} \ln T$. Then $n = \lfloor -\log_2 (3w) \rfloor$. Since  $f(T) < \frac{T}{(\ln T)^2}$ by the assumption in the proposition statement, we have $w < 1$ and thereby $n \geq 0$. Also recall the proposition statement assumes that  Bob's strategy  guarantees him a regret of at most $r f(T)$ for some $r > 0$. Moreover, $T$ was chosen such that $T > \exp \left(\frac{4r\Delta}{\delta}\right)$.
		
		Consider the Alice strategy described in Algorithm A.2 (Fig. \ref{alg:alice_strategy_nearly_myopic_bob}). 
		By Lemma~\ref{lm:atmost-one-deviation}, the exploration phase in Alice's strategy is well-defined. 
		
		Next we derive some useful observations and then combine them to upper-bound Alice's regret.

		\paragraph{Useful observations.}  By Lemma~\ref{lm:exploration-construction}, we have $m_B \in [x_n, y_n]$.  Consider the cut point $\chi$ in the exploitation phase.
		We write $\textsc{Intv}[x,y]$ to denote the interval $[x,y]$ if $y \ge x$ and $[y,x]$ if $x > y$.
		% Let us slightly abuse the notation $[\chi, m_B]$ to denote the interval $[m_B, \chi]$ \sscomment{make this notation} in case $\chi > m_B$ (otherwise $[\chi, m_B]$ is well-defined).
		Then, we obtain
		\begin{align}
			V_A(\textsc{Intv}[\chi, m_B]) 
			&\le V_A([x_n,y_n])
			\explain{By definition of $\chi$}
			\\
			&= 2^{-n} \explain{By property 2 of Lemma \ref{lm:exploration-construction}}
			% \\
			% &= 
			% 2^{-\lfloor -\log_2 (3w)\rfloor} 
			% \notag
			\\
			&\leq 
			2^{\log_2 (3w) + 1} 
			\explain{Plugging in $n$ and using $-\lfloor -x\rfloor \le x+1$ }
			\notag
			\\
			&= 
			6\sqrt{\frac{f(T)}{T}} \ln T, 
			\label{ineq:11041541}
		\end{align}
		where the last identity in \eqref{ineq:11041541} holds by definition of $w$.
		
		To upper-bound the number of times that Bob chooses the piece he likes less  in the exploitation phase, we consider the following three cases with respect to $\chi$:
		\begin{description}
			\item[$\;\;$(a)] If $\chi=x_n$ then $m_A < x_n$. Thus $x_n \neq 0$, so $V_B([\chi, m_B]) > r\sqrt{f(T)/T}$ by Lemma~\ref{lm:exploration-construction}. Since Bob's regret is at most  $r f(T)$, it follows that Bob  takes the wrong piece at most $\frac{1}{2}\sqrt{f(T) \cdot T}$ times.
			\item[$\;\;$(b)] If $\chi=y_n$ then $m_A > y_n$. Thus $y_n \neq 1$, so $V_B([m_B, \chi]) > r\sqrt{f(T)/T}$ by Lemma~\ref{lm:exploration-construction}. Since Bob's regret is at most  $r f(T)$, it follows that Bob  takes the wrong piece at most $\frac{1}{2} \sqrt{f(T) \cdot T}$ times.
			\item[$\;\;$(c)] If $\chi = m_A$, then there is no wrong piece because Alice values both equally. Thus this case does not increase the count of incorrect decisions.
		\end{description}
		
		\paragraph{Putting it all together.}   
		In the exploration phase, Alice accumulates regret at most $n \cdot 5\left\lceil \sqrt{f(T) \cdot T} \right\rceil$, since that is the length of the exploration phase.
		In the exploitation phase, the regret comes from two sources:
		\begin{itemize}
			\item The gap between $\chi$ and $m_B$, which is bounded in equation \eqref{ineq:11041541}.
			\item The rounds in the exploitation phase in which  Bob chooses his least favorite piece. There are at most  $\frac{1}{2}\sqrt{f(T) \cdot T}$ such rounds by cases (a-c). Thus Alice's cumulative regret due to these rounds is also at most $\frac{1}{2}\sqrt{f(T) \cdot T}$.
		\end{itemize}
		
		Then Alice's overall regret  is at most:
		\begin{align*}
			n \cdot 5\left\lceil \sqrt{f(T) \cdot T} \right\rceil 
			+ T \cdot &6\sqrt{f(T)/T} \ln T + \frac{1}{2} \sqrt{f(T) \cdot T} 
			\\
			\leq& 
			n \cdot 10\sqrt{f(T) \cdot T} + T \cdot 6\sqrt{f(T)/T} \ln T + \frac{1}{2} \sqrt{f(T) \cdot T} 
			\explain{Since $\lceil x \rceil \le 2x$ $\; \forall x \geq 1$}
			% \\
			% \leq& 
			% \sqrt{f(T) \cdot T} \left(10n + 6 \ln T + \frac{1}{2}\right) 
			% \\
			% \leq& 
			% \sqrt{f(T) \cdot T} \left(-10 \log_2 \left(3 \sqrt{f(T)/T} \ln T\right) + 6 \ln T + \frac{1}{2}\right) 
			% \\
			% =& 
			% \sqrt{f(T) \cdot T} \left(5 \log_2 T + 6 \ln T - 5 \log_2 f(T) - 10 \log_2 \ln T + \frac{1}{2} - 10 \log_2 3\right) 
			% \\
			% \leq& 
			% \sqrt{f(T) \cdot T} \left(5 \log_2 T + 6 \ln T\right) 
			\\
			\le& 
			\left(\frac{5}{\ln 2} + 6\right) \sqrt{f(T) \cdot T} \ln T,
			\explain{Plugging in $n$ and rearranging}
			% \\
		\end{align*}
		which is $O\left(\sqrt{f(T) \cdot T} \ln T\right)$. This completes the proof. 
	\end{proof}

	The following lemma shows that the exploration phase is well-defined.
	\begin{lemma}\label{lm:atmost-one-deviation}
		Alice's strategy from Algorithm A.2 (Fig. \ref{alg:alice_strategy_nearly_myopic_bob}) has the following properties:
		\begin{description}
			\item[$\;\;\;(i)$] If step $(3.c)$ is executed in   Alice's exploration phase, then there is a unique index $k \in [4]$ such that $c_{i,j} = R$   $\forall j \le k$ and $c_{i,j} = L$ $\forall j > k$. 
			\item[$\;\;\;(ii)$] For each $j \in [5]$, define $\tilde{c}_{i,j} = L$ if Bob prefers  $[0, a_{i,j}]$ to $[a_{i,j},1]$ and $\tilde{c}_{i,j} = R$ otherwise. 
			If there exists $j \in [5]$ such that $\tilde{c}_{i,j} \neq c_{i,j}$, then  % for some $j \in [5]$ where $\tilde{c}_{i,j}$ is Bob's genuinely preferred piece, we have 
			$m_B \in (a_{i,j-1}, a_{i,j+1})$. 
			\item[$\;\;\;(iii)$] For all $i \in \{0, \dots, n-1\}$ and $j \in \{0, 1, \ldots, 5\}$, we have  $V_B([a_{i, j}, a_{i, j+1}]) > 2r\sqrt{f(T)/T}$.
		\end{description}
	\end{lemma}
	\begin{proof}%[Proof of Lemma~\ref{lm:atmost-one-deviation}]
		We prove each of the parts $(i-iii)$ required by the lemma. 
		
		\paragraph{Proof of  part  $(iii)$.} Let $j \in \{0, \ldots, 5\}$. 
		Bob's valuation for the interval $[a_{i,j}, a_{i,j+1}]$ can be lower bounded as follows:
		\begin{align}
			V_B([a_{i, j}, a_{i, j+1}]) 
			&\geq 
			\delta \cdot (a_{i, j+1} - a_{i, j}) 
			\explain{Since $v_B(x) \ge \delta \; \forall x \in [0,1]$}
			\\
			&\geq 
			\frac{\delta}{\Delta} V_A([a_{i, j}, a_{i, j+1}])
			\explain{Since $v_A(x) \le \Delta \; \forall x \in [0,1]$}
			\\
			&= 
			\frac{\delta}{\Delta} \cdot \frac{1}{6} V_A([x_i, y_i]) \,. \label{eq:1983732}
			% \\
			% &= 
			% \frac{\delta}{12\Delta} 2^{-i} 
		\end{align} 
		
		By definition, Alice's strategy halves the cake interval considered with each iteration \\ $i \in \{0, \ldots, n-1\}$, that is:  $V_A([x_i, y_i]) = 1/2 \cdot V_A([x_{i-1}, y_{i-1}])$. Thus  $V_A([x_i,y_i]) = 2^{-i}$ and $2^{-i} \ge 2^{-n}$ for all $ i \in \{0, \ldots, n\}$. Combining these observations with inequality \eqref{eq:1983732}, we obtain 
		\begin{align} \label{eq:intermediate_inequality_step_iii}
			V_B([a_{i, j}, a_{i, j+1}])    &\geq 
			2 \cdot \frac{\delta}{12\Delta} 2^{-n}  \,.   
		\end{align}
		Let $w = \sqrt{f(T)/T} \ln T$. Then $n = \lfloor -\log_2 (3w)\rfloor$. We have 
		\begin{align}
			\frac{\delta}{12\Delta} 2^{-n}
			&= 
			\frac{\delta}{12\Delta} 2^{-\lfloor -\log_2 (3w)\rfloor} 
			\explain{By definition of $n$}
			\\
			&\geq 
			\frac{\delta}{12\Delta} 2^{\log_2(3w)} 
			\explain{Since $-\lfloor -x \rfloor \ge x$}
			% \\
			% &= 
			% \frac{\delta}{4\Delta} \cdot \sqrt{\frac{f(T)}{T}} \ln T 
			\\
			&> 
			\frac{\delta}{4\Delta} \cdot \sqrt{\frac{f(T)}{T}} \cdot \frac{4r\Delta}{\delta}
			\explain{By definition of $w$ and since $T > \exp \left(\frac{4r\Delta}{\delta}\right)$}
			\\
			&= 
			r\sqrt{f(T)/T}\,. \label{eq:almost_final_inequality_step_iii}
		\end{align}
		Combining inequalities \eqref{eq:intermediate_inequality_step_iii} and 
		\eqref{eq:almost_final_inequality_step_iii}, we conclude that 
		\begin{align}
			V_B([a_{i, j}, a_{i, j+1}])  > 2 r\sqrt{\frac{f(T)}{T}}.\label{ineq:11041412}
		\end{align}
		This concludes the proof of part  $(iii)$. 
		
		\paragraph{Proof of parts $(i)$ and $(ii)$.}
		
		For each $i=0, \dots, n-1$, we will show that at most one of the majority answers $c_{i,j}$, for $j \in [5]$, is different from Bob's truthful response.
		
		To be precise, recall from the lemma statement that   $\tilde{c}_{i,j} \in \{L,R\}$ is  Bob's truthful response that maximizes his value when Alice cut at $a_{i,j}$. 
		
		Define 
		\begin{align}  \label{eq:set_S_i_in_Lemma_1}
			S_{i} = \Bigl\{j \in [5]: c_{i,j} \neq \tilde{c}_{i,j}\Bigr\}\,.
		\end{align}
		Let $\textsc{Intv}[x,y]$ denote the interval $[x,y]$ if $y \ge x$ and $[y,x]$ if $x > y$. If  $j \in S_i$, it must be the case that Bob picked the wrong piece in at least  $\frac{1}{2}\sqrt{f(T) \cdot T}$ rounds in which the cut was $a_{i,j}$. Then Bob accumulated  at least $V_B(\textsc{Intv}[m_B, a_{i, j}])$ regret in each such round. 
		Let $\ell \in [4]$. We have 
		\begin{align}
			\frac{1}{2}\sqrt{f(T) \cdot T} \sum_{j \in S_i} V_B(\textsc{Intv}[m_B, a_{i, j}]) &\leq r \cdot f(T) \explain{Since Bob's total regret is at most $r \cdot f(T)$} \\
			&< \frac{1}{2}V_B([a_{i, \ell}, a_{i, \ell+1}])\sqrt{f(T) \cdot T}\,. \explain{By \eqref{ineq:11041412}}
		\end{align}
		
		Dividing both sides by $\frac{1}{2}\sqrt{f(T) \cdot T}$ gives:
		\begin{align}
			\label{ineq:si-bob-value-small}
			\sum_{j \in S_i} V_B(\textsc{Intv}[m_B, a_{i, j}]) < V_B([a_{i, \ell}, a_{i, \ell+1}]) \quad \forall \text{$\ell \in [4]$}\,.
		\end{align}
		
		We show that $|S_i| \leq 1$. Suppose towards a contradiction that $|S_i| > 1$, meaning there exist indices $j, \ell \in S_i$ with $j \neq \ell$. Then 
		\[ 
		\textsc{Intv}[a_{i, j}, a_{i, \ell}] \subseteq \left( \textsc{Intv}[m_B, a_{i, j}] \cup \textsc{Intv}[m_B, a_{i, \ell }] \right) \,.
		\]
		This implies that 
		\begin{align}
			V_B(\textsc{Intv}[a_{i, j}, a_{i, \ell}]) \leq V_B(\textsc{Intv}[m_B, a_{i, j}]) +  V_B(\textsc{Intv}[m_B, a_{i, \ell }]) \leq  \sum_{j \in S_i} V_B(\textsc{Intv}[m_B, a_{i, j}]),
		\end{align}
		which contradicts \eqref{ineq:si-bob-value-small}. Thus the assumption was false and $|S_i| \leq 1$.

		% To see why $|S_i| \leq 1$, we first observe the following inequality:
		% \begin{align*}
		%     \left|\left\{j: V_B(\textsc{Intv}[m_B, a_{i,j}]) \le r\sqrt{\frac{f(T)}{T}}\right\}\right| \le 1.
		% \end{align*}
		% This is because by~\eqref{ineq:11041412}, Bob values each interval between $a_{i,j}$ and $a_{i,j+1}$ as more than $2 \cdot r \sqrt{\frac{f(T)}{T}}$, so at most one of the boundaries of the interval that contains $m_B$ would satisfy $V_B(\textsc{Intv}[m_B, a_{i,j}]) \le r\sqrt{\frac{f(T)}{T}}$.
		
		% Therefore, for every other $a_{i, j}$, Bob values the piece between $m_B$ and $a_{i, j}$ as strictly larger than $r\sqrt{f(T)/T}$.
		% Thus, if $|S_{i}| \ge 2$, Bob's regret is strictly larger than
		% \begin{align}
		%     2 \cdot r\sqrt{\frac{f(T)}{T}} \frac{1}{2} \left\lceil \sqrt{f(T) \cdot T} \right\rceil 
		%     \geq 
		%     r f(T),\label{ineq:11041427}
		% \end{align}
		% which concludes that Bob's regret is strictly larger than $rf(T)$.
		% This contradicts the assumption that Bob's strategy has regret at most $rf(T)$.
		% Hence, we conclude that $|S_i| \le 1$.
		
		For any $j \in S_i$, we must have either $m_B \in (a_{i, j-1}, a_{i, j}]$ or $m_B \in [a_{i, j}, a_{i, j+1})$, as otherwise \eqref{ineq:si-bob-value-small} would be violated. Therefore, if $\tilde{c}_{i, j} \neq c_{i, j}$ for some $j \in [5]$, then $m_B \in (a_{i, j-1}, a_{i, j+1})$. This is  part  $(ii)$ required by the lemma. 
		% Furthermore, if $\tilde{c}_{i,j} \neq c_{i,j}$ for some $j \in [5]$, the above argument implies that $m_B \in (a_{i,j-1}, a_{i,j+1})$.
		% This is because $\tilde{c}_{i,j} \neq c_{i,j}$ but if $m_B \in (a_{i,j-1}, a_{i,j+1})$, Bob's regret from playing $\tilde{c}_{i,j}$ is at most $2 \cdot r\sqrt{\frac{f(T)}{T}}$ by~\eqref{ineq:11041412}, which again implies~\eqref{ineq:11041427}.
		
		Finally, we prove part 
		$(i)$. We will show  there exists a unique $k \in [4]$ such that $c_{i,j} = R$  $\forall j \le k$ and $c_{i,j} = L$ $\forall j > k$. The proof considers two cases:
		\begin{itemize}
			\item Case $|S_i| = 0$. Then every $c_{i, j}$ truthfully reflects Bob's preferences: $c_{i, j} = R$ if $a_{i, j} < m_B$, and $c_{i, j} = L$ if $a_{i, j} > m_B$; 
			$c_{i,j} \in \{L,R\}$ if $a_{i,j} = m_B$. An illustration can be seen in Figure~\ref{fig:S_i_zero_illustration_Alice_cuts_in_6_pieces}.
			In all cases, there is a single switch from $R$ to $L$, and so the index $k$ is unique. 
			\begin{figure}[h!] 
				\centering 
				\begin{subfigure}[]{
						\includegraphics[scale=0.7]{Alice_cuts_in_6_pieces_step2.png}
					}
				\end{subfigure}
				\begin{subfigure}[]{
						\includegraphics[scale=0.7]{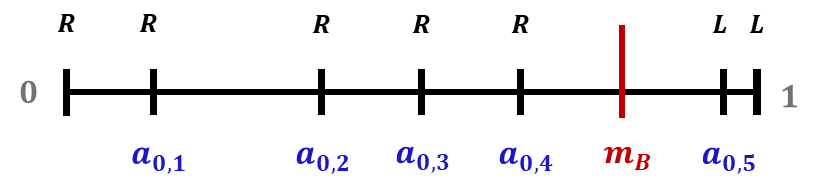}
					}
				\end{subfigure}
				\caption{Illustration of Alice's initial cuts for $i=0$. In this example, she cuts  $3$ times at each of the points $a_{0,j}$ and observes Bob's choices, which are marked with $L$/$R$ near each such cut point. By default, Alice knows what the answer would be if she cut at $0$ or $1$, so those are set to $R$ and $L$, respectively. In Figure (a), the truthful answers (reflecting Bob's favorite piece  according to his actual valuation) are marked with green, while the lying answers are marked with orange. The majority answer at each cut point $a_{0,j}$, denoted $c_{0,j}$, is illustrated in Figure (b). In this example, the majority  answer at each cut point is consistent with Bob's true preference.}
				\label{fig:S_i_zero_illustration_Alice_cuts_in_6_pieces}
			\end{figure}
			%  Regardless of Bob's choices at $m_B$ exactly (if any), there will be a single $k$ that divides $R$ from $L$.
			\item Case $|S_i| = 1$. Then all but one of the $c_{i,j}$'s   truthfully reflect Bob's preferences.  An illustration can be seen in Figure~\ref{fig:S_i_one_illustration_Alice_cuts_in_6_pieces}.
			\begin{figure}[h!] 
				\centering 
				\begin{subfigure}[]{
						\includegraphics[scale=0.7]{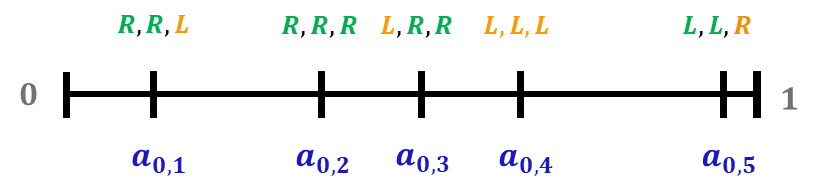}
					}
				\end{subfigure}
				\begin{subfigure}[]{
						\includegraphics[scale=0.7]{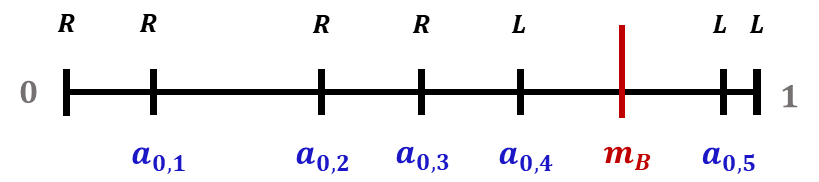}
					}
				\end{subfigure}
				\caption{Illustration of Alice's initial cuts for $i=0$. In this example, she cuts  $3$ times at each of the points $a_{0,j}$ and observes Bob's choices, which are marked with $L$/$R$ near each such cut point. By default, Alice knows what the answer would be if she cut at $0$ or $1$, so those are set to $R$ and $L$, respectively. In Figure (a), the truthful answers (reflecting Bob's favorite piece  according to his actual valuation) are marked with green, while the lying answers are marked with orange. The majority answer at each cut point $a_{0,j}$, denoted $c_{0,j}$, is illustrated in Figure (b). In this example, the majority  answer at each cut point is consistent with Bob's true preference \emph{except} for cut point $a_{0,4}$ where Bob lied every time and so the majority is incorrect as well.}
				\label{fig:S_i_one_illustration_Alice_cuts_in_6_pieces}
			\end{figure}
			\bigskip 
			
			By part $(ii)$ of the lemma, the only exception occurs at an index $j$ with   $m_B \in (a_{i, j-1}, a_{i, j+1})$. But then $c_{i, \ell} = \tilde{c}_{i, \ell} = R$ for $\ell \leq j-1$ and $c_{i, \ell} = \tilde{c}_{i, \ell} = L$ for $\ell \geq j+1$, so regardless of $c_{i, j}$ there will be a single switch from $R$ to $L$. %$a_{i, j}$ is so close to $m_B$ that, even under the $\tilde{c}_{i, j}$, it still divided $R$ from $L$. Changing $c_{i, j}$ only shifts where the dividing $k$ is.
		\end{itemize}
		% Since $|S_{i,j}| \le 1$, it must be either the case that (i) $c_{i,j} = \tilde{c}_{i,j}$ for every $j \in [5]$ or (ii) $c_{i,j} \neq \tilde{c}_{i,j}$ for exactly one $j \in [5]$.
		% In case (i), we can simply pick $k =\tilde{k}$. In case (ii), we must have $m_B \in [a_{i, \tilde{k}}, a_{i, \tilde{k}+1}]$ as otherwise Bob's regret bound would be violated by \eqref{ineq:11041427}. Therefore, we can pick $k = \tilde{k} - 1$ or $\tilde{k} + 1$ depending on the deviating dynamic.
		% In this way, we can argue that there exists a unique $k \in [5]$ such that $c_{i,j} = R$ for $j \le k$ and $c_{i,j} = L$ for $j > k$.
		This concludes that the conditions for $c_{i,j}$ in Step 3.c hold if the conditions in steps 3.a and 3.b do not. This concludes the proof of part $(i)$.
	\end{proof}
	
	The following lemma further reveals several properties of the constructed intervals during the execution of the algorithm. 
	\begin{lemma}\label{lm:exploration-construction}
		In the exploration phase of Algorithm A.2 (Fig. \ref{alg:alice_strategy_nearly_myopic_bob}), Alice constructs a sequence of intervals 
		\[ [x_0,y_0], [x_1,y_1], \ldots, [x_n,y_n]
		\] 
		such that the following properties hold:
		\begin{itemize}
			\item \emph{Property 1:} $x_0 = 0$ and $y_0 = 1$,
			\item \emph{Property 2:} $V_A([x_{i+1}, y_{i+1}]) = \frac{1}{2}V_A([x_i, y_i])$ for $i=0, \dots, n-1$,
			\item  \emph{Property 3:} $m_B \in [x_i, y_i]$, for all $i$,
			\item \emph{Property 4:} If $x_i \neq 0$, then $V_B([x_i, m_B]) > r\sqrt{f(T)/T}$.
			\item \emph{Property 5:} {If $y_i \neq 1$, then $V_B([m_B, y_i]) > r\sqrt{f(T)/T}$.}
		\end{itemize}
	\end{lemma}
	\begin{proof}%[Proof of Lemma~\ref{lm:exploration-construction}]
		Property $1$ holds since  $[x_0, y_0] = [0,1]$ by definition of the algorithm.
		
		Property 2 follows from our choice of $x_{i+1}$ and $y_{i+1}$ always ensuring that $[x_{i+1}, y_{i+1}]$ contains $3$ of the $6$ intervals of equal value the $a_{i, j}$ divide $[x_i, y_i]$ into.
		
		\medskip 
		We will show   Properties 3-5 by induction. 
		%Using Lemma~\ref{lm:atmost-one-deviation}, the proof of the Properties 3-5 of the $[x_i, y_i]$ in the lemma can be shown with induction
		The base case is  $i=0$. Then  $[x_0, y_0] = [0, 1]$. Properties 3-5 are vacuously true for this interval.
		
		Assume that  Properties 3-5 hold for  $i \in \{0,1,\ldots, n-1\}$. 
		For each $j \in [5]$, let $\tilde{c}_{i,j}$ represent Bob's truthful answer when the cut point is $a_{i,j}$. Formally, we have  $\tilde{c}_{i,j} = L$ if Bob prefers  $[0, a_{i,j}]$ to $[a_{i,j},1]$ and $\tilde{c}_{i,j} = R$ otherwise.
		
		We show Properties 3-5 also hold for $i+1$ by considering the next three cases:
		
		\bigskip 
		
		\begin{description}
			\item[\bf Case $c_{i, j} = L$ for all $ j$.] In this case, the majority of Bob's answers is $L$ at each cut point used by Alice. An illustration can be seen in Figure~\ref{fig:lemma_2_all_left}.
			
			\medskip 
			\begin{figure}[h!] 
				\centering 
				\begin{subfigure}[]{
						\includegraphics[scale=0.7]{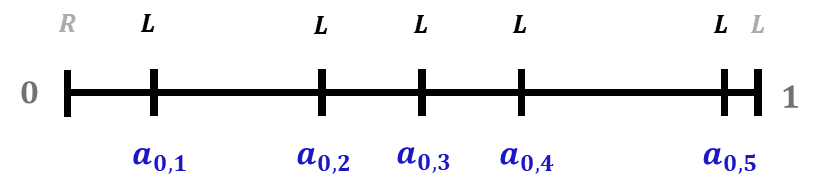}
					}
				\end{subfigure}
				\begin{subfigure}[]{
						\includegraphics[scale=0.7]{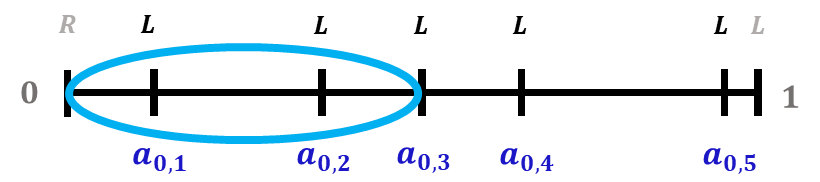}
					}
				\end{subfigure}
				\caption{Illustration of Alice's initial cuts for $i=0$. At  each cut point $a_{0,j}$, the majority of Bob's answers is $L$ (i.e. $c_{0,j} = L$). Then the algorithm recurses in the interval $[0, a_{0,3}]$.}
				\label{fig:lemma_2_all_left}
			\end{figure}
			Then the algorithm recurses in the interval $[x_{i+1}, y_{i+1}]$ given by   $x_{i+1}=x_i$ and $y_{i+1} = a_{i, 3}$.
			By Lemma~\ref{lm:atmost-one-deviation}, we have $\tilde{c}_{i,j} = c_{i,j}$ for each of $j \in \{2,3,4,5\}$, as otherwise Bob's ``true" preferences would alternate between $R$ and $L$ more than once. 
			
			\medskip 
			
			We claim that $m_B \in [x_i, a_{i, 2})$. To see this, consider two cases:
			\begin{description}
				\item[$\; \; \bullet$]  Case $\tilde{c}_{i, 1} = c_{i, 1}$: then $m_B \in [x_i, a_{i, 1}]$ since the majority answers are consistent with Bob's true preference.
				\item[$\; \; \bullet$]   Case $\tilde{c}_{i, 1} \neq c_{i, 1}$: then $m_B \in (x_i, a_{i, 2})$ by Lemma \ref{lm:atmost-one-deviation}.
			\end{description}
			
			%Regardless of whether or not $\tilde{c}_{i, 1} = c_{i, 1}$, we have $m_B \in [x_i, a_{i, 2})$; if $\tilde{c}_{i, 1} \neq c_{i, 1}$ then $m_B \in (x_i, a_{i, 2})$ by Lemma \ref{lm:atmost-one-deviation}, and if $\tilde{c}_{i, 1} = c_{i, 1}$ then $m_B \in [x_i, a_{i, 1}]$.
			Then $m_B \in [x_i, a_{i, 2}) \subset [x_{i+1}, y_{i+1}]$, which proves Property 3 for $i+1$.
			
			\medskip 
			
			If $x_{i+1} = 0$, then Property 4 vacuously follows.
			Otherwise, we have $x_{i+1} = x_i$.
			Then by the inductive hypothesis we obtain $V_B([x_{i+1}, m_B]) = V_B([x_i, m_B]) > r\sqrt{f(T)/T}$. Thus Property 4 holds for $i+1$ as well. 
			
			\medskip

			Since $m_B \in [x_i, a_{i, 2})$, we have 
			\begin{align} 
				V_B([m_B, y_{i+1}]) \geq V_B([a_{i, 2}, a_{i, 3}]) > 2r\sqrt{f(T)/T}, \notag 
			\end{align}
			and so Property 5 holds for $i+1$.  
			
			\bigskip 
			\item[\bf Case $c_{i, j} = R$ for all $ j$.] In this case, the majority of Bob's answers is $R$ at each cut point used by Alice. An illustration can be seen in Figure~\ref{fig:lemma_2_all_right}.
			\begin{figure}[h!] 
				\centering 
				\begin{subfigure}[]{
						\includegraphics[scale=0.7]{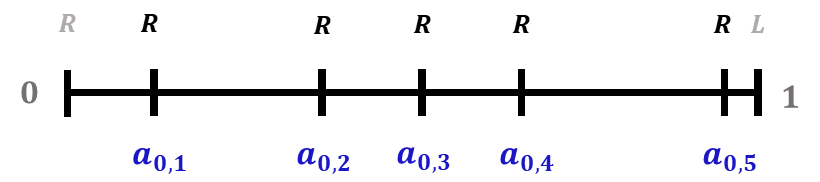}
					}
				\end{subfigure}
				\begin{subfigure}[]{
						\includegraphics[scale=0.7]{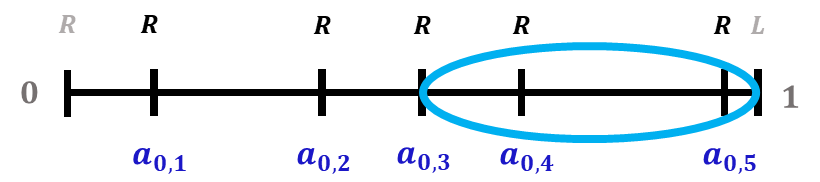}
					}
				\end{subfigure}
				\caption{Illustration of Alice's initial cuts for $i=0$. At  each cut point $a_{0,j}$, the majority of Bob's answers is $R$ (i.e. $c_{0,j} = R$). Then the algorithm recurses in the interval $[a_{0, 3}, 1]$.}
				\label{fig:lemma_2_all_right}
			\end{figure}

			Then the algorithm recurses in the interval $[x_{i+1}, y_{i+1}]$ given by $x_{i+1} = a_{i,3}$ and $y_{i+1} = y_i$. By Lemma~\ref{lm:atmost-one-deviation},  we have $\tilde{c}_{i,j} = c_{i,j}$ for each of $j \in \{1, 2, 3, 4\}$.
			
			\medskip
			
			We claim that $m_B \in (a_{i, 4}, y_i]$. To see this, consider two cases:
			\begin{description}
				\item[$\; \; \bullet$] Case $\tilde{c}_{i, 5} = c_{i, 5}$: then $m_B \in [a_{i, 5}, y_i]$ since the majority answers are consistent with Bob's true preference.
				\item[$\; \; \bullet$] Case $\tilde{c}_{i, 5} \neq c_{i, 5}$: then $m_B \in (a_{i, 4}, y_i)$ by Lemma \ref{lm:atmost-one-deviation}.
			\end{description}
			
			\medskip
			
			Then $m_B \in (a_{i, 4}, y_i] \subset [x_{i+1},y_{i+1}]$, which proves Property 3 for $i+1$.
			
			\medskip
			
			Since $m_B \in (a_{i,4}, y_i]$, we have
			\begin{align}
				V_B([x_{i+1}, m_B]) \geq V_B([a_{i, 3}, a_{i, 4}]) > 2r\sqrt{f(T)/T}, \notag
			\end{align}
			and so Property 4 holds for $i+1$.
			
			\medskip
			
			If $y_{i+1} = 1$, then Property 5 vacuously follows. Otherwise, we have $y_{i+1} = y_i$. Then by the inductive hypothesis we obtain $V_B([m_B, y_{i+1}]) = V_B([m_B, y_i]) > r\sqrt{f(T)/T}$. Thus Property 5 holds for $i+1$ as well.  
			
			\bigskip
			\item[\bf Case where there is a transition from $R$ to $L$ and the last $R$ is $c_{i, k}$ for some $k \in \{1, \ldots, 4\}$.]
			An illustration can be seen in Figure~\ref{fig:lemma_2_crossing}.
			\begin{figure}[h!] 
				\centering 
				\begin{subfigure}[]{
						\includegraphics[scale=0.7]{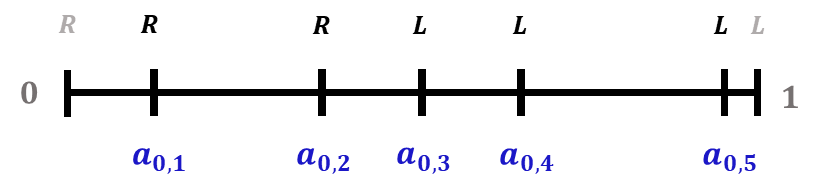}
					}
				\end{subfigure}
				\begin{subfigure}[]{
						\includegraphics[scale=0.7]{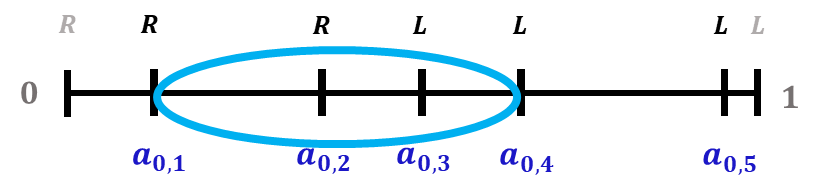}
					}
				\end{subfigure}
				\caption{Illustration of Alice's initial cuts for $i=0$. In this example, there is a transition from $R$ to $L$ in the interval $[a_{0,2}, a_{0,3}]$. Then the algorithm recurses in the interval $[a_{0, 1}, a_{0,4}]$, which is guaranteed to contain Bob's midpoint.}
				\label{fig:lemma_2_crossing}
			\end{figure}
			
			In this case, Alice recurses on the interval $[x_{i+1}, y_{i+1}]$ given by $x_{i+1} = a_{i, k-1}$ and $y_{i+1} = a_{i, k+2}$.

			\bigskip 
			
			If any of the $c_{i, j}$ differ from the $\tilde{c}_{i, j}$, it must be $c_{i, k}$ or $c_{i, k+1}$ because otherwise Bob's true preferences would alternate between $R$ and $L$ more than once, which is impossible. 
			We consider three sub-cases:
			\begin{enumerate}            
				\item \textbf{Case $\tilde{c}_{i, k} \neq c_{i, k}$.} 
				Each time Bob picked his less-preferred piece when Alice cut at $a_{i, k}$, he lost $2V_B(\textsc{Intv}[a_{i, k}, m_B])$ in value compared to his regret benchmark. Since $\tilde{c}_{i, k} \neq c_{i, k}$, he did so at least $\frac{1}{2}\sqrt{f(T) \cdot T}$ times. In order to have regret at most $r f(T)$, he must have
				\begin{align}
					\label{eq:vb-case-k-not-equal-small}
					V_B(\textsc{Intv}[a_{i,k}, m_B]) \le r\sqrt{\frac{f(T)}{T}}.
				\end{align}
				By Lemma \ref{lm:atmost-one-deviation}, we have 
				\begin{align} 
					m_B \in (a_{i, k-1}, a_{i, k+1}) \label{eq:case_a_lie_m_b_inside_a_i_k_minus_1_a_i_k_plus_1}
				\end{align}
				and 
				\begin{align}
					\label{eq:vb-case-k-not-equal-large}
					V_B([a_{i,k-1}, a_{i,k}]) > 2r\sqrt{\frac{f(T)}{T}} \,.
				\end{align}
				Combining equations \eqref{eq:vb-case-k-not-equal-small}, \eqref{eq:case_a_lie_m_b_inside_a_i_k_minus_1_a_i_k_plus_1}, and \eqref{eq:vb-case-k-not-equal-large} we obtain 
				\begin{align}
					V_B([x_{i+1}, m_B]) &= V_B([a_{i, k-1}, m_B]) \explain{Since $x_{i+1} = a_{i,k-1}$} \\
					&\geq V_B([a_{i, k-1}, a_{i, k}]) - V_B(\textsc{Intv}[a_{i, k}, m_B]) \explain{Since $m_B > a_{i, k-1}$ by \eqref{eq:case_a_lie_m_b_inside_a_i_k_minus_1_a_i_k_plus_1}} \\
					&> r\sqrt{f(T)/T} \,. \explain{By \eqref{eq:vb-case-k-not-equal-small} and  \eqref{eq:vb-case-k-not-equal-large}}
				\end{align}
				Thus $V_B([x_{i+1}, m_B]) > r\sqrt{f(T)/T}$, so Property 4 holds for $i+1$.
				
				Since $m_B \in (a_{i, k-1}, a_{i, k+1})$, we have $V_B([m_B, y_{i+1}]) \geq V_B([a_{i, k+1}, a_{i, k+2}]) > 2r\sqrt{f(T)/T}$, proving Property 5 for $i+1$.
				\bigskip 
				
				\item \textbf{Case  $\tilde{c}_{i, k+1} \neq c_{i, k+1}$.} 
				Each time Bob picked his less-preferred piece when Alice cut at $a_{i, k+1}$, he lost $2V_B(\textsc{Intv}[a_{i, k+1}, m_B])$ in value compared to his regret benchmark. Since $\tilde{c}_{i, k+1} \neq c_{i, k+1}$, he did so at least $\frac{1}{2}\sqrt{f(T) \cdot T}$ times. In order to have regret at most $r f(T)$, he must have
				\begin{align}
					\label{eq:vb-case-k-plus-one-not-equal-small}
					V_B(\textsc{Intv}[a_{i,k+1}, m_B]) \le r\sqrt{\frac{f(T)}{T}}.
				\end{align}
				By Lemma \ref{lm:atmost-one-deviation}, we have 
				\begin{align} 
					m_B \in (a_{i, k}, a_{i, k+2}) \label{eq:case_a_lie_m_b_inside_a_i_k_a_i_k_plus_2}
				\end{align}
				and 
				\begin{align}
					\label{eq:vb-case-k-plus-one-not-equal-large}
					V_B([a_{i,k}, a_{i,k+1}]) > 2r\sqrt{\frac{f(T)}{T}} \,.
				\end{align}
				Combining equations \eqref{eq:vb-case-k-plus-one-not-equal-small}, \eqref{eq:case_a_lie_m_b_inside_a_i_k_a_i_k_plus_2}, and \eqref{eq:vb-case-k-plus-one-not-equal-large} we obtain 
				\begin{align}
					V_B([m_B, y_{i+1}]) &= V_B([m_B, a_{i, k+2}]) \explain{Since $y_{i+1} = a_{i,k+2}$} \\
					&\geq V_B([a_{i, k+1}, a_{i, k+2}]) - V_B(\textsc{Intv}[a_{i, k+1}, m_B]) \explain{Since $m_B < a_{i, k+2}$ by \eqref{eq:case_a_lie_m_b_inside_a_i_k_a_i_k_plus_2}} \\
					&> r\sqrt{f(T)/T} \,. \explain{By \eqref{eq:vb-case-k-plus-one-not-equal-small} and  \eqref{eq:vb-case-k-plus-one-not-equal-large}}
				\end{align}
				Thus $V_B([m_B, y_{i+1}]) > r\sqrt{f(T)/T}$, so Property 5 holds for $i+1$.
				
				\medskip 
				
				Since $m_B \in (a_{i, k}, a_{i, k+2})$, we have $ 
				V_B([x_{i+1}, m_B]) \geq V_B([a_{i, k-1}, a_{i, k}]) > 2r\sqrt{f(T)/T}$, which proves         
				Property 4   for $i+1$.
				\bigskip 
				
				\item \textbf{Case $\tilde{c}_{i,j} = c_{i,j}$ for all $j \in [5]$.} Then $m_B \in [a_{i, k}, a_{i, k+1}]$. By Lemma \ref{lm:atmost-one-deviation}, we have 
				\begin{align} \label{eq:prop4_i_plus_1}
					V_B([x_i, m_B]) \geq V_B([x_i, a_{i, k}]) = V_B([a_{i, k-1}, a_{i, k}]) > 2r\sqrt{f(T)/T}
				\end{align} and 
				\begin{align} \label{eq:prop5_i_plus_1}
					V_B([m_B, y_i]) \geq V_B([a_{i, k+1}, y_i]) = V_B([a_{i, k+1}, a_{i, k+2}]) > 2r\sqrt{f(T)/T}\,.
				\end{align} 
				Inequality \eqref{eq:prop4_i_plus_1} implies Property 4 for $i+1$, while inequality \eqref{eq:prop5_i_plus_1} implies Property 5 for $i+1$.
			\end{enumerate}
			
			In all three cases, $m_B \in [a_{i, k-1}, a_{i, k+2}] = [x_{i+1}, y_{i+1}]$, showing Property 3 for $i+1$.
		\end{description}   
		Thus, in all three cases, Properties 3-5 hold for $i+1$.
		By induction, they hold for all $i \in \{0, 1, \dots, n\}$. This completes the proof.
	\end{proof}
	
	\begin{proposition}
		\label{prop:alice-known-alpha-lower-bound}
		Let $v_A$ be an arbitrary Alice density and $\alpha \in [0, 1)$. Then 
		% there exists a Bob with bounded value density function and corresponding strategy that guarantees Bob's regret $O(T^{\alpha})$ such that Alice's regret w.r.t. her Stackelberg value is
		there exists a value density function $\tilde{v}_B = \tilde{v}_B(v_A)$ and strategy $\tilde{S}_B = \tilde{S}_B(v_A, \alpha)$ that ensures Bob's  regret is $O(T^{\alpha})$ while  Alice's Stackelberg regret is 
		$\Omega\left(T^{\frac{\alpha+1}{2}}\right)$.
	\end{proposition}
	% \sscomment{Since this is for negative result, we may first just say that there exists a problem instance for Alice such that sth, and then as a remark say that this indeed holds whenever Alice's value density satisfies sth.
	% Doing so we may not need $m_B - T^{\alpha+1/2} >= 0$ sort of cumbersome argument.} \rpcomment{Yeah, we could do that (say Alice's density is uniform, for instance). It is nice that this lower bound applies to just about any Alice valuation function. It's worth noting that, even if we assume Alice's valuation is uniform, I don't think we can really get rid of the annoying $T^{\frac{\alpha-1}{2}}$'s. The goal is to show $\Omega(T^{\frac{\alpha+1}{2}})$, and that has to come from somewhere.}
	% \sscomment{Oh, I meant if we will use the general version of proof below, we will need to also argue that $m_B - T^{\alpha+1/2} \ge 0$.}
	
	% \sscomment{Value density does not need to be bounded below everywhere, but it seems sufficient to have positive density on three points, $m_a, m_b$ and $m_b-sth?$}\rpcomment{With this construction it is needed in a whole interval. Or, at the very least, we'd need $V_A([m_B, m_B + \varepsilon]) > \delta \cdot \varepsilon$ for some $\delta$ and sufficiently small $\varepsilon$. The reason we need that everywhere in the interval is that the inequalities still need to hold as $T$ gets larger and we start looking at points closer to Bob's midpoint.}
	
	\begin{proof}
		% Let us first consider the case $\alpha =1$.
		% In this case, Bob runs the following simple algorithm.
		% \definecolor{mycolor}{rgb}{0.96,0.96,0.96}
		%     \begin{tcolorbox}[colback=mycolor,colframe=gray!75!black]
		%     % \begin{center}
		%     {\bf Bob's algorithm for $\alpha = 1$:}
		%         \begin{itemize}
		%             \item For every $t \in [T]$, $b_t = L$ with probability $1/2$, and $R$ with probability $1/2$.
		%         \end{itemize}
		%     % \end{center}
		%     \end{tcolorbox}
		% Against such Bob, Alice obtains an expected payoff of $1/2$ per round, which is a constant amount less than her Stackelberg value, assuming that Bob's midpoint does not coincide with Alice's.
		% Therefore, Alice's regret is $\Omega(T) = \Omega\left(T^{\frac{\alpha+1}{2}}\right)$.
		
		At a high level, the Bob we construct will behave as if his midpoint were at $m_B - T^{\frac{\alpha-1}{2}}$.
		If Alice calls his bluff, i.e. 
		cuts  close to $m_B$ ``enough'' times, then Bob reverts to being honest by actually selecting his  preferred piece.
		% \sscomment{could you replace "calls his bluff" by something a bit more formal?}. \rpcomment{This is the intuition part: the "calls his bluff" is formalized in the algorithm below.}
		% This achieves regret $O(T^{\alpha})$ for Bob, since in the worst case he gives up $2T^{\frac{\alpha-1}{2}}$ in value $T^{\frac{\alpha+1}{2}}$ times.
		
		Formally, let $y$ be an arbitrary point such that $y > m_A$. Define the Bob value density $\tilde{v}_B$ as follows:
		\begin{align} \label{eq:v_tilde_B}
			\tilde{v}_B(x) = 
			\begin{cases}
				\frac{1}{2y} & \forall x \in [0,y] %x \leq y 
				\\
				\frac{1}{2(1-y)} & \forall x \in (y, 1] \,. 
			\end{cases}
		\end{align}
		
		Bob's value density $\tilde{v}_B$ is  bounded since $y$ is a fixed constant. 
		Moreover, Bob's midpoint $m_B$ is exactly at $y$. 
		Let $\tilde{S}_B$ be Bob's strategy as defined in Figure~\ref{alg:Bob_strategy_S_B_tilde}.
		
		\definecolor{mycolor}{rgb}{0.96,0.96,0.96}
		\begin{tcolorbox}[colback=mycolor,colframe=black]
			%{\bf Bob's value density for $\alpha < 1$:}
			%\begin{align*}
			%   v_B(x) = 
			%       \begin{cases}
			%           \frac{1}{2y} & x \leq y 
			%           \\
			%           \frac{1}{2(1-y)} & x > y
			%       \end{cases}
			%\end{align*}
			{\bf Bob strategy $\tilde{S}_B$:} \\ 
			
			$\;$   \textbf{Input:} $\alpha$.
			\smallskip

			%Define interval $P = [m_B - T^{\frac{\alpha-1}{2}}, m_B]$, and 
			$\;$  Initialize $c = 0$.  For $t \in [T]$:
			\begin{itemize}
				\item If $m_B < a_t \leq 1$ then play $b_t = L$.
				\item If $0 \leq a_t < m_B-T^{\frac{\alpha-1}{2}}$ then play $b_t = R$.
				\item If $m_B - T^{\frac{\alpha-1}{2}} \leq a_t \leq  m_B$ then: %$(1)$ If $c \ge T^{\frac{\alpha+1}{2}}$, pick %$b_t = R$; $(2)$ Else, pick $b_t = L$.
				
				$\; \; $ If $c \ge T^{\frac{\alpha+1}{2}}$ then play $b_t = R$; 
				Else, play $b_t = L$. 
				
				$\; \; $  Update $c \gets c+1$.
			\end{itemize}
			\captionof{figure}{Algorithm A.3}\label{alg:Bob_strategy_S_B_tilde}
		\end{tcolorbox}
		Strategy $\tilde{S}_B$ ensures that Bob takes his less-favorite piece at most $T^{\frac{\alpha+1}{2}}$ times, and this only happens when Alice cuts in the interval 
		\begin{align} 
			P = [m_B - T^{\frac{\alpha-1}{2}}, m_B]\,.
		\end{align} Since the density $\tilde{v}_B$ is  bounded from above by a constant, Bob gives up a value of at most $O(T^{\frac{\alpha-1}{2}})$ in each such round. Thus Bob's regret is $O(T^{\alpha})$.

		Now let us compute Alice's regret with respect to her Stackelberg value.
		Suppose that $T$ is sufficiently large, so that $m_B-T^{\frac{\alpha-1}{2}} > m_A$. There are two cases depending on whether Alice triggers the switch in Bob's strategy or not.
		%\simina{Reminder: we have an  $\eta$ notation; not sure if we can use it in the proof, it was useful in the proof sketch so may be good to mention it somewhere.}
		\begin{itemize}
			\item \emph{Alice cuts in $P$ at least $T^{\frac{\alpha+1}{2}}$ times.} Consider the first $T^{\frac{\alpha+1}{2}}$ times Alice cuts in $P$. Whenever Alice cuts at some point $x \in P$, her utility will be  $V_A([x, 1])$ since Bob plays $L$. Therefore, her payoff in that round is maximized when $x$ is minimized, which occurs at the left-hand endpoint $x = m_B - T^{\frac{\alpha-1}{2}}$. However,
			\begin{align} 
				m_B - T^{\frac{\alpha-1}{2}} > m_A, \notag 
			\end{align}
			and so 
			\begin{align} V_A\left(\left[m_B - T^{\frac{\alpha-1}{2}}, 1\right]\right) < {1}/{2}\,.
			\end{align} 
			Then her Stackelberg regret over all the rounds in which she cuts in $P$ is at least 
			\begin{align*}
				T^{\frac{\alpha+1}{2}} \cdot \left(V_A([0, m_B]) - V_A\left(\left[m_B - T^{\frac{\alpha-1}{2}}, 1\right]\right)\right) &> T^{\frac{\alpha+1}{2}} \cdot \left(V_A([0, m_A]) + V_A([m_A, m_B]) - \frac{1}{2}\right) \\
				&= T^{\frac{\alpha+1}{2}} V_A([m_A, m_B]) \in \Omega\left(T^{\frac{\alpha+1}{2}}\right)\,.
			\end{align*}
			%which is $\Omega\left(T^{\frac{\alpha+1}{2}}\right)$.
			\item \emph{Alice cuts in $P$ fewer than $T^{\frac{\alpha+1}{2}}$ times}. In this case, we will show that Alice's payoff per round cannot be more than $V_A\left(\left[0, m_B - T^{\frac{\alpha-1}{2}}\right]\right)$. To see this, we consider two sub-cases. 
			
			If Alice cuts at a point $x < m_B - T^{\frac{\alpha-1}{2}}$, then Bob picks R and Alice's utility in that round is 
			\begin{align}  \label{eq:Alice_cuts_at_x_left_of_m_B_minus_epsilon}
				V_A([0, x]) < V_A\left(\left[0, m_B - T^{\frac{\alpha-1}{2}}\right]\right)\,.
			\end{align} 
			If Alice cuts at a point $x \geq m_B - T^{\frac{\alpha-1}{2}}$, then Bob picks L and Alice's utility in that round is
			\begin{align} \label{eq:Alice_cuts_at_x_right_of_m_B_minus_epsilon}
				V_A([x, 1]) < \frac{1}{2} < V_A\left(\left[0, m_B - T^{\frac{\alpha-1}{2}}\right]\right)\,.
			\end{align} 
			Combining \eqref{eq:Alice_cuts_at_x_left_of_m_B_minus_epsilon} and \eqref{eq:Alice_cuts_at_x_right_of_m_B_minus_epsilon}, we obtain that  in each round $t \in [T]$, Alice's utility is 
			\begin{align} 
				u_A^t \leq V_A\left(\left[0, m_B - T^{\frac{\alpha-1}{2}}\right]\right)\,.
			\end{align}
			
			Summing over all rounds $t \in [T]$, we obtain that Alice's  Stackelberg  regret is 
			\begin{align}
				\sum_{t=1}^{T} \left(u_A^* - u_A^t\right) &\geq  
				\sum_{t=1}^T \left(V_A([0, m_B]) - V_A\left(\left[0, m_B - T^{\frac{\alpha-1}{2}}\right]\right)\right) \\
				&= \sum_{t=1}^T V_A\left(\left[m_B - T^{\frac{\alpha-1}{2}}, m_B\right]\right) \\
				&\geq \sum_{t=1}^T \delta \cdot T^{\frac{\alpha-1}{2}} \in \Omega\left(T^{\frac{\alpha+1}{2}}\right)\,.
			\end{align}
			% By summing over every round, her total regret is $\Omega\left(T^{\frac{\alpha+1}{2}}\right)$.
		\end{itemize}
		Thus in both cases, Alice's Stackelberg regret is at least $\Omega\left(T^{\frac{\alpha+1}{2}}\right)$, which concludes the proof.
	\end{proof}
	
	\begin{lemma}
		\label{lem:unknown-alpha-alice-regret-lower-bound}
		Let $\alpha, \beta \in [0,1)$. Suppose Alice's density is $v_A$ and her strategy is $S_A$.
		%Then there exist Bob types $Bob_i = (v_{B,i}, S_{B,i}) \in \mathcal{T}_B$ for $i \in [2]$, such that if Alice's strategy $S_A$ has regret $O(T^{\alpha})$ against $Bob_1$, then her regret is at least $T/6$ against $Bob_2$. \\
		%If Alice's strategy $S_A$ has regret $O(T^{\alpha})$ against some  $Bob_1 \in \mathcal{T}_B$, then her regret is at least $T/6$ against another $Bob_2 \in \mathcal{T}_B$, where $Bob_2$ depends on $v_A$.
		%\\
		There exists a Bob $\bob_1 = (v_{B,1}, S_{B,1})$ that depends on $v_A$ and a Bob $\bob_2 =  (v_{B,2}, S_{B,2})$ that depends on $v_A$ and $S_A$, such that 
		\begin{itemize} 
			\item $\bob_1$ has regret  $O(T^{\alpha})$  (i.e. strategy $S_{B,1}$ ensures that a player with density $v_{B,1}$ has regret $O(T^{\alpha})$)
			\item $\bob_2$ has regret  $O(T^{\beta})$; and   
			\item if $S_A$ ensures Alice Stackelberg regret of $O(T^{\beta})$ against $\bob_1$, then $S_A$ has regret at least ${T}/{6}$ against $\bob_2 $.
		\end{itemize}
		% There exists  $Bob_1 = Bob_1(v_A) \in \mathcal{T}_B$  such that, if Alice's strategy $S_A$ has Stackelberg regret $O(T^{\alpha})$ against $\bob_1$, then $S_A$ has regret at least $T/6$ against another $\bob_2 = \bob_2(v_A, \alpha) \in \mathcal{T}_B$.
		%there exists $Bob_2 \in \mathcal{T}_B$ where Alice's regret from playing $S_A$ against $Bob_2$ is at least $T/6$.\\
		%
		%If Alice uses a strategy $S_A$ that has regret $O(T^{\alpha})$ against a Bob   $(v_{B,1}, S_{B,2}) \in \mathcal{T}$, then 
		%
		%There exists a type of Bob, denoted $Bob_1 $, with density  
		%${v}_{B,1}$ and strategy $S_{B,1}$,  such that  if Alice uses a strategy with $O(T^{\alpha})$ regret against $Bob_1$, then her regret is $\Omega(T)$ against another.
	\end{lemma}
	\begin{proof}
		Let $x$ be the cake position such that $V_A([0, x]) = 2/3$, and let $y$ be the cake position such that $V_A([0, y]) = 5/6$. 
		The first Bob, $\bob_1$, will have a valuation function $v_{B, 1}$ that has midpoint $x$. 
		We define $\bob_1$'s strategy $S_{B, 1}$ so that it  truthfully picks his preferred piece, \ie
		\begin{align} \label{eq:definition_S_B_1}
			S_{B, 1}(A_t, B_{t-1}) = \begin{cases}
				L & \text{if} \;  a_t \in (x, 1] \\
				R & \text{if} \;  a_t \in [0, x]\,.
			\end{cases}
		\end{align}
		Against $\bob_1$, Alice's Stackelberg value is $2/3$. Suppose $S_A$ ensures Alice Stackelberg regret at most $r T^{\beta}$ against $\bob_1$ for some $r > 0$. 
		
		Then we define our second Bob, denoted  $\bob_2$, having a valuation function $v_{B, 2}$ which has a midpoint at $y$.
		Let $k(t)$ be the number of times Alice cuts in the interval $(x, y]$ in rounds $\{1, \ldots, t\}$. 
		Then $\bob_2$'s strategy will be defined as follows:
		\begin{align} \label{eq:definition_S_B_2}
			S_{B,2}(A_t, B_{t-1}) = \begin{cases}
				S_{B, 1}(A_t, B_{t-1}) & \text{if} \; k(t) \leq 3rT^{\beta} \\
				L &  \text{if} \;  k(t) > 3rT^{\beta} \; \mbox{and} \; a_t \in (y, 1] \\
				R &  \text{if} \;  k(t) >3rT^{\beta} \; \mbox{and} \; a_t \in [0, y]\,.
			\end{cases}
		\end{align}
		Intuitively, $S_{B, 2}$ switches to being honest after Alice cuts in $(x, y]$ sufficiently many times. 
		This transition gives $\bob_2$ a regret guarantee of $O(T^{\beta})$. 
		
		When Alice plays $S_A$ against $\bob_1$, her total payoff is 
		\begin{align} \label{eq:S_A_against_Bob1_ub}
			u_A(S_A, S_{B,1}) \geq {2}T/3 - rT^{\beta}\,.
		\end{align} 
		However, every round she cuts in $(x, y]$, her payoff is less than $1/3$. Therefore, against $\bob_1$, her total payoff can also be bounded by 
		\begin{align} \label{eq:S_A_against_Bob1_lb}
			u_A(S_A, S_{B,1}) \leq {2}T/{3} - k(T)/3\,. 
		\end{align}
		Combining inequalities \eqref{eq:S_A_against_Bob1_ub} and \eqref{eq:S_A_against_Bob1_lb}, we obtain 
		\begin{align} \label{eq:k_T_bound_under_S_A}
			k(T) &\leq 3rT^{\beta}.
		\end{align}
		By definition, strategy $S_{B,2}$ behaves the same as strategy $S_{B,1}$ when $k(T) \leq 3rT^{\beta}$. By \eqref{eq:k_T_bound_under_S_A}, we have  $k(T) \leq 3rT^{\beta}$ when Alice uses strategy  $S_A$. 
		Thus, if Alice uses strategy $S_A$, then $\bob_1$ and $\bob_2$ behave exactly the same way. 
		Therefore, Alice receives no more than $2/3$ per round against $\bob_2$, so her regret is at least $T/6 \in \Omega(T)$.
	\end{proof}
	
	If Bob's density is not lower bounded by any constant, then Theorem~\ref{thm:response_bounded_regret_Bob} fails as  shown in the next remark.
	
	\begin{remark}
		\label{no-bound-bob-beat-alice}
		Let $\alpha \in (0,1)$. 
		There exist value densities $v_A$ and $v_B$, where $v_A(x) \in [\delta, \Delta]$  $\forall x \in [0,1]$ and $v_B(x) \in (0, \Delta]$ $\forall x \in [0,1]$ such that Bob has a strategy $S_B$ which  
		guarantees his  regret is at most $T^{\alpha}$ and  Alice's Stackelberg  regret is at least $\Omega(T/\log T)$, no matter what strategy she uses.
	\end{remark}
	
	\begin{proof}
		The specific valuations will be defined in terms of cumulative valuations. Let Alice's valuation be:
		\[
		V_A([0, x]) = \begin{cases}
			\frac{1}{2}x & \text{if} \; x \in [0, 1/2] \\
			\frac{3}{2}x - \frac{1}{2} & \text{if} \; x \in (1/2, 1]
		\end{cases}
		\]
		Let Bob's valuation be:
		\begin{align}
			V_B([0, x]) = 
			\begin{cases}
				x & \text{if} \; x \in [0, 1/2] \\
				\frac{1}{2} + 2^{-\frac{1}{2x-1}} & \text{if} \; x \in (1/2, 1]
			\end{cases}\label{eq:02121107}    
		\end{align}
		Intuitively, Alice has a well-behaved piecewise uniform density with midpoint $m_A = {2}/{3}$. Bob's density is well-behaved for $x \leq {1}/{2}$, but the density rapidly approaches zero just to the right of $x={1}/{2}$.
		
		Fix an arbitrary $\alpha \in (0, 1)$. There exists a point $y$ on the cake such that $V_B([1/2, y]) = \frac{1}{2}T^{\alpha - 1}$. Define Bob's strategy $S_B$ as follows:
		\[
		S_B(A_{t}, B_{t-1}, v_B, T) = \begin{cases}
			R & \text{if} \; a_t < y \\
			L & \text{if} \; a_t \geq y\,.
		\end{cases}
		\]
		This strategy would be honest if Bob's midpoint were at $y$ rather than $1/2$. That means it only differs from his preferred piece when $a_t \in (1/2, y]$. The worst outcome for Bob occurs when  $a_t = y$. But by construction, even in a round where  Alice cuts at $y$, Bob only loses $T^{\alpha - 1}$ utility compared to picking his preferred piece. Since there are $T$ rounds, his overall regret is at most $T^{\alpha}$.
		
		Alice, on the other hand, cannot do very well compared to her Stackelberg value. Her Stackelberg value is $3/4$, achieved by cutting at Bob's midpoint of $1/2$ and receiving her preferred piece. But Bob prevents this payoff by pretending his midpoint is at $y$. 
		To obtain the exact location of $y$, note that
		\begin{align*}
			V_B([1/2, y]) = \frac{1}{2}T^{\alpha - 1}.
			%-\frac{1}{2y-1} &= -1 + (\alpha - 1) \log_2 T \\
			%2y-1 &= \frac{1}{(1-\alpha) \log_2 T + 1} \\
		\end{align*}
		Plugging~\eqref{eq:02121107}, this is equivalent to
		\begin{align*}
			2^{-\frac{1}{2y-1}} = \frac{1}{2} T^{\alpha - 1}.
		\end{align*}
		Solving for $y$, we have
		\begin{align*}
			y = \frac{1}{2} + \frac{1}{(2-2\alpha) \log_2 T + 2}\,.
		\end{align*}
		% \begin{align*}
		%     y &= \frac{1}{2} + \frac{1}{(2-2\alpha) \log_2 T + 2}
		% \end{align*}
		
		For sufficiently large $T$, we have $y \in  (m_B,m_A)$. Alice's best cut location in each round is then at $y$ itself, which gives her a per-round payoff of
		\[
		V_A([y, 1]) = \frac{3}{4} - \frac{3}{(4-4\alpha) \log_2 T + 4}\,.
		\]
		Adding this up over all $T$ rounds gives a total regret of at least, we obtain
		\[
		\frac{3T}{(4-4\alpha) \log_2 T + 4},
		\]
		which is $\Omega(T/ \log T)$ as required. This completes the proof. 
		% \sscomment{Come to read the proof, can't we simplify such that: define Bob's valuation to be again piecewise uniform density, but having two non differentiable points at $1/2$ and $y$ such that $V_B([0,y]) = 1-1/2(1-T^{\alpha -1})$, and then Bob will still do the same strategy while Alice's regret is $\Omega(T)$.} \rpcomment{That would make the algebra much cleaner, but it's a bit awkward for Bob's valuation function to depend on $T$. Ideally it would be a single valuation function which covers all $T$ and $\alpha$, which this does. However, getting $\Omega(T)$ regret for Alice would be nice, since this argument only manages $\Omega(T / \log T)$.}   
		% \sscomment{Can we further obtain tightness on the conditions for $\delta, D$ obtained in Theorem 2.7? (such that $D/\delta < \ln T / (4r)$, i.e., assuming that we will loosen the fact that both are constants). As a naive example, suppose $\delta$ is fixed. We can parameterize the quantity $y$ defined as above, and also the value $V_B([0,y])$ both as  functions of $\alpha, \delta$ resp. If Bob's linear function located at the middle is too steep, against $T^\alpha$ Bob, Alice cannot obtain $T^{(1+\alpha)/2}$, then we may inversely track the value of $\delta / D$.}
	\end{proof}
	
	%\newpage 
	
	\section{Appendix: Equitable payoffs} \label{app:safety_payoffs}
	
	This appendix has two parts. In Appendix \ref{app:Alice_enforcing_safety_payoffs}, we prove Theorem \ref{thm:Alice_safety_payoffs_summary}, which shows that Alice can enforce equitable payoffs. In Appendix \ref{app:Bob_enforcing_safety_payoffs}, we prove Theorem \ref{thm:safety_payoffs_summary}, which shows that Bob can enforce equitable payoffs.

	\subsection{Appendix: Alice enforcing equitable payoffs}
	\label{app:Alice_enforcing_safety_payoffs}
	
	In this section we prove Theorem \ref{thm:Alice_safety_payoffs_summary}, the statement of which is included next.
	
	\begin{customthm}{\ref{thm:Alice_safety_payoffs_summary}}[Alice enforcing equitable payoffs; formal]
		In both the sequential and simultaneous settings, Alice has a pure strategy $S_A$, such that for every  Bob strategy $S_B$:
		\begin{itemize}
			\item   on every trajectory of play, Alice's  average payoff is  at least  $1/2 - o(1)$, while Bob's average payoff is at most $1/2 + o(1)$. More precisely, for all $t \in \{3, \ldots, T\}$:
			\begin{align}
				\frac{u_B(1, t)}{t} &\leq \frac{1}{2} + \frac{5\Delta+11}{\ln(2t/5)} \label{ineq:bob-punishment-explicit} \\
				\frac{u_A(1, t)}{t} &\geq \frac{1}{2} - \frac{4}{\sqrt{t-1}} \label{ineq:bob-punishment-explicit-alice}, 
			\end{align} 
			recalling that $\Delta$ is the upper bound on the players' value densities. 
		\end{itemize}
		%
		%    In both the sequential and simultaneous settings, there exists a strategy $S_A$ for Alice which ensures that, for every Bob strategy $S_B$, Bob valuation $v_B$, and all $t \geq 3$:
		%
		%   \begin{align}
		%        \frac{u_B(1, t)}{t} \leq \frac{1}{2} + \frac{5\Delta+11}{\ln(2t/5)} \qquad \mbox{and} \qquad
		%       \frac{u_A(1, t)}{t} \geq \frac{1}{2} - \frac{4}{\sqrt{t-1}} \,. \notag
		%    \end{align}
		% That is, payoffs to Bob are upper bounded by $1/2$ on average, while Alice gets at least $1/2$ on average. 
		Moreover, even if Bob's value density is unbounded, his average payoff will still converge to ${1}/{2}$.
	\end{customthm}
	
	The proof of the theorem  is  deferred until additional definitions have been stated and helpful lemmas have been proved.
	We first define some notations.

	\begin{definition}[Set of valuations $\mathcal{W}_n$]
		\label{def:set_of_valuations_W_n}
		For each $n \in \mathbb{N}^*$, we define the following set of  non-decreasing piecewise linear  functions with $n$ pieces:
		\begin{align}
			\mathcal{W}_n & = \Bigl\{ f : [0,1] \to [0,1]  \mid f \mbox{ is  non-decreasing  with } f(0)=0, f(1)=1,  f(i/n) \cdot n \in \mathbb{Z}_{\ge 0}, \mbox{ and } \Bigr. \notag \\
			& \Bigl. \qquad f(x) = f\Bigl(\frac{i}{n}\Bigr) + \Bigl(x-\frac{i}{n}\Bigr)\Bigl(f\left(\frac{i+1}{n}\right) - f\left(\frac{i}{n}\right)\Bigr) \forall  i \in \{ 0, \ldots, n-1\} \, \forall  x \in \left[\frac{i}{n},\frac{i+1}{n} \right]  \Bigr\}  \,. \notag 
		\end{align}
	\end{definition}
	
	\begin{definition}[Set of functions $\mathcal{V}_n$] 
		\label{def:set_of_valuations_V_n}%\label{def:rich_set_of_valuations}
		For each $n \in \mathbb{N}^*$, recall  $\mathcal{W}_n$ was given in  Definition~\ref{def:set_of_valuations_W_n} and define   $\mathcal{V}_n$ as the following set of functions:
		\begin{align}
			\mathcal{V}_n = 
			\begin{cases}
				\mathcal{W}_n & \text{ if } n \neq 2\\
				\mathcal{W}_2 \cup  \{f_A\} & \text{ if } n = 2, \mbox{ where } f_A : [0,1] \to [0,1]$ is the function $f_A(x) = V_A([0,x]).
			\end{cases} \notag
		\end{align}
	\end{definition}
	
	\begin{definition}[The set $\overline{\mathcal{V}}$]
		\label{def:rich_set_of_valuations} Let $
		\overline{\mathcal{V}} = \bigcup_{n=1}^{\infty} \left\{(n, V) : V \in \mathcal{V}_n\right\},
		$  
		where  $\mathcal{V}_n$ is given by   Definition~\ref{def:set_of_valuations_V_n}. 
	\end{definition}

	\begin{remark} \label{lem:v_n_increasing}
		By construction, for each $n \in \mathbb{N}^*$, every function $f \in \mathcal{V}_n$ is non-decreasing.
	\end{remark}
	
	%In Definition~\ref{def:rich_set_of_valuations}, the set   $\mathcal{W}_n$ represents all the  nondecreasing piecewise linear over the lattice points of $\left(\frac{1}{n}\mathbb{Z}\right)^2$.
	For each $n \in \mathbb{N}^*$, $\mathcal{W}_n$ contains the nondecreasing piecewise linear functions through a grid with spacing $1/n$. For large $n$, then, $\mathcal{W}_n$ should contain approximations to any given function accurate to roughly $O(1/n)$.
	
	If Alice bases her strategy on limiting the payoff of the Bobs in each $\mathcal{W}_n$, then she will limit Bob's payoff but not necessarily guarantee herself a very good payoff. The inclusion of Alice's valuation function in $\overline{\mathcal{V}}$ rectifies this, allowing us to show a much tighter bound on Alice's payoff.
	
	To formalize the ability of the elements of the $\mathcal{V}_n$ to approximate arbitrary valuation functions, we prove the following three lemmas. The first of them (Lemma \ref{lem:arbitrary_vn_tightness_given_function_bounds}) is somewhat technical, but most directly shows the richness of the $\mathcal{V}_n$. It will be used to prove Lemmas \ref{Vn-arbitrary-function-approximation} and \ref{Vn-approximation-tightness}, which will be used directly to bound the payoff of an unbounded-density and bounded-density Bob, respectively.
	% Alice's strategy will use the sets $\mathcal{V}_n$ and including her valuation in $\mathcal{V}_2$ will be useful for improving her regret bound.
	% The sets $\mathcal{V}_n$ are rich enough to approximate well any valuation function that Bob might have. 
	
	%The other elements of each $\mathcal{V}_n$ serve to approximate arbitrary valuation functions of Bob. 

	\begin{lemma}
		\label{lem:arbitrary_vn_tightness_given_function_bounds}
		Let $f : [0, 1] \to [0, 1]$ be  continuous and increasing  with $f(0)=0$ and $f(1) =1$. Suppose there exist $n \in \mathbb{N}^*$ and  $\varepsilon \in (0, \infty)$ such that 
		\begin{align}
			|f(x) - f(y)| \leq \varepsilon \qquad \forall x, y \in [0, 1] \mbox{ with }  |x-y| \leq {1}/{n} \,. \label{eq:lemma_assumed_absolute_value_f_x_minus_f_y_at_most_epsilon} 
		\end{align}
		Then there exists $V_n \in \mathcal{V}_n$, where  $\mathcal{V}_n$ is  the set of functions from  Definition~\ref{def:set_of_valuations_V_n}, such that\\
		$
		|f(x) - V_n(x)| \leq \varepsilon + {2}/{n} \; \; \forall x \in [0, 1].
		$
	\end{lemma}
	\begin{proof}
		%We first explain the choice of $V_n$ informally and then state its definition. 
		For $x \in \mathbb{R}$, let  $\lfloor x \rceil$ denote the nearest integer to $x$, {breaking ties in favor of $\lceil  x \rceil$} when $x = \lfloor x \rfloor + 1/2$.
		
		%{By Definition~\ref{def:set_of_valuations_W_n}, the set  of functions $\mathcal{W}_n$  contains exactly one function for each choice of non-decreasing values at $i/n$ for $i \in [n-1]$.} 
		Recall the set of functions  $\mathcal{W}_n$ from Definition~\ref{def:set_of_valuations_W_n}. Let  $V_n$  be the function in $\mathcal{W}_n$ such that:
		\begin{align}
			\label{eq:vn_rounded_from_f}
			V_n(i/n) = \frac{\lfloor n \cdot f(i/n) \rceil}{n} \qquad \forall i \in [n-1]\,.
		\end{align}
		
		%. More intuitively, at each $i/n$ $V_n$ takes on the value of $f$, rounded to the nearest integer multiple of $1/n$.
		%
		%\rpcomment{If we want to fully write out exactly what this function is, it's the following. To me, it's less clear that $V_n \in \mathcal{W}_n$ than with the previous description.}
		%\begin{align}
		%    V_n(x) = \frac{\left\lfloor n \cdot f\left(\frac{\lfloor nx \rfloor}{n}\right) \right\rceil}{n} + \left(x-\frac{\lfloor nx \rfloor}{n}\right)\left(\frac{\left\lfloor n \cdot f\left(\frac{\lceil nx \rceil}{n}\right) \right\rceil}{n} - \frac{\left\lfloor n \cdot f\left(\frac{\lfloor nx \rfloor}{n}\right) \right\rceil}{n}\right)
		%\end{align}
		% \simina{The function $V_n$ below doesn't seem to be as defined in the definition of the set $\mathcal{V}_n$. We should write exactly what it is, and can give intuition afterwards, or the other way around.}
		
		Then we claim  the function $V_n$ approximates well the function $f$ at the points $i/n$, that is: % $V_n(i/n)$ is accurate to within $\frac{1}{2n}$, or:}
		\begin{align}
			\label{eq:vn_choice_within_1_over_2n}
			\left|V_n(i/n) - f(i/n) \right|  = \left|\frac{\lfloor n \cdot f(i/n) \rceil}{n} - f(i/n) \right| \leq \frac{1}{2n} \qquad  \forall i \in \{0, \dots, n\},
		\end{align}
		%Note that such a $V_n$ can be found by rounding the values of $f(i/n)$ to the nearest integer multiple of $1/n$. Linearly interpolating between these values gives the desired $V_n \in \mathcal{W}_n$.
		
		where for $i \in [n-1]$ the inequality in \eqref{eq:vn_choice_within_1_over_2n} follows from \eqref{eq:vn_rounded_from_f}, for $i=0$ it follows from $V_n(0)=0=f(0)$, and for $i=n$ it follows from $V_n(1)=1=f(1)$. 
		
		\medskip 
		
		Let $x \in [0, 1]$. 
		We show  three inequalities next:
		\begin{enumerate}
			\item By inequality \eqref{eq:lemma_assumed_absolute_value_f_x_minus_f_y_at_most_epsilon} from  the lemma statement with parameters $\lfloor  xn \rfloor /n$ and $\lceil xn \rceil / n$, we get 
			\begin{align} \label{eq:ineq_1_1_over_2n}
				f\left(\frac{\lceil xn \rceil}{n}\right) - f\left(\frac{\lfloor xn \rfloor}{n}\right) \leq \varepsilon \,. 
			\end{align}
			\item By inequality \eqref{eq:vn_choice_within_1_over_2n} {with $i=\lceil xn \rceil$}, we have 
			\begin{align}  \label{eq:ineq_2_1_over_2n}
				\left| V_n\left(\frac{\lceil xn \rceil}{n}\right) - f\left(\frac{\lceil xn \rceil}{n}\right) \right| \leq \frac{1}{2n}\,.
			\end{align}
			\item By inequality  \eqref{eq:vn_choice_within_1_over_2n} {with $i=\lfloor xn \rfloor$}, we have 
			\begin{align} \label{eq:ineq_3_epsilon_plus_1_over_n}
				\left| f\left(\frac{\lfloor xn \rfloor}{n}\right) - V_n\left(\frac{\lfloor xn \rfloor}{n}\right) \right| \leq \frac{1}{2n}\,. 
			\end{align}
		\end{enumerate}
		
		Summing up  inequalities \eqref{eq:ineq_1_1_over_2n}, \eqref{eq:ineq_2_1_over_2n}, \eqref{eq:ineq_3_epsilon_plus_1_over_n}  and applying the triangle inequality, we obtain
		\begin{align}
			\label{eq:vn_interval_width_under_epsilon_plus_stuff}
			\left| V_n \left(\frac{\lceil xn \rceil}{n}\right) - V_n\left(\frac{\lfloor xn \rfloor}{n}\right) \right| \leq  \varepsilon + \frac{1}{n} \,.
		\end{align}
		
		We obtain 
		\begin{align}
			V_n\left(\frac{\lfloor xn \rfloor}{n}\right) - \frac{1}{2n} & \leq f\left(\frac{\lfloor xn \rfloor}{n}\right) \explain{By \eqref{eq:vn_choice_within_1_over_2n} with $i = \lfloor xn \rfloor$} \\
			& \leq f(x) \explain{Since $f$ is non-decreasing} \\
			& \leq f\left(\frac{\lceil xn \rceil}{n}\right) \explain{Since $f$ is non-decreasing}  \\
			& \leq V_n \left(\frac{\lceil xn \rceil}{n}\right) + \frac{1}{2n} \,. \explain{By \eqref{eq:vn_choice_within_1_over_2n} with $i = \lceil  xn \rceil$}
		\end{align}
		Denoting $J = \Bigl[ V_n\bigl(\frac{\lfloor xn \rfloor}{n}\bigr) - {1}/({2n}), V_n \bigl(\frac{\lceil xn \rceil}{n}\bigr) + {1}/({2n})\Bigr]$, we conclude that $f(x) \in J$.
		Since the function $V_n$ is non-decreasing by Remark~\ref{lem:v_n_increasing}, we have 
		\begin{align}
			V_n\left(\frac{\lfloor xn \rfloor}{n}\right) - \frac{1}{2n} \leq V_n\left({ x}\right) - \frac{1}{2n} < V_n(x) < V_n(x) + \frac{1}{2n} \leq  V_n \left(\frac{\lceil xn \rceil}{n}\right) + \frac{1}{2n} \,.
		\end{align}
		Thus $V_n(x) \in J$. Since both $f(x) \in J$ and $V_n(x) \in J$,  we can bound  $|f(x) - V_n(x)|$ as follows:
		\begin{align}
			|f(x) - V_n(x)| &\leq |J| \notag \\
			&= \left(V_n \left(\frac{\lceil xn \rceil}{n}\right) + \frac{1}{2n}\right) - \left(V_n\left(\frac{\lfloor xn \rfloor}{n}\right) - \frac{1}{2n}\right) \notag \\
			%&= V_n \left(\frac{\lceil xn \rceil}{n}\right) - V_n\left(\frac{\lfloor xn \rfloor}{n}\right) + \frac{1}{n} \notag \\
			&\leq \varepsilon + \frac{2}{n} \,. \explain{Since $V_n \bigl(\frac{\lceil xn \rceil}{n}\bigr) - V_n\bigl(\frac{\lfloor xn \rfloor}{n}\bigr) \leq \varepsilon + 1/n$ by \eqref{eq:vn_interval_width_under_epsilon_plus_stuff}}
		\end{align}
		Since this holds for all $x \in [0, 1]$, the function $V_n$ required by the lemma exists.
	\end{proof}
	
	\begin{lemma}
		\label{Vn-arbitrary-function-approximation}
		Let  $f : [0, 1] \to [0, 1]$ be continuous and  increasing  with $f(0)=0$ and $f(1)=1$. For each  $\varepsilon > 0$, there exists $n \in \mathbb{N}^*$ and a function $V_n \in \mathcal{V}_n$ such that 
		\begin{align} 
			|V_n(x) - f(x)| \leq \varepsilon \; \; \forall x \in [0, 1],
		\end{align} 
		recalling that the set of functions $\mathcal{V}_n$ is given by Definition~\ref{def:set_of_valuations_V_n}.
	\end{lemma}
	\begin{proof}
		Let $\epsilon > 0$.
		Since $f$ is continuous and increasing, its inverse $f^{-1}$ is also continuous and increasing. Since  $f(0)=0$ and $f(1)=1$, we have  $f^{-1}(0)=0$ and $f^{-1}(1)=1$. Consider a function  $g: \left[0, 1-{\varepsilon}/{4}\right] \to [0,1]$ defined as 
		\begin{align} \label{def:g}
			g(x) = f^{-1}\left(x+{\varepsilon}/{4}\right) - f^{-1}(x)\,.
		\end{align} %be a function from $$ to $[0, 1]$.
		Since  $f^{-1}$ is continuous and bounded {over a closed interval}, so is the function $g$. Therefore, by the extreme value theorem,  $g$ attains a global minimum value $\delta^*$. Since $f^{-1}$ is strictly increasing, $g$ is never zero, so 
		\begin{align} \delta^* > 0 \,. \label{eq:delta_star_strictly_positive}
		\end{align} 
		
		Since we will analyze $g(f(x))$, it will be useful to define the set of values $x$ for which $g(f(x))$ is well defined, that is: 
		\begin{align}
			S_{\varepsilon} = \left \{  x \in [0,1]  \mid  f(x) \in [0,1-\varepsilon/4] \right \} \,.
		\end{align}
		Since $\delta^*$ is a global minimum of  $g$, we have 
		\begin{align} \label{eq:g_f_x_at_least_delta_star}
			g(f(x)) &\geq \delta^* \qquad \forall x \in S_{\varepsilon} \,.
		\end{align}
		Using the definition of $g$ (equation \eqref{def:g}) in \eqref{eq:g_f_x_at_least_delta_star} yields 
		\begin{align}
			g(f(x)) & = f^{-1}\left(f(x)+{\varepsilon}/
			{4}\right) - f^{-1}(f(x))  \geq \delta^* \qquad \forall x \in S_{\varepsilon} \,. \label{eq:rewriting_g_f_x_inequality_wrt_delta_star}
		\end{align}
		Since $f^{-1}(f(x)) = x$,  inequality \eqref{eq:rewriting_g_f_x_inequality_wrt_delta_star} yields $f^{-1}\left(f(x)+{\varepsilon}/{4}\right) - x  \geq \delta^*$, or equivalently, \begin{align}  \label{eq:f_to_minus_1_arg_f_plus_varepsilon}
			f^{-1}\left(f(x)+{\varepsilon}/{4}\right) \geq x + \delta^* \qquad \forall x \in S_{\varepsilon} \,.
		\end{align}
		Applying $f$ to both sides of \eqref{eq:f_to_minus_1_arg_f_plus_varepsilon} and using  monotonicity of $f$, we obtain  $f(x) +  {\varepsilon}/{4}   \geq f(x+ \delta^*)$ for all $ x \in S_{\varepsilon}$, or equivalently 
		\begin{align}
			f\left(x + \delta^*\right) - f(x) & \leq {\varepsilon}/{4} \qquad \forall x \in S_{\varepsilon}\,.  \label{eq:f_gap_smaller_than_epsilon_four}
		\end{align}
		Since $\varepsilon > 0$ and $\delta^* > 0$ by \eqref{eq:delta_star_strictly_positive}, there exists 
		$n \in \mathbb{N}$  such that 
		\begin{align} 
			n > 1/\delta^* \; \;  \mbox{and} \; \;  n > {4}/{\varepsilon}\,. \label{eq:choosing_n_larger_than_1_over_delta_and_2_over_epsilon}
		\end{align}
		
		Note that because the range of $f^{-1}$ is $[0, 1]$, the left side of \eqref{eq:f_to_minus_1_arg_f_plus_varepsilon} is at most $1$. Therefore, from $\eqref{eq:f_to_minus_1_arg_f_plus_varepsilon}$ and the fact that $\delta^* > 1/n$ by \eqref{eq:choosing_n_larger_than_1_over_delta_and_2_over_epsilon}:
		\begin{align}
			\label{eq:x_in_S_epsilon_bounded_by_1-delta}
			1 \geq x+\delta^* > x+1/n \qquad \forall x \in S_{\varepsilon}
		\end{align}
		
		Consider an arbitrary $x, y \in [0, 1]$ with $x \leq y \leq x+1/n$. We consider two cases:
		
		\begin{itemize}
			\item Case $x \in S_{\varepsilon}$: Since $f$ is strictly increasing and $x+\delta^* \leq 1$ by \eqref{eq:x_in_S_epsilon_bounded_by_1-delta}, we have:
			\begin{align}
				\label{eq:fy_less_than_f_x_plus_delta}
				f(y) \leq f(x+1/n) < f(x+\delta^*)
			\end{align}
			Subtracting $f(x)$ from both sides of \eqref{eq:fy_less_than_f_x_plus_delta} and applying \eqref{eq:f_gap_smaller_than_epsilon_four}, we obtain
			\begin{align}
				\label{eq:fx_fy_gap_when_x_in_S_epsilon}
				f(y) - f(x) < \varepsilon/4
			\end{align}
			\item Case $x \not \in S_{\varepsilon}$: Then $f(x) > 1-\varepsilon/4$. Since $f(y) \leq 1$, we have:
			\begin{align}
				f(y) - f(x) &< 1 - (1-\varepsilon/4) = \varepsilon/4\,.
			\end{align}
		\end{itemize}

		Combining  cases $x \in S_{\varepsilon}$ and $x \not \in S_{\varepsilon}$ gives 
		$ f(y) - f(x) < \varepsilon/4$ $\forall x,y \in [0,1] \mbox{ with } x \leq y \leq x+ 1/n\,.$ 
		Since $f$ is strictly increasing, we have $f(x) \leq f(y) $ when $x \leq y$, and so 
		\begin{align}
			|f(x) - f(y)| < {\varepsilon}/{4} \qquad \forall x,y \in [0,1] \mbox{ with } |x - y| \leq {1}/{n}\,.
		\end{align} 
		
		%    In particular, choosing $n$ \simina{What is $n$? It has to be declared} such that $\frac{1}{n} < \delta^*$ guarantees that the value of $f$ changes by less than $\frac{\varepsilon}{4}$ in any interval of width $\frac{1}{n}$.
		
		By  Lemma \ref{lem:arbitrary_vn_tightness_given_function_bounds}, there exists a function $V_n \in \mathcal{V}_n$ such that:
		\begin{align} \label{eq:small_distance_between_f_x_and_v_n_x}
			|f(x) - V_n(x)| &\leq {\varepsilon}/{4} + {2}/{n} \qquad \forall x \in [0, 1] \,. 
		\end{align}
		Since $ n > 4/\varepsilon$ by \eqref{eq:choosing_n_larger_than_1_over_delta_and_2_over_epsilon}, inequality  \eqref{eq:small_distance_between_f_x_and_v_n_x} gives 
		\begin{align} 
			|f(x) - V_n(x)| &\leq {\varepsilon}/{4} + {2}/{n}  < {\varepsilon}/{4} + {\varepsilon}/{2} 
			< \varepsilon \qquad \forall x \in [0,1] \,. 
		\end{align}
		This completes the proof.
	\end{proof}
	
	\begin{lemma}
		\label{Vn-approximation-tightness}
		Suppose $v_B$ is a  value density function for Bob with $v_B(x) \leq \Delta$ for some $\Delta > 0$ and all $x \in [0,1]$. 
		%Consider an arbitrary valuation function $V_B$ whose density $v_B$ is upper-bounded by a constant $D$.
		Then for all $n \in \mathbb{N}^*$, there exists a function $V_n \in \mathcal{V}_n$ such that 
		\[ \Bigl|V_n(x) - V_B([0, x]) \Bigr| \leq \frac{\Delta+2}{n} \qquad \forall x \in [0, 1] \,.
		\] 
	\end{lemma}
	\begin{proof}
		Let  $n \in \mathbb{N}^*$. Since Bob's density is upper bounded by $\Delta$, we have 
		\begin{align}
			|V_B([0, x]) - V_B([0, y])| \leq \Delta|x-y| \qquad \forall x, y \in [0,1]\,.
		\end{align}
		
		When $|x-y| \leq {1}/{n}$, we get 
		$ |V_B([0, x]) - V_B([0, y])| \leq \Delta|x-y| 
		\leq {D}/{n} $.
		\medskip 
		
		By Lemma \ref{lem:arbitrary_vn_tightness_given_function_bounds} applied to the function $f: [0,1] \to [0,1]$ given by  $f(x) = V_B([0,x])$, there exists $V_n \in \mathcal{V}_n$ with 
		$
		|V_n(x) - V_B([0, x])| \leq {\Delta}/{n} + {2}/{n}$ for all $x \in [0, 1]\,.
		$
		This completes the proof.
	\end{proof}
	
	% \sscomment{I don't understand sentence the following at first glance, too}
	% It's worth asking how many functions are being considered here, as that will tie directly into the final bound on the players' utilities. Luckily, the sizes of the $\mathcal{V}_n$ are well-behaved.
	
	The set of functions $\mathcal{V}_n$ will be used to construct a strategy for Alice. The rate of growth of the size of the set $|\mathcal{V}_n|$ as a function of $n$ will influence  the error bounds on the utilities of the  players. Next we bound the size of the set  $\mathcal{V}_n$.
	
	\begin{lemma}
		\label{vn-size-bound}
		$|\mathcal{V}_n| \leq 4^{n-1}$   $\; \forall  n \in \mathbb{N}^*$. 
	\end{lemma}
	
	\begin{proof}
		%\simina{Continue from here.}   
		We first estimate  the size of each set $\mathcal{W}_n$, and then will infer the bound for the size of $\mathcal{V}_n$. Consider the density function corresponding to any particular $V \in \mathcal{W}_n$. Because $V$ is piecewise linear, its density is piecewise constant. Each $V$ is then uniquely determined by a sequence $d_1, d_2, \dots, d_n$, where $d_i$ is the value density between $\frac{i-1}{n}$ and $\frac{i}{n}$. Each $d_i$ must be a non-negative integer, because each of these intervals has width $1/n$ and sees $V$ rise by an integer multiple of $1/n$. More strongly, because $V(0)=0$ and $V(1)=1$, we must have the next relation between the $d_i$'s:
		\begin{align}
			& \sum_{i=1}^n \frac{d_i}{n} = 1 \iff   
			\sum_{i=1}^n d_i = n\,.
		\end{align}
		%Or, put another way, $\sum_{i=1}^n d_i = n$.
		
		Thus, the size $|\mathcal{W}_n|$ is the number of possible partitions of $n$ into $n$ parts with nonnegative integer sizes.
		The size of $\mathcal{W}_n$ can then be counted by a standard combinatorics technique. Represent each choice of $d_1, d_2, \dots, d_n$ with a sequence of $n$ ``stars" and $n-1$ ``bars": $d_1$ stars, then a bar, then $d_2$ stars, then another bar, and so on. Each choice of $d_1, \dots, d_n$ corresponds to a unique arrangement of stars and bars. The opposite is also true: given an arrangement, the number of stars between each pair of bars can be read off as $d_1, \dots, d_n$. The size of $\mathcal{V}_n$ is then the number of arrangements, which is $\binom{2n-1}{n}$. 
		
		The upper bound can then be shown by induction. As a base case, $\binom{2 \cdot 1 - 1}{1} = \binom{1}{1} = 1 \leq 4^{1-1}$.
		
		Now assume $\binom{2n-1}{n} \leq 4^{n-1}$ for an arbitrary $n \geq 1$. We have 
		\begin{align*}
			\binom{2(n+1)-1}{n+1} &= \frac{(2n+1)!}{(n+1)! \cdot n!} = \frac{(2n-1)! \cdot 2n \cdot (2n+1)}{n! \cdot (n-1)! \cdot n(n+1)} = \binom{2n-1}{n} \cdot \frac{2n(2n+1)}{n(n+1)} \\
			&\leq 4^{n-1} \cdot 4\frac{n+\frac{1}{2}}{n+1} \explain{By the inductive hypothesis} \\
			&\leq 4^{n-1} \cdot 4 = 4^{(n+1)-1}
		\end{align*}
		So by induction, the bound holds for all $n \geq 1$.
		
		Because $\mathcal{V}_n = \mathcal{W}_n$ for all $n \neq 2$, the only case left to verify is $n=2$. By calculation, $|\mathcal{W}_2| = \binom{3}{2} = 3$, so including the extra function makes $|\mathcal{V}_2| = 4 = 4^{2-1}$.
	\end{proof}
	
	Some other technical lemmas are necessary. The first two relate to the following function, which acts like an infinite-dimensional inner product.
	
	\begin{definition}
		% Let $\overline{\mathcal{V}} = \bigcup_{n=1}^{\infty} \left\{(n, V) : V \in \mathcal{V}_n\right\}$ be the set of all elements from all $\mathcal{V}_n$, labelled with the associated $n$.
		Let $\mathcal{X}$ be the collection of all functions $X:\overline{\mathcal{V}} \to [-1,1]$.
		For functions $X, Y : \overline{\mathcal{V}} \to \left[-1, 1\right]$, define a function $P:\mathcal{X} \times \mathcal{X} \to \R$ as follows:
		\[
		P(X, Y) = \sum_{n=1}^{\infty} \frac{1}{2^n \left|\mathcal{V}_n\right|} \sum_{V \in \mathcal{V}_n} X(n, V) Y(n, V)\,.
		\]
	\end{definition}
	
	\begin{lemma}
		\label{any-bob-function-inner-product}
		The function $P$ has the properties of an inner product. In particular, for all functions $X, Y, Z : \overline{\mathcal{V}} \to \left[-1, 1\right]$ and all constants $c \in [-1, 1]$, the following holds:
		\begin{enumerate}[(1)]
			\item $P(X, Y)$ exists, and the infinite sum converges absolutely
			\item $P(X, Y) = P(Y, X)$
			\item $|P(X, Y)| \leq 1$
			\item $P(cX, Y) = cP(X, Y)$
			\item $P(X+Z, Y) = P(X, Y) + P(Z, Y)$, assuming $X+Z$ is in the domain of $P$
			\item $P(X, X) = 0$ if $X(n, V) = 0$ for all $n$ and all $V$, and $P(X, X) > 0$ otherwise
		\end{enumerate}
	\end{lemma}
	
	\begin{proof}
		We separately prove each of the properties.
		
		For the first property, note that the individual terms go to zero
		\begin{align*}
			\lim_{n \to \infty} \left|\frac{1}{2^n |\mathcal{V}_n|} \sum_{V \in \mathcal{V}_n} X(n, V)Y(n, V)\right| &\leq \lim_{n \to \infty} \frac{1}{2^n |\mathcal{V}_n|} \sum_{V \in \mathcal{V}_n} \left|X(n, V)Y(n, V)\right| \\
			&\leq \lim_{n \to \infty} \frac{1}{2^n |\mathcal{V}_n|} \sum_{V \in \mathcal{V}_n} 1 \explain{By the bounds on $X$ and $Y$} \\
			&= \lim_{n \to \infty} \frac{1}{2^n |\mathcal{V}_n|} |\mathcal{V}_n| = \lim_{n \to \infty} \frac{1}{2^n} \\
			&= 0\,.
		\end{align*}
		Using the above upper bound for individual terms,  the sum of the absolute values does not diverge as follows
		\begin{align*}
			\sum_{n=1}^{\infty} \left|\frac{1}{2^n |\mathcal{V}_n|} \sum_{V \in \mathcal{V}_n} X(n, V)Y(n, V)\right| &\leq \sum_{n=1}^{\infty} \frac{1}{2^n}  = 1\,.
		\end{align*}
		So $P(X, Y)$ exists and the infinite sum converges absolutely.

		The existence is enough for property (2), which follows directly from the product $X(n, V)Y(n, V)$ being commutative. This calculation directly verifies property (3).

		The absolute convergence allows for the linear operations in properties (4) and (5) to factor through the sum.
		
		First, for property (4) we obtain
		\begin{align*}
			P(cX, Y) &= \sum_{n=1}^{\infty} \frac{1}{2^n |\mathcal{V}_n|} \sum_{V \in \mathcal{V}_n} cX(n, V)Y(n, V) \\
			&= c\sum_{n=1}^{\infty} \frac{1}{2^n |\mathcal{V}_n|} \sum_{V \in \mathcal{V}_n} X(n, V)Y(n, V) \\
			&= cP(A, B)\,.
		\end{align*}
		For property (5) observe that:
		\begin{align*}
			P(X+Z, Y) &= \sum_{n=1}^{\infty} \frac{1}{2^n |\mathcal{V}_n|} \sum_{V \in \mathcal{V}_n} (X(n, V) + Z(n, V)) Y(n, V) \\
			&= \sum_{n=1}^{\infty} \frac{1}{2^n |\mathcal{V}_n|} \sum_{V \in \mathcal{V}_n} X(n, V)Y(n, V) + Z(n, V)Y(n, V) \\
			&= \sum_{n=1}^{\infty} \frac{1}{2^n |\mathcal{V}_n|} \sum_{V \in \mathcal{V}_n} X(n, V)Y(n, V) + \sum_{n=1}^{\infty} \frac{1}{2^n |\mathcal{V}_n|} \sum_{V \in \mathcal{V}_n} Z(n, V)Y(n, V) \\
			&= P(X, Y) + P(Z, Y)\,.
		\end{align*}
		For property (6), observe that the expression can be rewritten as:
		\[
		P(X, X) = \sum_{n=1}^{\infty} \frac{1}{2^n |\mathcal{V}_n|} \sum_{V \in \mathcal{V}_n} X(n, V)^2\,.
		\]
		This is a sum of squares, which will be zero if all the included $X(n, V)$ are zero and positive otherwise.
	\end{proof}
	
	\begin{lemma}
		\label{any-bob-function-continuity}
		For each $x \in [0, 1]$, let $G_x : \overline{\mathcal{V}} \to [-1/2, -1/2]$ be a function defined by $G_x(n, V) = V(x)-\frac{1}{2}$. Then, for any $Z : \overline{\mathcal{V}} \to \left[-1, 1\right]$, the following function is continuous in $x$:
		\[
		P(G_x, Z) = \sum_{n=1}^{\infty} \frac{1}{2^n |\mathcal{V}_n|} \sum_{V \in \mathcal{V}_n} \left(V(x) - \frac{1}{2}\right)Z(n, V).
		\]
		% Is continuous in $x$.
	\end{lemma}
	
	\begin{proof}
		Alice's valuation function in $\mathcal{V}_2$ needs to be handled separately. Accordingly, write:
		\begin{align*}
			P(G_x, Z) &= \sum_{n=1}^{\infty} \frac{1}{2^n |\mathcal{V}_n|} \sum_{V \in V_n} \left(V(x) - \frac{1}{2}\right)Z(n, V) \\
			&= \frac{1}{2^2 |\mathcal{V}_2|} \left(V_A([0, x]) - \frac{1}{2}\right)Z(n, V_A) + \sum_{n=1}^{\infty} \frac{1}{2^n|\mathcal{V}_n|} \sum_{V \in \mathcal{W}_n} \left(V(x) - \frac{1}{2}\right)Z(n, V)\,.
		\end{align*}
		As long as each of these two parts is continuous, the sum will be too. The first part is continuous because $V_A([0, x])$ is continuous.
		To prove the second part, for notational simplicity, define
		\[
		P'(G_x, Z) = \sum_{n=1}^{\infty} \frac{1}{2^n|\mathcal{V}_n|} \sum_{V \in \mathcal{W}_n} \left(V(x) - \frac{1}{2}\right)Z(n, V)
		\]
		
		The continuity of $P'$ will be shown by directly appealing to the definition of a limit. First, note that each $V \in \mathcal{W}_n$ is made of $n$ linear segments, each of width $\frac{1}{n}$ and height at most $1$. 
		Therefore, if $V \in \mathcal{W}_n$, for all $x, y \in [0, 1]$:
		\begin{align}
			|V(y) - V(x)| \leq |y-x|n    \label{ineq:10232156}
		\end{align}
		Fix an arbitrary $\varepsilon > 0$ and $x \in [0, 1]$. Choose $\delta = \varepsilon/2$. For any $y$ such that $|y-x| < \delta$:
		\begin{align*}
			\left|P'(G_y, Z) - P'(G_x, Z)\right| &= \left|P'(G_y - G_x, Z)\right| \explain{By property (5)}\\
			&= \left| \sum_{n=1}^{\infty} \frac{1}{2^n |\mathcal{V}_n|} \sum_{V \in \mathcal{W}_n} (V(y) - V(x))Z(n, V)\right| \\
			&\leq \sum_{n=1}^{\infty} \frac{1}{2^n |\mathcal{V}_n|} \sum_{V \in \mathcal{W}_n} |V(y) - V(x)| \cdot |Z(n, V)| \explain{Triangle inequality} \\
			&\leq \sum_{n=1}^{\infty} \frac{1}{2^n |\mathcal{V}_n|} \sum_{V \in \mathcal{W}_n} n|y-x| \cdot 1 \explain{By ineq~\eqref{ineq:10232156} and $Z \le 1$} \\
			&\leq |y-x| \sum_{n=1}^{\infty} \frac{n}{2^n} \explain{Since $|\mathcal{W}_n| \le |\mathcal{V}_n|$}\\
			&= |y-x| \cdot 2 \\
			&< 2\delta \explain{By choice of $\delta$} \\
			&= \varepsilon \,.
		\end{align*}
		Therefore, for all $x \in [0, 1]$:
		\[
		\lim_{y \to x} P'(G_y, Z) = P'(G_x, Z),
		\]
		which is the definition of $P'(G_x, Z)$ being continuous in $x$. 
		This concludes that $P(G_x, Z)$ is continuous in $x$.
	\end{proof}
	
	Finally, the last one is a strengthening of Lemma 1 from \cite{blackwell1956}, under a stronger hypotheses.
	
	\begin{lemma}
		\label{any-bob-delta-limit}
		Suppose a sequence of nonnegative values $\delta_1, \delta_2, \dots$ satisfies, for all $t \geq 1$:
		\begin{align*}
			\delta_{t+1} &\leq \frac{1}{(t+1)^2} + \left(\frac{t-1}{t+1}\right)\delta_t\,.
		\end{align*}
		Then $\lim_{t \to \infty} \delta_t = 0$. In particular, for all $t \geq 2$:
		\[
		\delta_t \leq \frac{1}{t(t-1)} \sum_{i=2}^t \frac{i-1}{i} \leq \frac{1}{t}\,.
		\]
	\end{lemma}
	
	\begin{proof}
		The latter bound will be shown by induction. As a base case, consider $t=2$. The condition for $t=1$ gives:
		\begin{align}
			\delta_2 &\leq \frac{1}{(1+1)^2} + \left(\frac{1-1}{1+1}\right) \delta_1 = \frac{1}{4} = \frac{1}{2(2-1)} \sum_{i=2}^2 \frac{i-1}{i} \,.
		\end{align}
		So the base case holds. Now suppose the conclusion holds for some $t \geq 2$. 
		Assume that the claim holds for $\delta_t$.
		We can expand the upper bound of $\delta_{t+1}$ as follows.
		\begin{align}
			\delta_{t+1} &\leq \frac{1}{(t+1)^2} + \left(\frac{t-1}{t+1}\right) \delta_t 
			\\
			&\leq \frac{1}{(t+1)^2} + \left(\frac{t-1}{t+1}\right) \frac{1}{t(t-1)} \sum_{i=2}^t \frac{i-1}{i} \explain{{By the induction hypothesis}}
			\\
			&= \frac{1}{(t+1)^2} + \frac{1}{t(t+1)} \sum_{i=2}^t \frac{i-1}{i} \\
			%&= \frac{1}{(t+1)t} \left(\frac{t}{t+1} + \sum_{i=2}^t \frac{i-1}{i}\right) \\
			&= \frac{1}{(t+1)t} \sum_{i=2}^{t+1} \frac{i-1}{i} \,. \label{eq:rhs_sum_series_1_over_t_stuff}
		\end{align}
		
		Now it suffices to show that the last term in \eqref{eq:rhs_sum_series_1_over_t_stuff} is upper bounded by $1/t$, which follows from the following inequality:
		\begin{align}
			\frac{1}{t(t-1)} \sum_{i=2}^t \frac{i-1}{i} &\leq \frac{1}{t(t-1)} \sum_{i=2}^t 1 = \frac{1}{t(t-1)} \cdot (t-1) 
			= \frac{1}{t}\,.
		\end{align}
		Thus $\delta_{t+1} \leq 1/t$, which  completes the proof.
	\end{proof}
	
	Now the groundwork is laid to prove Theorem \ref{thm:Alice_safety_payoffs_summary}. The idea is to limit every strategy and every valuation in every $\mathcal{V}_n$ to an average payoff of $1/2$, which is sufficient to limit Bob equipped with arbitrary valuation as well.
	
	\begin{proof}[Proof of Theorem~\ref{thm:Alice_safety_payoffs_summary}]
		We first define the Alice's strategy in an analytical manner, and then prove that it is well-defined.
		Then we will show that this strategy will force Bob's payoff to be at most $1/2$ on average, by showing both that Bob's valuation can be well-approximated by functions in $\overline{\mathcal{V}}$ and that such valuation functions are limited to a payoff of $1/2$ on average.
		Finally, we establish explicit convergence rates for Bob's payoff if his valuation is bounded, as well as the convergence of Alice's payoff to $1/2$.
		
		\paragraph{Defining Alice's strategy $S_A$}
		Consider Alice's decision in round $T$. In each round $t < T$, let the payoff to Bob whose cumulative valuation function is $V$ be $u_{t,V}$. 
		% \sscomment{It might be better to recall that what $V(x)$ is.}
		For each $x \in [0, 1]$, 
		let $G_x:\overline{\mathcal{V}} \to [-1, 1]$ be a function defined as $G_x(n, V) = V(x) - 1/2$.
		% let $G_x(n, V) = V(x) - \frac{1}{2}$ be a function from $\overline{\mathcal{V}}$ to $[-1, 1]$. 
		Let $U_t(n, V) = u_{t, V} - 1/2$  for $t=1, \dots, T-1$, and let $\overline{U}_t(n, V) = \sum_{i=1}^t U_i(n, V)/t$. Let $W_t(n, V) = \max \{0, \overline{U}_t(n, V)\}$ for $t=1, \dots, T-1$. 
		In round $T$, Alice's strategy $S_A$ is to cut at a point $x$ such that $P(G_x, W_{T-1}) = 0$. 
		% \rpcomment{$W_{T-1}$ here is right because it's the one depending on $\overline{U}_{T-1}$}
		
		We then show that $S_A$ is well-defined.
		It suffices to prove that such an $x$ always exists.
		To this end, observe that
		\begin{align*}
			P(G_0, W_{T-1}) &= \sum_{n=1}^{\infty} \frac{1}{2^n |\mathcal{V}_n|} \sum_{V \in \mathcal{V}_n} \left(V(0) - \frac{1}{2}\right)W_{T-1}(n, V) \\
			&= \sum_{n=1}^{\infty} \frac{1}{2^n|\mathcal{V}_n|} \sum_{V \in \mathcal{V}_n} -\frac{1}{2}W_{T-1}(n, V) \\
			&\leq \sum_{n=1}^{\infty} \frac{1}{2^n|\mathcal{V}_n|} \sum_{V \in \mathcal{V}_n} 0 \explain{Since $W_{T-1}(\cdot)$ is nonnegative} \\
			&= 0\,.
		\end{align*}
		Using similar algebra for $P(G_1, W_{T-1})$, we obtain
		\begin{align*}
			P(G_1, W_{T-1}) &= \sum_{n=1}^{\infty} \frac{1}{2^n |\mathcal{V}_n|} \sum_{V \in \mathcal{V}_n} \left(V(1) - \frac{1}{2}\right)W_{T-1}(n, V) \\
			&= \sum_{n=1}^{\infty} \frac{1}{2^n|\mathcal{V}_n|} \sum_{V \in \mathcal{V}_n} \frac{1}{2}W_{T-1}(n, V) \\
			&\geq \sum_{n=1}^{\infty} \frac{1}{2^n|\mathcal{V}_n|} \sum_{V \in \mathcal{V}_n} 0 \explain{Since $W_{T-1}(\cdot)$ is nonnegative} \\
			&= 0
		\end{align*}
		By Lemma \ref{any-bob-function-continuity}, $P(G_x, W_{T-1})$ is continuous in $x$, so by the Intermediate Value Theorem there exists $x$ such that $P(G_x, W_{T-1}) = 0$.
		
		\paragraph{Bounding Bob's payoff} Define
		\begin{align*}
			\mathcal{S} = \left\{X : -\frac{1}{2} \leq X(n, V) \leq 0 \text{ for all $(n, V) \in \overline{\mathcal{V}}$}\right\} .
		\end{align*}
		For each round $t$, let $\delta_t$ be the distance from $\mathcal{S}$ to $\overline{U}_t$, defined as:
		\begin{align} \label{def:delta_t}
			\delta_t = \inf_{X \in S} P(\overline{U}_t-X, \overline{U}_t-X).
		\end{align}
		For $t < T$, let $Y_t: \overline{\mathcal{V}} \to [-1/2, 1/2]$ be the function defined by:
		\begin{align*}
			Y_t(n, V) = \min\{0, \overline{U}_t(n, V)\}
		\end{align*}
		
		By Claim~\ref{claim:upper_bound_on_delta_t}, we have $\delta_t \leq 1/t$ for all $t \geq 2$.
		
		Now we will show that this strategy guarantees that, for any Bob's valuation function $v_B$ and any strategy $S_B$ Bob employs, his average payoff is at most ${1}/{2}$. Consider an arbitrary $\varepsilon > 0$. It will be shown that, for some sufficiently large $T$, Bob's average payoff up to any point after round $T$ is at most ${1}/{2} + \varepsilon$.
		
		In the general case, this convergence can be established from Lemma \ref{Vn-arbitrary-function-approximation}. 
		Consider $N, T \in \mathbb{N}^*$ such that 
		\begin{itemize} 
			\item  $N$ is  such that there exists $V' \in \mathcal{V}_N$ with $|V'(x) - V_B([0, x])| < {\varepsilon}/{2}$ for all $x \in [0, 1]$. 
			\item  $T$ is sufficiently large so that 
			\begin{align}
				\delta_t < \frac{\varepsilon^2}{4 \cdot 2^N |\mathcal{V}_N|} \qquad \forall t \geq T\,. \label{ineq:02100924}
			\end{align}
			For example, taking $T = \bigl\lceil \frac{4 \cdot 2^N |\mathcal{V}_N|}{\varepsilon^2}\bigr\rceil$ would suffice. 
		\end{itemize}
		We will show that for each $t \geq T$, we have $\overline{U}_t(N, V') <{\varepsilon}/{2}$. 
		This trivially holds if $\overline{U}_t(N, V') \leq 0$, so assume $\overline{U}_t(N, V') > 0$. We then obtain
		\begin{align}
			\overline{U}_t(N, V') &= \sqrt{2^N |\mathcal{V}_N| \cdot \frac{1}{2^N |\mathcal{V}_N|}\overline{U}_t(N, V')^2} 
			\nonumber
			\\
			&= \sqrt{2^N |\mathcal{V}_N| \cdot \frac{1}{2^N |\mathcal{V}_N|}\left(\overline{U}_t(N, V') - Y_t(N, V')\right)^2},
			% &\leq 
			% \sqrt{2^N |\mathcal{V}_N| \sum_{n=1}^{\infty} \frac{1}{2^n |\mathcal{V}_n|} \sum_{V \in \mathcal{V}_n} \left(\overline{U}_t(n, V) - Y_t(n, V)\right)^2} 
			% \\
			% &= \sqrt{2^N |\mathcal{V}_N| P(\overline{U}_t - Y_t, \overline{U}_t - Y_t)} \\
			% &= \sqrt{2^N |\mathcal{V}_N| \delta_t} \\
			% &< \sqrt{2^N |\mathcal{V}_N| \cdot \frac{\varepsilon^2}{4 \cdot 2^N |\mathcal{V}_N|}} \\
			% &= \frac{\varepsilon}{2}.
			\label{eq:U_bar_t_n_v_prime_identity}
		\end{align}
		where \eqref{eq:U_bar_t_n_v_prime_identity} follows from the fact that $Y_t(N, V') = \min \{0, \overline{U}_t(N, V')\} = 0$.
		Using \eqref{eq:U_bar_t_n_v_prime_identity}, we can upper bound $ \overline{U}_t(N, V')$ as follows:
		\begin{align}
			\overline{U}_t(N, V') &\leq  \sqrt{2^N |\mathcal{V}_N| \sum_{n=1}^{\infty} \frac{1}{2^n |\mathcal{V}_n|} \sum_{V \in \mathcal{V}_n} \left(\overline{U}_t(n, V) - Y_t(n, V)\right)^2}  \notag \\ 
			&= 
			\sqrt{2^N |\mathcal{V}_N| P(\overline{U}_t - Y_t, \overline{U}_t - Y_t)} \notag \\
			&= \sqrt{2^N |\mathcal{V}_N| \delta_t} \label{eq:ut-less-than-sqrt-2n-vn-deltat} \\
			&< \sqrt{2^N |\mathcal{V}_N| \cdot \frac{\varepsilon^2}{4 \cdot 2^N |\mathcal{V}_N|}} = \frac{\varepsilon}{2}\,.\explain{By \eqref{ineq:02100924}}
		\end{align}

		% Note that the first inequality comes from $Y_t(N, V') = \min \{0, \overline{U}_t(N, V')\} = 0$, and adding in all the other terms \sscomment{might be better to expand on it}. 
		Therefore, we have $\overline{U}_t(N, V') < {\varepsilon}/{2}$ for all $t \geq T$. Translating back into payoffs, this gives a Bob whose cumulative valuation function is $V'$ an average payoff of less than ${1}/{2} + {\varepsilon}/{2}$. By the choice of $V'$, we have 
		\begin{align}
			|V'(x) - V_B([0, x])| < {\varepsilon}/{2} \qquad \mbox{for all} \; x \in [0, 1] \,.
		\end{align} so the average payoff to a Bob whose valuation function is $v_B$ is less than ${1}/{2} + \varepsilon$. Since this construction works for all $\varepsilon > 0$, Bob's average payoff satisfies the following inequality as required:
		\[
		\frac{u_B(1, t)}{t} \leq \frac{1}{2} + o(1)\,.
		\]
		
		\paragraph{Explicit bounds} It remains to prove inequality~\eqref{ineq:bob-punishment-explicit} when Bob's density is upper bounded by $\Delta$, and prove Alice's explicit payoff in inequality~\eqref{ineq:bob-punishment-explicit-alice}.
		Choose an integer $N$ large enough that
		\begin{align*}
			\frac{\Delta+2}{N} < \frac{\varepsilon}{2}\, .
		\end{align*}
		For instance, taking $N=\left\lceil 2(\Delta+2)/\varepsilon\right\rceil$ is sufficient.
		By Lemma \ref{Vn-approximation-tightness}, there exists $V' \in \mathcal{V}_N$ such that $|V'(x) - V_B([0, x])| < (\Delta+2)/N$ for all $x \in [0, 1]$. Choose $T$ in exactly the same way as the general case, i.e. $T = \lceil 4 \cdot 2^N |\mathcal{V}_N|/\varepsilon^2 \rceil$. By exactly the same algebra as the general case, we can conclude that $\overline{U}_t(N, V') < \eps/2$ for all $t \geq T$. By the choice of $V'$, we have
		\begin{align}
			|V'(x) - V_B([0, x])| < {\varepsilon}/{2} \qquad \mbox{for all} \; x \in [0, 1] \, ,
		\end{align}
		so the average payoff to Bob is at most $\frac{1}{2} + \varepsilon$.
		
		In this case, Bob's average payoff will be within $\varepsilon$ of $1/2$ by time 
		\begin{align}
			T_{\varepsilon} = \left\lceil \frac{4 \cdot 2^N |\mathcal{V}_N|}{\varepsilon^2}\right\rceil \, , \qquad \mbox{where} \quad N = \left\lceil 2(\Delta+2)/\varepsilon\right\rceil \,.
		\end{align}  
		Hence, for such $T_\eps$, we have
		\begin{align*}
			T_{\varepsilon} &\leq \left\lceil \frac{4 \cdot 2^N \cdot 4^{N-1}}{\varepsilon^2}\right\rceil.
		\end{align*}
		Since $\lceil x \rceil \le 2x$ for $x \ge 1/2$, this implies
		\begin{align*}
			T_{\varepsilon} &\leq 2\frac{8^N}{\varepsilon^2}.
		\end{align*}
		Due to our choice of $N$, we can further obtain
		\begin{align} \label{eq:ub_t_epsilon_involved_delta}
			T_{\varepsilon} &\leq 2\frac{8^{\frac{2(\Delta+2)}{\varepsilon}+1}}{\varepsilon^2}.
		\end{align}
		Taking the natural logarithm of both sides in \eqref{eq:ub_t_epsilon_involved_delta}, we have
		\begin{align*}
			\ln (T_{\varepsilon}) &\leq \ln(16) + \frac{2(\Delta+2)}{\varepsilon}\ln(8) + 2\ln \left(\frac{1}{\varepsilon}\right).
		\end{align*}
		Using the fact that $\ln(x) \le x -1 $ for every $x >0$, it follows that
		\begin{align}
			\ln (T_{\varepsilon}) &\leq \ln(16) + \frac{2(\Delta+2)}{\varepsilon}\ln(8) + 2\left(\frac{1}{\varepsilon}-1\right). \label{eq:will_be_rearranged_ub_tepsilon}
		\end{align}
		Rearranging \eqref{eq:will_be_rearranged_ub_tepsilon}, we finally obtain
		\begin{align*}
			\varepsilon &\leq \frac{2\ln(8)(\Delta+2) + 2}{\ln(T_{\varepsilon}) - \ln(16) + 2}.
		\end{align*}
		Note that this requires the step of dividing by $\ln(T_{\varepsilon}) - \ln(16) + 2$ and it needs this quantity to be positive, which is indeed positive for $T \geq 3$. 
		Rounding the terms gives the desired regret bound for Bob.
		% By algebra and Lemma \ref{vn-size-bound}:
		% \begin{align*}
		%     T_{\varepsilon} &= \left\lceil \frac{4 \cdot 2^N |\mathcal{V}_N|}{\varepsilon^2}\right\rceil \\
		%     T_{\varepsilon} &\leq \left\lceil \frac{4 \cdot 2^N \cdot 4^{N-1}}{\varepsilon^2}\right\rceil \\
		%     T_{\varepsilon} &\leq 2\frac{8^N}{\varepsilon^2} \explain{For $x \geq 1/2$, $\lceil x \rceil \leq 2x$} \\
		%     T_{\varepsilon} &\leq 2\frac{8^{\frac{2(D+2)}{\varepsilon}+1}}{\varepsilon^2} \\
		%     \ln (T_{\varepsilon}) &\leq \ln(16) + \frac{2(D+2)}{\varepsilon}\ln(8) + 2\ln \left(\frac{1}{\varepsilon}\right) \\
		%     \ln (T_{\varepsilon}) &\leq \ln(16) + \frac{2(D+2)}{\varepsilon}\ln(8) + 2\left(\frac{1}{\varepsilon}-1\right) \\
		%     \ln (T_{\varepsilon}) - \ln(16) + 2 &\leq \frac{2\ln(8)(D+2)+2}{\varepsilon} \\
		%     \varepsilon &\leq \frac{2\ln(8)(D+2) + 2}{\ln(T_{\varepsilon}) - \ln(16) + 2}
		% \end{align*}
		
		To provide Alice's payoff bound in inequality~\eqref{ineq:bob-punishment-explicit-alice}, consider a hypothetical Bob whose valuation function $\tilde{v}_B$ is exactly the same as Alice's valuation function $v_A$. For all rounds $\tau$, the payoff $\tilde{u}_B^{\tau}$ to this Bob then satisfies:
		\begin{align}
			u_A^{\tau} + \tilde{u}_B^{\tau} &= V_A([0, a_{\tau}]) + V_A([a_{\tau}, 1]) \explain{One player gets $V_A([0, a_{\tau}])$ and the other gets $V_A([a_{\tau}, 1])$} \\
			&= 1
		\end{align}
		For any $t$, summing $u_A^{\tau} + \tilde{u}_B^{\tau}$ over $\tau \in [t]$ and dividing by $t$ then gives:
		\begin{align}
			\frac{u_A(1, t)}{t} + \frac{\tilde{u}_B(1, t)}{t} = 1 \label{eq:bob-punishment-players-total-payoff-equals-one}
		\end{align}
		So it suffices to upper-bound this particular Bob's payoff to lower-bound Alice's payoff.
		\medskip 
		
		Choose an arbitrary $\varepsilon > 0$.
		Let $T_{\varepsilon} = \bigl\lceil \frac{2^2 |\mathcal{V}_2|}{\varepsilon^2} \bigr\rceil$, ensuring $\delta_t < \frac{\varepsilon^2}{2^2 |\mathcal{V}_2|}$ for all $t \geq T_{\varepsilon}$. 
		By construction, Alice's valuation function $V_A([0, x]) \in \mathcal{V}_2$. Taking $N=2$ and $V' = V_A$, we can then use identical algebra as the general-Bob case up to \eqref{eq:ut-less-than-sqrt-2n-vn-deltat} to conclude that, for $t \geq T_{\varepsilon}$:
		\begin{align}
			\overline{U}_t(2, V_A) &\leq \sqrt{2^2 |\mathcal{V}_2| \delta_t} \explain{Copying over \eqref{eq:ut-less-than-sqrt-2n-vn-deltat}} \\
			&< \sqrt{2^2 |\mathcal{V}_2| \cdot \frac{\varepsilon^2}{2^2 |\mathcal{V}_2|}} = \varepsilon
		\end{align}
		So this Bob's payoff is upper-bounded by $1/2 + \varepsilon$, and so by \eqref{eq:bob-punishment-players-total-payoff-equals-one} Alice's payoff is lower-bounded by $1/2 - \varepsilon$.
		To solve for $\varepsilon$, observe that $T_{\varepsilon} \leq {16}/{\varepsilon^2} + 1$.
		Solving this bound on $T_{\eps}$ for $\eps$ gives
		\begin{align*}
			\varepsilon &\leq \frac{4}{\sqrt{T_{\varepsilon} - 1}}\, ,
		\end{align*}
		and it finishes the proof.
		% Solving the choice of $T_{\varepsilon}$ for $\varepsilon$:
		% \begin{align*}
		%     T_{\varepsilon} &= \left\lceil \frac{2^2 |\mathcal{V}_2|}{\varepsilon^2} \right\rceil \\
		%     T_{\varepsilon} &\leq \frac{16}{\varepsilon^2} + 1 \\
		%     \varepsilon^2 &\leq \frac{16}{T_{\varepsilon} - 1} \\
		%     \varepsilon &\leq \frac{4}{\sqrt{T_{\varepsilon} - 1}}
		% \end{align*}
		% Which gives the desired regret bound for Alice.
	\end{proof}
	We finally provide the claims and their proofs used throughout the main proof.
	\begin{claim}\label{cl:y-minimum}
		In the setting of Theorem~\ref{thm:Alice_safety_payoffs_summary}, if $\overline{U}_t \notin \mathcal{S}$, then $\argmin_{X\in \mathcal{S}}P(\overline{U}_t - X, \overline{U}_t - X)$ exists. Furthermore:        
		\begin{align*}
			\argmin_{X\in \mathcal{S}}P(\overline{U}_t - X, \overline{U}_t - X) = Y_t = \min\{0, \overline{U}_t(n, V)\}   
		\end{align*}
	\end{claim}
	
	\begin{proof}[Proof of the claim]
		% Because all the values are nonpositive for $X(n,V)$, given $Y_t = \min\{0, \overline{U}_t(n, V)\} \in \mathcal{S}$, and for any other $X \in \mathcal{S}$, we have
		
		Let $Y_t = \min\{0, \overline{U}_t(n, V)\} \in \mathcal{S}$. For any $X \in \mathcal{S}$, 
		\begin{align}
			P(\overline{U}_t-X, \overline{U}_t-X) 
			&= 
			\sum_{n=1}^{\infty} \frac{1}{2^n |\mathcal{V}_n|} \sum_{V \in \mathcal{V}_n} \left(\overline{U}_t(n, V) - X(n, V)\right)^2
			\nonumber
			\\
			&= 
			\sum_{n=1}^{\infty} \frac{1}{2^n |\mathcal{V}_n|} \left(\sum_{\substack{V \in V_n \\ \overline{U}_t(n, V) > 0}}\left(\overline{U}_t(n, V) - X(n, V)\right)^2 + \sum_{\substack{V \in V_n \\ \overline{U}_t(n, V) \leq 0}} \left(\overline{U}_t(n, V) - X(n, V)\right)^2\right) \label{ineq:10251639}
		\end{align}
		
		Since $X \in \mathcal{S}$, we have that $X(n,V) \le 0$ for every $n$ and $V$, due to our construction of $\cS$.
		
		Therefore, replacing them with $0$ brings them closer to any positive value and so decreases the first sum of squares. The second sum of squares is nonnegative, so it can be reduced by replacing it with $0$. Applying these simplifications, we obtain
		\begin{align*}
			\sum_{n=1}^{\infty} \frac{1}{2^n |\mathcal{V}_n|}& \left(\sum_{\substack{V \in V_n \\ \overline{U}_t(n, V) > 0}}\left(\overline{U}_t(n, V) - X(n, V)\right)^2 + \sum_{\substack{V \in V_n \\ \overline{U}_t(n, V) \leq 0}} \left(\overline{U}_t(n, V) - X(n, V)\right)^2\right) \\
			\geq& 
			\sum_{n=1}^{\infty} \frac{1}{2^n |\mathcal{V}_n|} \left(\sum_{\substack{V \in V_n \\\overline{U}_t(n, V) > 0}} \left(\overline{U}_t(n, V) - 0\right)^2 + 0\right) 
			% \explain{$X(n,V)$ is nonpositive}
			\\
			=& 
			D(\overline{U}_t - Y_t, \overline{U}_t - Y_t), \explain{Only the $\overline{U}(n, V) > 0$ terms remain}
		\end{align*}
		and it proves the claim.
	\end{proof}
	
	\begin{claim}\label{cl:util-properties}
		In the setting of Theorem~\ref{thm:Alice_safety_payoffs_summary}, the following properties hold when $\overline{U}_t \notin \mathcal{S}$.
		\begin{enumerate}
			\item $P(Y_t, W_t) = 0$
			\item $P(U_{t+1}, W_t) = 0$
			\item $P(\overline{U}_t, W_t) > 0$
			\item $P(X, W_t) \leq 0$ for all $X \in \mathcal{S}$
		\end{enumerate}
	\end{claim}
	\begin{proof}[Proof of the claim]
		All of these properties can be explained by expanding the dot product $P$ and referring to the strategy $S_A$. The first one is the most straightforward: $Y_t$ and $W_t$ are never nonzero in the same coordinate, so $P(Y_t, W_t) = 0$. For the second one, by definition, in round $t+1$ Alice cuts at a point $x$ such that $P(G_x, W_t) = 0$. If Bob selects the left piece, $U_{t+1} = G_x$. If Bob instead selects the right piece:
		\begin{align*}
			U_{t+1}(n, V) &= (1 - V(x)) - \frac{1}{2} \\
			&= \frac{1}{2} - V(x) \\
			&= -G_x(n, V)\,.
		\end{align*}
		So by property 4 of Lemma \ref{any-bob-function-inner-product}, we have $P(U_{t+1}, W_t) = -P(G_x, W_t) = 0$. Because $P$ is linear as per property 5 in Lemma \ref{any-bob-function-inner-product}, any mixed strategy over these outcomes must also satisfy $P(U_{t+1}, W_t) = 0$.
		
		For part (3) of the lemma, expanding the dot product gives:
		\begin{align*}
			P(X, W_t) &= \sum_{n=1}^{\infty} \frac{1}{2^n |\mathcal{V}_n|} \sum_{V \in \mathcal{V}_n} X(n, V) \cdot \max\{0, \overline{U}_t(n, V)\} \\
			&= \sum_{n=1}^{\infty} \frac{1}{2^n|\mathcal{V}_n|} \sum_{V \in \mathcal{V}_n, \overline{U}_t(n, V) > 0} X(n, V) \overline{U}_t(n, V)\,.
		\end{align*}
		As there exists $(n, V)$ such that $\overline{U}_t(n, V) > 0$, this sum is strictly positive when $X=\overline{U}_t(n, V)$, and so $P(\overline{U}_t, W_t) > 0$. Similarly, for $X \in \mathcal{S}$, we have  $X(n, V) \leq 0$ for all $(n, V)$, so $P(X, W_t) \leq 0$.
	\end{proof}
	
	\begin{claim} \label{claim:upper_bound_on_delta_t}
		In the setting of Theorem~\ref{thm:Alice_safety_payoffs_summary}, the sequence $\{\delta_t\}_{t = 1}^{\infty}$ defined in \eqref{def:delta_t} satisfies the inequality $\delta_t \leq 1/t$ for all $t \geq 2$.
	\end{claim}
	\begin{proof}
		
		We will show that our construction of $\delta_t$ satisfies the recursion formula defined in Lemma~\ref{any-bob-delta-limit}.
		% Now the properties of the $\delta_t$ can be established.
		We first focus on the case such that $\delta_t > 0$, which is equivalent to $\overline{U}_t \notin \mathcal{S}$, and then prove that it also holds for $\delta_t = 0$.
		Given that $\delta_t > 0$, we observe that
		
		\begin{align}
			\delta_{t+1} = \inf_{X \in \mathcal{S}} P(\overline{U}_{t+1} - X, \overline{U}_{t+1} - X)
			&\leq P(\overline{U}_{t+1} - Y_t, \overline{U}_{t+1} - Y_t) \,.
			\label{eq:delta_t_plus_1_upper_bound}
		\end{align}
		By rewriting the right hand side of \eqref{eq:delta_t_plus_1_upper_bound}, we obtain 
		\begin{align}
			\delta_{t+1} 
			&\leq  P((\overline{U}_{t+1} - \overline{U}_t) + (\overline{U}_t - Y_t), (\overline{U}_{t+1} - \overline{U}_t) + (\overline{U}_t - Y_t)) \label{eq:delta_t_plus_1_less_than_or_equal_to_P_of_stuff} 
		\end{align}
		Distributing over $P$ in the expression on the right hand side of \eqref{eq:delta_t_plus_1_less_than_or_equal_to_P_of_stuff}, we get that \eqref{eq:delta_t_plus_1_less_than_or_equal_to_P_of_stuff} is equivalent to 
		\begin{align}
			\delta_{t+1}        & \leq   P(\overline{U}_{t+1} - \overline{U}_t, \overline{U}_{t+1} - \overline{U}_t) + 2P(\overline{U}_{t+1} - \overline{U}_t, \overline{U}_t - Y_t) + P(\overline{U}_t - Y_t, \overline{U}_t - Y_t) \notag  \\
			& =  P(\overline{U}_{t+1} - \overline{U}_t, \overline{U}_{t+1} - \overline{U}_t) + 2P(\overline{U}_{t+1} - \overline{U}_t, \overline{U}_t - Y_t) + \delta_t, \label{ineq:0208-delta}
			% \\
			% &= 
			% P(\overline{U}_{t+1} - \overline{U}_t, \overline{U}_{t+1} - \overline{U}_t) + 2P(\overline{U}_{t+1} - \overline{U}_t, \overline{U}_t - Y_t) + P(\overline{U}_t - Y_t, \overline{U}_t - Y_t) \explain{Distributing over $P$} 
			% \\
			% &= 
			% P(\overline{U}_{t+1} - \overline{U}_t, \overline{U}_{t+1} - \overline{U}_t) + 2P(\overline{U}_{t+1} - \overline{U}_t, \overline{U}_t - Y_t) + \delta_t
			% \explain{By Claim~\ref{cl:y-minimum}}
		\end{align}
		where the inequality follows from Claim ~\ref{cl:y-minimum}.
		
		Noting that $\overline{U}_{t+1} - \overline{U}_t = (U_{t+1} - \overline{U}_t)/(t+1)$, we have
		\begin{align*}
			P(\overline{U}_{t+1} - \overline{U}_t, \overline{U}_t - Y_t)
			&= 
			\frac{1}{t+1} P(U_{t+1} - \overline{U}_t, \overline{U}_t - Y_t) 
			\\
			&= 
			\frac{1}{t+1}P((U_{t+1} - Y_t) + (Y_t - \overline{U}_t), \overline{U}_t - Y_t) 
			\\
			&= 
			\frac{1}{t+1}\left(P(U_{t+1} - Y_t, \overline{U}_t - Y_t) + P(Y_t - \overline{U}_t, \overline{U}_t - Y_t)\right). 
		\end{align*}
		By using $\overline{U}_t = W_t + Y_t$, we can expand it by
		\begin{align}
			P(\overline{U}_{t+1} - \overline{U}_t, \overline{U}_t - Y_t) 
			&= 
			\frac{1}{t+1} \left(P(U_{t+1}, W_t) - P(Y_t, W_t) - P(Y_t - \overline{U}_t, \overline{U}_t - Y_t)\right).
			\nonumber
			\\
			&=
			\frac{1}{t+1} \left(P(U_{t+1}, W_t) - P(Y_t, W_t) - \delta_t\right)
			\explain{By Claim~\ref{cl:y-minimum}}
			\\
			&=
			\frac{1}{t+1} \left(0 - 0 - \delta_t\right)
			\explain{By Claim~\ref{cl:util-properties}}
			\\
			&=
			-\frac{1}{t+1} \cdot \delta_t.\label{ineq:0208-negative}
		\end{align}
		
		Further, observe that
		\begin{align}
			P(\overline{U}_{t+1} - \overline{U}_t, \overline{U}_{t+1} - \overline{U}_t) &= \frac{1}{(t+1)^2} P(U_{t+1} - \overline{U}_t, U_{t+1} - \overline{U}_t) \notag
			\\
			&\leq \frac{1}{(t+1)^2}, \label{ineq:11090047}
		\end{align}
		where the inequality follows from Property (3) in Lemma~\ref{any-bob-function-inner-product}.
		
		Putting~\eqref{ineq:0208-delta},~\eqref{ineq:0208-negative} and~\eqref{ineq:11090047} together, we obtain the following inequality for $\delta_t > 0$:
		\begin{align*}
			\delta_{t+1} &\leq P(\overline{U}_{t+1} - \overline{U}_t, \overline{U}_{t+1} - \overline{U}_t) + 2P(\overline{U}_{t+1} - \overline{U}_t, \overline{U}_t - Y_t) + \delta_t \\
			&\leq \frac{1}{(t+1)^2} - \frac{2}{t+1}\delta_t + \delta_t \\
			&= \frac{1}{(t+1)^2} + \left(1-\frac{2}{t+1}\right)\delta_t\,.
		\end{align*}
		Now, we show that the same inequality holds for the case such that $\delta_t =0$ too. Since $\overline{U}_t \in \mathcal{S}$, we obtain
		
		% In fact, it also holds for $\delta_t=0$. This case happens if and only if $\overline{U}_t \in \mathcal{S}$, thus we have
		\begin{align*}
			\delta_{t+1} 
			&= 
			\inf_{X \in \mathcal{S}} P(\overline{U}_{t+1} - X, \overline{U}_{t+1} - X) \\
			&\leq
			P(\overline{U}_{t+1} - \overline{U}_t, \overline{U}_{t+1} - \overline{U}_t) 
			\\
			&\leq 
			\frac{1}{(t+1)^2} 
			\explain{By \eqref{ineq:11090047}}
			\\
			&= 
			\frac{1}{(t+1)^2} + \left(1-\frac{2}{t+1}\right) \delta_t.
		\end{align*}
		By Lemma \ref{any-bob-delta-limit}, we therefore have that $\delta_t \leq 1/t$ for $t \geq 2$.
	\end{proof}

	\subsection{Appendix: Bob enforcing equitable payoffs}
	\label{app:Bob_enforcing_safety_payoffs}
	
	In this section, we prove Theorem \ref{thm:safety_payoffs_summary}, which shows how Bob can enforce equitable payoffs.

	\begin{customthm}{\ref{thm:safety_payoffs_summary}}[Bob enforcing equitable payoffs; formal]
		\phantom{s}
		\phantom{s}
		%Bob has a strategy $S_B$, such that for every Alice strategy $S_A$:
		\begin{itemize}
			\item \emph{In the sequential setting:}  Bob has a pure strategy $S_B$, such that for every Alice strategy $S_A$, on every trajectory of play, Bob's average payoff is at least $1/2 - o(1)$, while Alice's average payoff is at most  $1/2+o(1)$. More precisely,
			\begin{align}
				\frac{u_B}{T} \geq \frac{1}{2} - \frac{1}{\sqrt{T}} \qquad \mbox{and} \qquad  \frac{u_A}{T} \leq \frac{1}{2} + \left(\frac{\Delta}{2\delta}+2\right) \frac{1}{\sqrt{T}}, \notag 
			\end{align}
			recalling that $\delta$ and $\Delta$ are, respectively, the lower and upper bounds on the players' value densities. 
			\item \emph{In the simultaneous setting:}  Bob has a mixed strategy $S_B$, such that for every Alice strategy $S_A$,  both players have average payoff $1/2$ in expectation. % Bob's average payoff is $1/2$ in expectation, while Alice's average payoff is $1/2$ in expectation.
		\end{itemize}  
	\end{customthm}

	\begin{proof}[Proof of Theorem \ref{thm:safety_payoffs_summary}]
		Bob's strategy for the simultaneous setting follows from Proposition \ref{prop:simultaneous-bob-equitable-payoffs}. His strategy for the sequential setting follows from Proposition \ref{bob-fix-half-payoff}.
	\end{proof}
	
	\begin{proposition}
		\label{prop:simultaneous-bob-equitable-payoffs}
		In the simultaneous setting, Bob has a mixed strategy $S_B$ such that, for every Alice strategy $S_A$:
		\[
		\frac{\mathbb{E} [u_A]}{T} = \frac{\mathbb{E} [u_B]}{T} = \frac{1}{2}\,.
		\]
	\end{proposition}
	\begin{proof}
		Bob's strategy is very simple: in each round, randomly pick $L$ or $R$ with equal probability.
		
		To analyze the expected payoffs, consider an arbitrary player $i \in \{A, B\}$ and arbitrary round $t$. Bob is equally likely to pick $L$ or $R$ in round $t$, so each player is equally likely to receive $[0, a_t]$ or $[a_t, 1]$ in round $t$. Therefore, their expected payoff is:
		\begin{align}
			\mathbb{E} [u_i^t] &= \frac{1}{2}V_i([0, a_t]) + \frac{1}{2}V_i([a_t, 1]) \notag \\
			%    &= \frac{1}{2}(V_i([0, a_t]) + V_i([a_t, 1])) \notag \\
			&= \frac{1}{2}V_i([0, 1]) \explain{Since valuations are additive} \\
			&= \frac{1}{2} \,. \label{eq:bob-random-choice-average-payoff-one-half}
		\end{align}
		
		Summing \eqref{eq:bob-random-choice-average-payoff-one-half} over  the $T$ rounds gives the desired expected payoffs for each player.
	\end{proof}
	
	\begin{proposition}
		\label{bob-fix-half-payoff}
		Bob has a pure strategy $S_B$, such that for every Alice strategy  $S_A$, the cumulative utilities in the sequential game are bounded by:
		\begin{align} 
			u_A(S_A, S_B) \leq T/2 + \left(\frac{\Delta}{2\delta}+2\right) \cdot \sqrt{T} \qquad \mbox{and} \qquad u_B(S_A, S_B) \geq  T/2 -   \sqrt{T} \,. \notag 
		\end{align}
	\end{proposition}
	\begin{proof}
		Bob devises his strategy $S_B$ by considering a division of the cake into {$P = \lceil \sqrt{T} \rceil$}
		{consecutive intervals} $I_1, \ldots, I_P$  of equal value to him. That is, Bob chooses points $0 = z_0 \leq z_1 \ldots \leq z_P = 1$ such that  
		\begin{align}
			I_j = 
			\begin{cases} 
				[z_{j-1}, z_{j}) & \text{ if } 1 \leq j \leq P-1; \\
				[z_{P-1}, z_{P}] & \text{ if } j = P;
			\end{cases}
			\qquad
			\mbox{and}
			\qquad V_B(I_j) = 1/P \; \; \forall j \in [P]\,.
		\end{align}
		
		An illustration  of the division into intervals used by Bob can be seen in Figure~\ref{fig:division_by_Bob_into_P_intervals}.
		\begin{figure}[h!] 
			\centering 
			\includegraphics[scale=0.5]{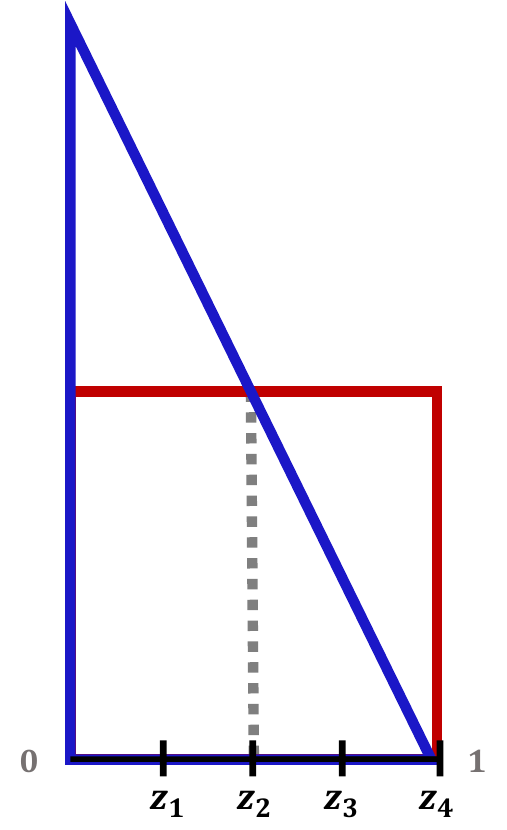}
			\caption{Example of a Bob density and the discretization used by Bob when $T=10$. The number of intervals is $P = \lceil \sqrt{T} \rceil = 4$. The intervals are $I_j = [z_{j-1},z_j)$ for $j \in [3]$ and $I_4 = [z_3, z_4]$, with $V_B(I_j) = 1/4$ $\forall j \in [4]$.}
			\label{fig:division_by_Bob_into_P_intervals}
		\end{figure}
		
		Bob's strategy $S_B$ is defined as follows. Bob keeps a counter $j$   associated with each interval $I_j$, such that the value of the counter at time $t$, denoted   $c_{j,t}$, represents how many times Alice has cut inside the interval $I_j$ in the first $t$ rounds (including round $t$). For each time $t \in [T]$: 
		\begin{itemize}
			\item Let $I_j$ be the interval that contains Alice's cut at time $t$ (that is, $a_t \in I_j$), for $j \in [P]$.
			\item If $c_{j,t}$ is even then Bob plays L;  if $c_{j,t}$ is odd, then Bob plays R. 
		\end{itemize}
		Informally,  Bob alternates between L and R inside each  interval $I_j$. 
		We argue that this Bob strategy ensures his   payoff is at least $1/2 - o(1)$ per round, while Alice cannot get more than $1/2 + o(1)$ per round.

		For each $i \in [P]$ and $j \in [c_{i,T}]$, let $r_{i,j}$ be the first time when the number of cuts in $I_i$ reached $j$. That is,
		\begin{align}
			r_{i,j} = \min \{t \in \mathbb{N} \mid  c_{i,t} = j \; \mbox{and} \; a_t \in I_i \}  \,. \notag 
		\end{align}
		
		%Partition the rounds into pairs 
		By definition of the sequence $\{r_{i,\ell}\}$, for each  $i,j \in \mathbb{N}$ with $2j  \leq c_{i,T}$:
		\begin{itemize}
			\item Alice cut in the interval $I_i$ in both rounds  $r_{i, 2j-1}$ and $r_{i, 2j}$ (meaning  $a_{r_{i, 2j-1}}, a_{r_{i, 2j}} \in I_i$);
			\item  Bob played different actions in the rounds $r_{i, 2j-1}$ and $r_{i, 2j}$.
		\end{itemize}  
		We view rounds $r_{i,2j-1}$ and $r_{i,2j}$ as a pair.   
		For each $i \in [P]$, there is at most one round $r_{i,j}$ that does not have a pair, namely round  $r_{i, c_{i, T}}$: the last round Alice cut in $I_i$. However, this loss only represents at most $P$ rounds in total, which will translate to a sub-linear loss for either Alice's or Bob's utility estimates.
		
		Now we can bound the cumulative utility of each player.

		\begin{description}    
			\item[{Bob's payoff}.]
			% By definition of the sequence $r$, Alice cut in the interval $I_i$ in both of these rounds, that is: $a_{r_{i, 2j-1}}, a_{r_{i, 2j}} \in I_i$.
			For each $i,j \in \mathbb{N}$ with $2j  \leq c_{i,T}$,  Bob's payoff  across the two rounds $r_{i,2j-1}$ and $r_{i,2j}$ is bounded by 
			\begin{align}
				u_B^{r_{i, 2j-1}} + u_B^{r_{i, 2j}} \geq 1 - V_B(I_i) = 1 - \frac{1}{P}\,. \label{eq:Bob_payoff_in_rounds_pair}
			\end{align}

			Since the rounds $r_{i,j}$ with $i \in [P]$ and $2j \leq c_{i,T}$ represent a subset of the total set of rounds $[T]$, Bob's cumulative payoff is at least 
			\begin{align}
				\sum_{t = 1}^{T} u_B^t  & \geq  \sum_{i=1}^{P} \left( \sum_{j \in \mathbb{N}: 2j \leq c_{i,T}} u_B^{r_{i, 2j-1}} + u_B^{r_{i, 2j}} \right) \notag   \\
				& \geq   \sum_{i=1}^{P} \left( \sum_{j \in \mathbb{N}: 2j \leq c_{i,T}} \left( 1 - \frac{1}{P} \right) \right) \explain{By \eqref{eq:Bob_payoff_in_rounds_pair}} \\
				& \geq   \left(1 - \frac{1}{P}\right) \frac{T-P}{2} \explain{{Since the sum is over at least $\frac{T-P}{2}$ pairs of rounds}}
			\end{align}
			
			Since $P = \lceil \sqrt{T} \rceil \geq \sqrt{T}$, we get 
			\begin{align}
				u_B = \sum_{t=1}^T u_B^t & \geq \left(1 - \frac{1}{P}\right) \frac{T-P}{2}  = \frac{T}{2} - \frac{\lceil \sqrt{T} \rceil}{2} - \frac{T}{2\lceil \sqrt{T} \rceil} + \frac{1}{2} \notag \\
				& \geq \frac{T}{2} - \frac{\sqrt{T}+1}{2} - \frac{\sqrt{T}}{2} + \frac{1}{2} \notag \\
				&= \frac{T}{2} - \sqrt{T} \,.
				%\geq \frac{T}{2} - \sqrt{T}\,.
			\end{align}% as desired.
			
			\item[{Alice's payoff.}]
			For each $i,j \in \mathbb{N}$ with $2j  \leq c_{i,T}$, Alice's payoff  across rounds
			$r_{i,2j-1}$ and $r_{i,2j}$ is 
			\begin{align}
				u_A^{r_{i, 2j-1}} + u_A^{r_{i, 2j}} \leq 1 + V_A(I_i) \leq 1 + \frac{\Delta}{\delta} V_B(I_i) = 1 + \frac{\Delta}{\delta} \cdot \frac{1}{P}\,. \label{eq:Alice_payoff_in_rounds_pair}
			\end{align}
			
			There is at most one round without a pair for each interval $I_i$, namely round $r_{i,c_{i,T}}$: the last time Alice cut in $I_i$. For each such round without a pair, we upper bound Alice's payoff by $1$. Then we can upper bound Alice's cumulative utility  by 
			\begin{align}
				\sum_{t = 1}^{T} u_A^t  
				& \leq  
				P + \sum_{i=1}^{P} \left( \sum_{j \in \mathbb{N}: 2j \leq c_{i,T}} u_A^{r_{i, 2j-1}} + u_A^{r_{i, 2j}} \right) \notag   
				\\
				& \leq   
				P + \sum_{i=1}^{P} \left( \sum_{j \in \mathbb{N}: 2j \leq c_{i,T}} \left( 1 + \frac{\Delta}{\delta}\cdot \frac{1}{P} \right) \right) 
				\explain{By \eqref{eq:Alice_payoff_in_rounds_pair}} 
				\\
				& \leq   
				P + \frac{T}{2} \left(1 + \frac{\Delta}{P\delta}\right) \explain{{Since the sum is over at most ${T}/{2}$ pairs}}
				\\
				&= 
				\frac{T}{2} + \frac{T}{P} \cdot \frac{\Delta}{2\delta} + P \notag \\
				&\leq \frac{T}{2} + \sqrt{T} \left(\frac{\Delta}{2\delta} + 2\right) \explain{$P = \lceil \sqrt{T} \rceil \leq 2\sqrt{T}$} \notag
			\end{align}
		\end{description}    
		This completes the proof.
	\end{proof}
	
	%\newpage 

	\section{Appendix: Fictitious play}\label{app:fictitious_play}
	In this section, we analyze the  fictitious play dynamics from Definition~\ref{def:fictitious_play_dynamic} and prove Theorem~\ref{thm_dynamic}.  
	We first prove the main theorem using lemmas that will be presented later, and then prove the required lemmas.
	
	\begin{customthm}{\ref{thm_dynamic}}
		When both Alice and Bob run fictitious play, regardless of tie-breaking rules, their average payoff will converge to $1/2$ at a rate of $O(1/\sqrt{T})$. Formally:
		\begin{align}
			& \left| \frac{u_A}{T} - \frac{1}{2} \right| \leq \frac{2\sqrt{10}}{\sqrt{T}} \qquad \mbox{and} \qquad  \left| \frac{u_B}{T} - \frac{1}{2} \right| \leq \frac{\sqrt{10}}{\sqrt{T}} \qquad  \forall T  \geq 5\,.
		\end{align}
	\end{customthm}
	
	\begin{proof}
		The bounds on Bob's payoff follow immediately from Lemma \ref{lem:fictitious-play-bob-payoff-bound}, which states that
		\begin{align} \label{eq:Bob_payoff_fictitious}
			\frac{T}{2} - \sqrt{10T} \leq \sum_{t=1}^T u_B^t \leq \frac{T}{2} + \sqrt{10T}\,.
		\end{align}
		
		By Lemma \ref{lem:fictitious-play-total-payoff-bound}, 
		\begin{align} \label{eq:welfare_fictitious}
			T - \sqrt{10T} \leq \sum_{t=1}^T \left( u_A^t + u_B^t \right) \leq T + \sqrt{10T}.
		\end{align}
		
		Subtracting equation \eqref{eq:Bob_payoff_fictitious} from \eqref{eq:welfare_fictitious} gives
		\begin{align}
			\frac{T}{2} - 2\sqrt{10T} \leq \sum_{t=1}^T u_A^t \leq \frac{T}{2} + 2\sqrt{10T}\,.
		\end{align}
		Since $u_A = \sum_{t=1}^T u_A^t$ and $u_B = \sum_{t=1}^T u_B^t$,       both players' utilities are bounded as required, which completes the proof. 
	\end{proof}

	To prove the required lemmas, we first introduce several notations.
	For ease for exposition, we often use round $t=0$ as a fake round in which nothing actually happens and all the defined quantities are zero.
	
	\begin{definition}
		For $t \in \{0, 1, \dots, T\}$:
		\begin{itemize}
			\item Let $r_t$ be the number of rounds in which Bob has picked $R$, up to and including round $t$. 
			\item Let $\ell_t$ be the number of rounds in which  Bob has picked $L$ up to and including round $t$.
			\item Let $\alpha_t = r_t - \ell_t$. %We call $\alpha_t$ {\bf \em Alice's potential quantity}.
			\item Let $\beta_t = \sum_{i=1}^t (2V_B([0, a_i]) - 1)$. % We call $\beta_t$  {\bf \em Bob's potential quantity}.
			\item Let $\rho_t = |\alpha_t| + |\beta_t|$, which we call the  {\bf \em radius}.
		\end{itemize}
	\end{definition}
	We remark that in the fake round of $t = 0$, all these variables have the value zero.
	Alice's action under fictitious play entirely depends on the variable $\alpha_t$ while Bob's action entirely depends on  $\beta_t$.
	Overall, we will bound the  payoffs using the growth of the radius $\rho_t$ over $t = 0,1,\ldots, T$. %, and then use the lower and upper bounds of such summation to obtain the final bounds.
	To this end, we introduce two lemmas that will play an essential role throughout the proof.
	
	First, the following lemma formally argues that the action taken by each player is guided by their corresponding variable $\alpha$ and $\beta_t$, respectively. 
	%given the corresponding potential quantity.
	
	\begin{lemma}
		\label{lem:fictitious-basic-dynamics}
		Let $t \in \{0, 1, \ldots, T-1\}$. Then the following hold for round $t+1$:
		\begin{itemize}
			\item Alice's action with respect to $\alpha_t$:
			\begin{itemize}
				\item If $\alpha_t > 0$, Alice will cut at 1.
				\item If $\alpha_t < 0$, Alice will cut at 0.
				\item If $\alpha_t = 0$, any action would incur the same payoff, so she could cut anywhere.
			\end{itemize}
			\item Bob's action with respect to $\beta_t$:
			\begin{itemize}
				\item If $\beta_t > 0$, Bob will pick $L$.
				\item If $\beta_t < 0$, Bob will pick $R$.
				\item If $\beta_t = 0$, any action would incur the same payoff, so he could pick either $L$ or $R$.
			\end{itemize}
		\end{itemize}
	\end{lemma}
	
	\begin{proof}
		First, consider Alice's choice. For $t = 0$, there is no history, so any choice has the same value to her; accordingly, $\alpha_0 = 0$. 
		For $t \geq 1$, the expected value she assigns to any particular cut location $x$ is:
		\begin{align}
			\frac{1}{t}\left(r_t V_A([0, x]) + \ell_t V_A([x, 1])\right) &= \frac{1}{t}\left(r_t V_A([0, x]) + \ell_t (1-V_A([0, x])\right) \notag \\
			&= \frac{r_t - \ell_t}{t} V_A([0, x]) + \frac{\ell_t}{t} \notag \\
			\label{eq:fictitious-play-alice-choice-value-based-on-alpha}
			&= \frac{\alpha_t}{t} V_A([0, x]) + \frac{\ell_t}{t}\,.
		\end{align}
		
		From \eqref{eq:fictitious-play-alice-choice-value-based-on-alpha}, we get that Alice's cut decision is entirely based on $\alpha_t$. If $\alpha_t < 0$, she will minimize $V_A([0, x])$, which means she cuts at 0. If $\alpha_t > 0$, she will maximize $V_A([0, x])$, which means she cuts at 1. If $\alpha_t = 0$, then her choice of $x$ doesn't affect her expected value.
		
		Now consider Bob's choice. For $t = 0$, there is no history, so Bob is indifferent between $L$ and $R$; accordingly, $\beta_0 = 0$. For $t \geq 1$, the expected value he assigns to choosing $L$ is:
		\begin{align} \label{eq:e_ell}
			E_L =   \frac{1}{t}\sum_{i=1}^t V_B([0, a_i])\,.
		\end{align}
		The expected value he assigns to choosing $R$ is:
		\begin{align} \label{eq:e_right}
			E_R =   \frac{1}{t} \sum_{i=1}^t V_B([a_i, 1])\,.
		\end{align}
		
		Combining \eqref{eq:e_ell} and \eqref{eq:e_right},  the difference between his expected value for choosing $L$ and $R$ is:
		\begin{align}
			E_L - E_R &=  \left(\frac{1}{t}\sum_{i=1}^t V_B([0, a_i])\right) - \left(\frac{1}{t} \sum_{i=1}^t V_B([a_i, 1])\right) = \frac{1}{t} \sum_{i=1}^t (2V_B([0, a_i]) - 1) =  \frac{\beta_t}{t} \,.
			\label{eq:fictitious-play-bob-choice-value-based-on-beta}
		\end{align}
		
		From \eqref{eq:fictitious-play-bob-choice-value-based-on-beta}, we get that Bob's decision is entirely based on $\beta_t$. If $\beta_t < 0$, he values $R$ more than $L$, so he picks $R$. If $\beta_t > 0$, he values $L$ more than $R$, so he picks $L$. If $\beta_t = 0$, he values each equally.
	\end{proof}
	
	The following lemma further describes the evolution of $\alpha_t$ and $\beta_t$ given the actions taken by the players.

	\begin{lemma}
		\label{lem:fictitious-play-actions-affect-alpha-beta}
		For each round $t \in [T]$: 
		\begin{itemize}
			\item Alice's action affects $\beta_t$ as follows:
			\begin{itemize}
				\item If Alice cuts at $0$ in round $t$, then $\beta_t = \beta_{t-1} - 1$.
				\item If Alice cuts at $1$ in round $t$, then $\beta_t = \beta_{t-1} + 1$.
				\item If she cuts at $x \in (0,1)$, then $|\beta_t - \beta_{t-1}| < 1$.
			\end{itemize}
			\item Bob's action affects $\alpha_t$ as follows:
			\begin{itemize}
				\item If Bob picks $L$ in round $t$, then $\alpha_t = \alpha_{t-1} - 1$.
				\item If Bob picks $R$ in round $t$, then $\alpha_t = \alpha_{t-1} + 1$.
			\end{itemize}
		\end{itemize}
	\end{lemma}
	
	\begin{proof}
		First, we consider the impact of Alice's cut point $a_t$ on $\beta_t$. Explicitly writing out the difference $\beta_t - \beta_{t-1}$ gives:
		\begin{align}
			\beta_t - \beta_{t-1} &= \sum_{i=1}^t (2V_B([0, a_i]) - 1) - \sum_{i=1}^{t-1} (2V_B([0, a_i]) - 1) = 2V_B([0, a_t]) - 1\,. \label{eq:fictitious-beta-change-equals-two-vb-minus-one}
		\end{align}
		
		If $a_t = 0$, then $V_B([0, a_t]) = V_B([0, 0]) = 0$, so by \eqref{eq:fictitious-beta-change-equals-two-vb-minus-one} we have $\beta_t - \beta_{t-1} = -1$ as desired.
		
		If $a_t = 1$, then $V_B([0, a_t]) = V_B([0, 1]) = 1$, so by \eqref{eq:fictitious-beta-change-equals-two-vb-minus-one} we have $\beta_t - \beta_{t-1} = 1$.
		
		If $a_t \in (0, 1)$, then there is at least some cake on each side of $a_t$, so we have $V_B([0, a_t]) \in (0, 1)$. By \eqref{eq:fictitious-beta-change-equals-two-vb-minus-one}, we have $\beta_t - \beta_{t-1} \in (-1, 1)$, so $|\beta_t - \beta_{t-1}| < 1$ as stated by the lemma.
		
		Second, we consider the impact of Bob's choice on $\alpha_t$. Explicitly writing out the difference $\alpha_t - \alpha_{t-1}$, we obtain the following:
		\begin{align}
			\alpha_t - \alpha_{t-1} &= (r_t - \ell_t) - (r_{t-1} - \ell_{t-1}) = (r_t - r_{t-1}) - (\ell_t - \ell_{t-1}) = \begin{cases} 1 & b_t = R \\ -1 & b_t = L \end{cases}
		\end{align}
		
		This completes the proof.
	\end{proof}
	
	Let us  elaborate more the dynamics of each player based on Lemma~\ref{lem:fictitious-basic-dynamics} and~\ref{lem:fictitious-play-actions-affect-alpha-beta}.
	If either of $\alpha_t$ or $\beta_t$ is exactly 0, the corresponding player will use their tie-breaking rules. Ignoring these cases, these choices lead to movement through $(\alpha_t, \beta_t)$ space that spirals counter-clockwise around the origin at exactly $45$ degree angles. 
	Figure~\ref{fig:fictitious-dynamic}-(a) describes the overall dynamics of the variables $(\alpha_t, \beta_t)$.
	
	\begin{figure}[h!]
		\centering
		\subfigure[Overall dynamics]{
			\includegraphics[scale=0.6]{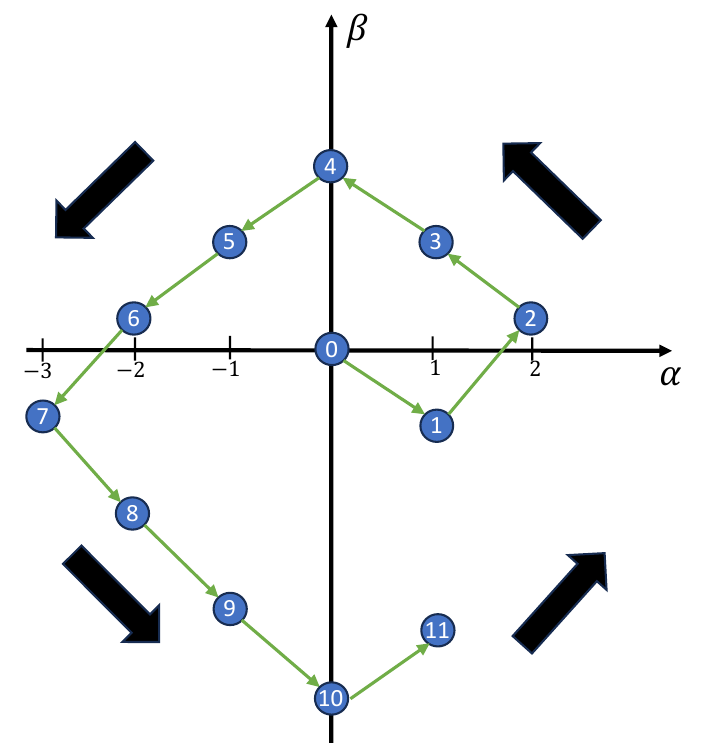}
		}\qquad  
		\subfigure[Reflected dynamics by Lemma~\ref{lem:fictitious-symmetry}]{
			\includegraphics[scale=0.6]{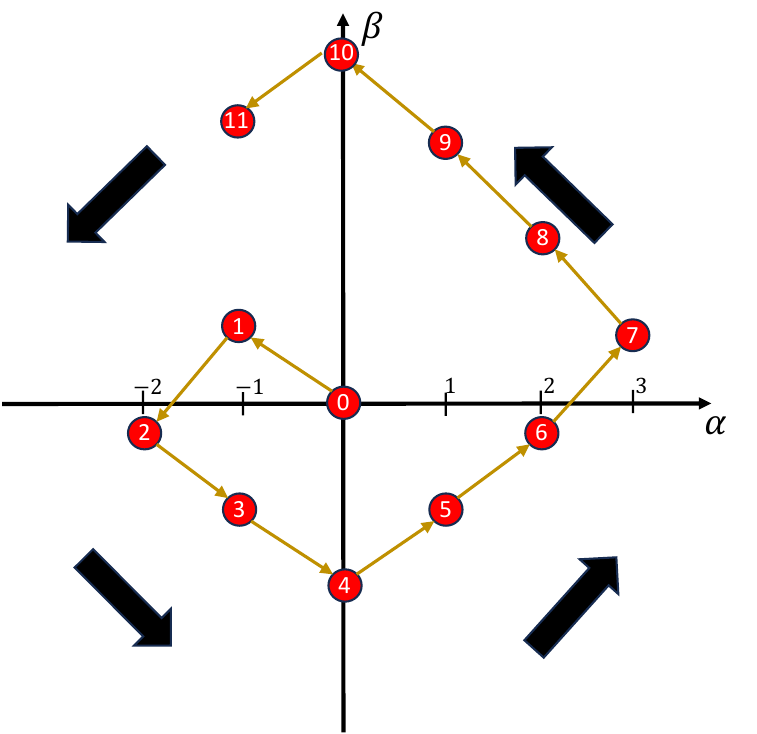}
		}
		% \centering 
		% \includegraphics[scale=0.85]{fictitious-dynamic.pdf}
		\caption{
			Figure (a) represents the overall illustration of the dynamics of the action quantities $\alpha_t$ and $\beta_t$ for $t=0,1,\ldots, T$. The $x$-axis denotes the quantity of $\alpha_t$ and the $y$-axis denotes that of $\beta_t$. Blue circles represent the sequence of points in the plane, where the number inside the circle denotes the index $t$.
			Note that $\alpha_t$ only takes integer values, while $\beta_t$ possibly takes noninteger values at some rounds.
			Figure (b) also depicts an overall dynamics implemented by another pair of tiebreaking rules guaranteed by Lemma~\ref{lem:fictitious-symmetry}. Note that each point is reflected with respect to the origin point. Importantly, Lemma~\ref{lem:fictitious-symmetry} guarantees that $\rho_t$ and $u_B^t$ for $t=0,1,\ldots, T$ remain exactly the same for both dynamics.
		}\label{fig:fictitious-dynamic}
	\end{figure}

	% \includegraphics[width=0.5\textwidth]{fictitious-dynamic.pdf}
	
	% \includegraphics[width=0.3\textwidth]{FictitiousPlayPhaseSpace.png}
	
	% To cut out many cases from the following arguments, we can use the fact that the fictitious play dynamics are symmetric under left-right reversal. This symmetry manifests itself as the following lemma.
	
	The following lemma shows a symmetry that will help reduce the number of cases in the subsequent analysis. 
	%To cut down on the number of cases in the analysis, we will use the following lemma, which shows that without loss of generality we can focus on certain ranges of $\alpha_t$ and $\beta_t$ to analyze the dynamics of $\rho_t$ and $u_B$. 
	Figure~\ref{fig:fictitious-dynamic}-(b) depicts the rotational symmetry of fictitious play dynamics in the $\alpha$-$\beta$-plane shown by the lemma.
	Specifically, the symmetry will allow us to assume $\alpha_t \geq 0$ without loss of generality when analyzing $\rho_t$ and $u_B^t$.
	% which formalizes the 180-degree rotational symmetry of fictitious play dynamics in the $\alpha$-$\beta$ plane.
	\begin{lemma}\label{lem:fictitious-symmetry}
		Consider an arbitrary pair of tie-breaking rules for both players. Consider the resulting sequence for the variables $\alpha_t, \beta_t$ and $u_B^t$  for $t =0,\ldots, T$.
		Then, there exists another choice of tie-breaking rules that would result in the sequence of variables $\tilde{\alpha}_t$, $\tilde{\beta}_t$, and $\tilde{u}_B^t$ for $t =0,\ldots, T$ such that
		% Suppose that, following some tiebreak rules, fictitious play dynamics trace out some particular sequence of $\alpha_t$, $\beta_t$, and $\gamma_B(t)$. Then there exists another choice of tiebreak rules that would result in $\alpha'_t$, $\beta'_t$, and $\gamma'_B(t)$ such that, for all $t \in \{0, \dots, T\}$:  
		\begin{align} \label{eq:alpha_tilde_beta_tilde_u_B_tilde}
			\tilde{\alpha}_t = -\alpha_t; \; \;  
			\tilde{\beta}_t = -\beta_t; \; \;  
			\tilde{u}_B^t = u_B^t \,.
		\end{align}
	\end{lemma}
	
	\begin{proof}
		The proof will proceed by induction.
		
		\paragraph{Base case} All of $\alpha_0$, $\tilde{\alpha}_0$, $\beta_0$, $\tilde{\beta}_0$, $u_B^0$, and $\tilde{u}_B^0$ are zero, so  \eqref{eq:alpha_tilde_beta_tilde_u_B_tilde}  trivially holds.
		
		\paragraph{Inductive hypothesis} Assume that \eqref{eq:alpha_tilde_beta_tilde_u_B_tilde} holds  for some $t \geq 0$. 
		
		\paragraph{Inductive step} We show that \eqref{eq:alpha_tilde_beta_tilde_u_B_tilde} holds for $t+1$. Suppose that, under the original tie-breaking rules, Alice cut at $a_{t+1}$ and Bob picked $b_{t+1} \in \{L, R\}$ in round $t+1$.
		First, we analyze for Alice. Let $a_{t+1}' \in [0,1]$ be the unique point for which 
		\begin{align}
			V_B([0, a_{t+1}]) = V_B([a'_{t+1}, 1])\,.
		\end{align}
		The point $a_{t+1}'$ is uniquely defined since Bob's density is strictly positive. 
		
		We show that there exists another choice of  tie-breaking rules under which, starting from $\tilde{\alpha}_t$ and $\tilde{\beta}_t$, Alice cuts at $a'_{t+1}$ in round $t+1$. We split into cases based on $\alpha_t$:
		\begin{description}
			\item[$\; \; \; (\alpha_t > 0)$:] Then $a_{t+1} = 1$ by Lemma \ref{lem:fictitious-basic-dynamics}. Since $V_B([0, a_{t+1}]) = V_B([0, 1])$, we have $a'_{t+1} = 0$ by definition. By the inductive hypothesis, we have $\tilde{\alpha}_t = -\alpha_t$, and so $\tilde{\alpha}_t < 0$. Then, regardless of tie-breaking rules, Alice cuts at $0 = a'_{t+1}$.
			\item[$\; \; \;  (\alpha_t < 0)$:] Then $a_{t+1} = 0$ by Lemma~\ref{lem:fictitious-basic-dynamics}. Since $V_B([0, a_{t+1}] = 0 = V_B([1, 1])$, we have $a'_{t+1} = 1$ by definition. By the inductive hypothesis, $\tilde{\alpha}_t = -\alpha_t$, and so $\tilde{\alpha}_t > 0$. Then, regardless of tie-breaking rules, Alice cuts at $1 = a'_{t+1}$.
			\item[$\; \; \;  (\alpha_t = 0)$:] Then Alice can break ties any way she likes by Lemma~\ref{lem:fictitious-basic-dynamics}. However, $\tilde{\alpha}_t = -\alpha_t = 0$, so  any cut point is a valid choice for her new tie-breaking rule. In particular, she can cut at $a'_{t+1}$.
		\end{description}
		
		Second, we analyze for Bob. Let $b'_{t+1} \in \{L, R\}$ be such that $b'_{t+1} \neq b_{t+1}$. We show that there exist tie-breaking rules under which, starting from $\tilde{\alpha}_t$ and $\tilde{\beta}_t$, Bob chooses $b'_{t+1}$. We split into cases based on $\beta_t$:
		\begin{description}
			\item[$\; \; \;  (\beta_t > 0)$:] Then $b_{t+1} = L$ by Lemma~\ref{lem:fictitious-basic-dynamics}; therefore, $b'_{t+1} = R$. By the inductive hypothesis, $\tilde{\beta}_t = -\beta_t$, so $\tilde{\beta}_t  < 0$. Then, regardless of tie-breaking rules, Bob picks $R = b'_{t+1}$.
			\item[$\; \; \;  (\beta_t < 0)$:] Then $b_{t+1} = R$ by Lemma~\ref{lem:fictitious-basic-dynamics}; therefore, $b'_{t+1} = L$. By the inductive hypothesis, $\tilde{\beta}_t = -\beta_t$, so $\tilde{\beta}_t > 0$. Then, regardless of tie-breaking rules, Bob picks $L = b'_{t+1}$.
			\item[$\; \; \;  (\beta_t = 0)$:] Then Bob can break ties any way he likes by Lemma~\ref{lem:fictitious-basic-dynamics}. However, $\tilde{\beta}_t = -\beta_t = 0$, so either $L$ or $R$ is a valid choice for his new tie-breaking rule. In particular, he can choose $b'_{t+1}$.
		\end{description}
		
		Finally, we show that these opposite choices have exactly the desired effect on $\tilde{\alpha}_{t+1}$, $\tilde{\beta}_{t+1}$, and $\tilde{u}_B^{t+1}$. Covering each in turn:
		\begin{itemize}
			\item Since Bob picks the opposite side under the trajectory associated with $\tilde{\alpha}$ and $\tilde{\beta}$, the change from $\alpha_t$ to $\alpha_{t+1}$ is in the opposite direction as the change from $\tilde{\alpha}_t$ to $\tilde{\alpha}_{t+1}$ by Lemma \ref{lem:fictitious-play-actions-affect-alpha-beta}. By the inductive hypothesis, we have $\tilde{\alpha}_t = -\alpha_t$, and so  $\tilde{\alpha}_{t+1} = -\alpha_{t+1}$.
			\item Since Alice picks the mirror image of her cut point under the trajectory associated with $\tilde{\alpha}$ and $\tilde{\beta}$, the change from $\beta_t$ to $\beta_{t+1}$ is exactly opposite to the change from $\tilde{\beta}_t$ to $\tilde{\beta}_{t+1}$. Specifically,
			\begin{align}
				\beta_{t+1} = \beta_t + \Bigl(  2V_B([0, a_{t+1}]) - 1\Bigr), \label{eq:beta_change_diff_trajectory}
			\end{align}
			while 
			\begin{align}
				\tilde{\beta}_{t+1}  & = \tilde{\beta}_t + \Bigl(     2V_B([0, a'_{t+1}]) - 1 \Bigr) \notag \\
				& = \tilde{\beta}_t + 2(1-V_B([a'_{t+1}, 1])) - 1  \notag \\
				& = \tilde{\beta}_t + 2(1-V_B([0, a_{t+1}])) - 1  \notag \\ 
				& = \tilde{\beta}_t - \Bigl(2V_B([0, a_{t+1}])-1\Bigr)\,. \label{eq:tilde_beta_change_diff_trajectory}
			\end{align}
			
			By the inductive hypothesis, we have $\tilde{\beta}_t = -\beta_t$. Using  equations \eqref{eq:beta_change_diff_trajectory} and \eqref{eq:tilde_beta_change_diff_trajectory}, we obtain $\tilde{\beta}_{t+1} = -\beta_{t+1}$.
			\item Under Alice's new cut point $a'_{t+1}$, Bob's valuation of the left and right sides of the cake swap, \ie $V_B([0,a'_{t+1}]) = V_B([a_{t+1},1])$. But he also chooses the opposite side, so he gets exactly the same payoff as under the original tie-breaking rules in round $t+1$. Therefore, $\tilde{u}_B^{t+1} = u_B^{t+1}$.
		\end{itemize}
		By induction, the claim holds for all $t$, which completes the proof.
	\end{proof}

	\begin{figure}[h!]
		\centering
		\subfigure[Axis-crossing rounds]{
			\includegraphics[scale=0.6]{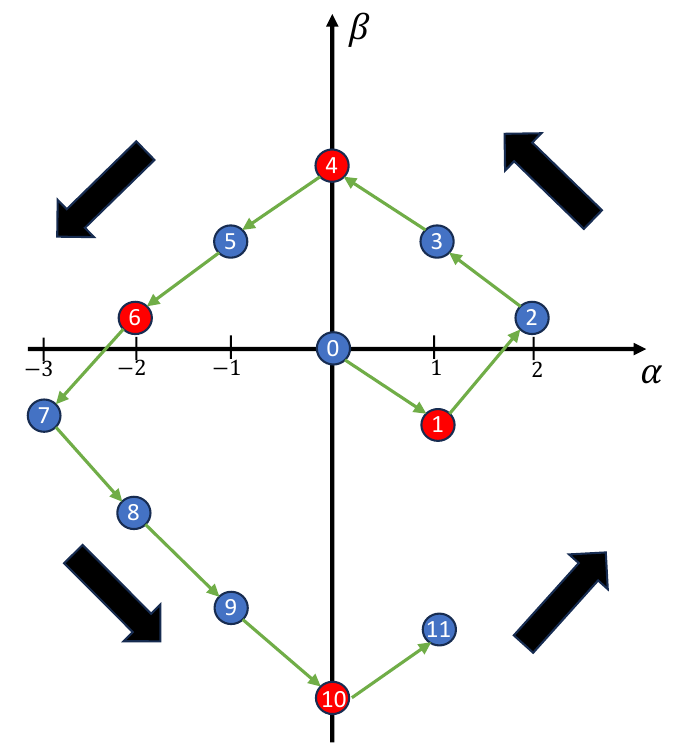}
		}\qquad
		\subfigure[Non-decreasing radius $\rho_t$ over rounds]{
			\includegraphics[scale=0.6]{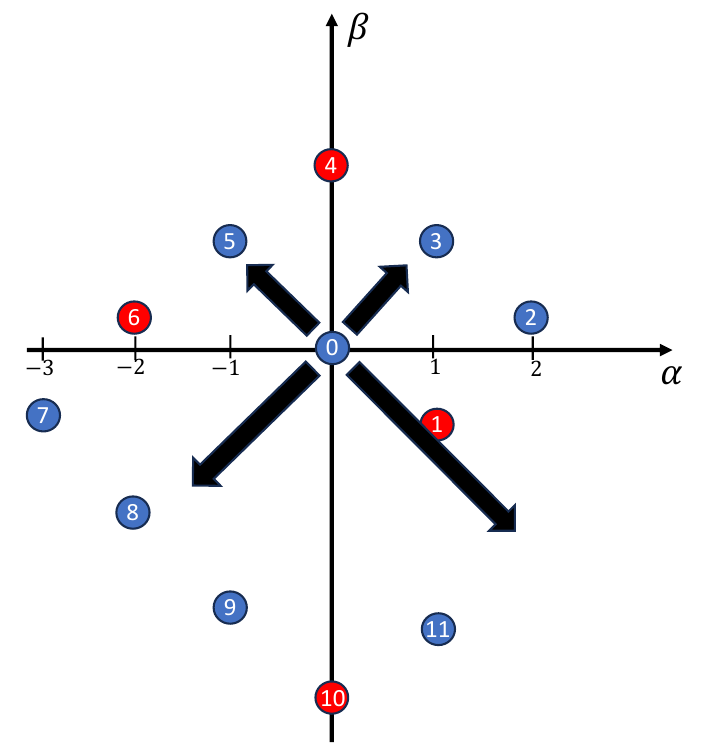}
		}
		% \centering 
		% \includegraphics[scale=0.85]{fictitious-dynamic.pdf}
		\caption{
			Overall dynamics of $\alpha_t$ and $\beta_t$ with non axis-crossing rounds (\textcolor{blue}{blue} circles) and axis-crossing rounds (\textcolor{red}{red} circles).
			Figure (b) shows that the radius $\rho_t = |\alpha_t| + |\beta_t|$ is nondecreasing in $t$ as shown by Lemma~\ref{lem:fictitions-phase-space-radius-nondecreasing}.
			In particular, $\rho_t$ remains the same for non axis-crossing rounds  but possibly increases for axis-crossing rounds.
		}
		\label{fig:fictitious-dynamic-axis-rho}
	\end{figure}
	
	In addition, it is helpful to distinguish between the rounds that cross an axis in the $\alpha$-$\beta$ plane and those that do not.
	The following formalizes the definition of such rounds, which are depicted in Figure~\ref{fig:fictitious-dynamic-axis-rho}-(a).
	
	\begin{definition}
		An \emph{axis-crossing round} is a round $t$ where at least one of the following occurs:
		\begin{itemize}
			\item $\alpha_t = 0$
			\item $\beta_{t+1} > 0$, but $\beta_t \leq 0$
			\item $\beta_{t+1} < 0$, but $\beta_t \geq 0$.
		\end{itemize}
	\end{definition}
	
	Importantly, we will show that $\rho_t$ can strictly increase only if the current round is axis-crossing, while it is non-decreasing over the entire game.
	The following lemma formalizes this observation. We  provide an example in Figure~\ref{fig:fictitious-dynamic-axis-rho}-(b).
	\begin{lemma}
		\label{lem:fictitions-phase-space-radius-nondecreasing}
		Let $t \in \{0, 1,\ldots, T\}$. The radius $\rho_t$ satisfies  the following properties:
		\begin{itemize}
			\item[(a)] $\rho_t = \rho_{t+1}$ if $t$ is   not axis-crossing
			\item[(b)] $\rho_t \leq \rho_{t+1} \leq 2+ \rho_t$ if  $t$ is axis-crossing
			\item[(c)] $\rho_0 = 0$
			\item[(d)] $\rho_t \geq 1$ for $t \geq 1$.
		\end{itemize}
		In particular, the radius $\rho_t$ is non-decreasing in $t$.
	\end{lemma}
	\begin{proof}
		First, we will show (a) and (b). Consider an arbitrary round $t \in \{0, \ldots, T\}$. By Lemma \ref{lem:fictitious-play-actions-affect-alpha-beta}, we have $|\alpha_{t+1} - \alpha_t| \leq 1$ and $|\beta_{t+1} - \beta_t| \leq 1$, so by the triangle inequality
		\begin{align}
			\Bigl|\rho_{t+1} - \rho_t\Bigr| = \Bigl| |\alpha_{t+1}| + |\beta_{t+1}|  - |\alpha_{t}| - |\beta_{t}| \Bigr| \leq 2\,.
		\end{align} Therefore, $\rho_{t+1} \leq 2+ \rho_t$, as required by (b). Thus for (a) and (b) it remains  to show that $\rho_{t+1} \geq \rho_t$ $\forall t \in \{0, \ldots, T\}$ and that  $\rho_{t+1} = \rho_t$ if $t$ is not axis-crossing.
		
		%We break into cases based on $\alpha_t$ and $\beta_t$. 
		
		By Lemma~\ref{lem:fictitious-symmetry}, it suffices to consider $\alpha_t \geq 0$. 
		More precisely, this is because if $\alpha_t \le 0$, then there exists a tie-breaking rule with associated $\tilde{\alpha}_t$ satisfying $\tilde{\alpha}_t  = -\alpha_t$ for every $t=0,1,\ldots, T$, so $\tilde{\alpha}_t \geq 0$. % by Lemma~\ref{lem:fictitious-symmetry}.
		Crucially, the lemma ensures  the sequences $\tilde{\alpha}_t$ and $\alpha_t$ have the same radius $\rho_t$.
		%we have $\rho_t = \tilde{\rho}_t$, where $\tilde{\rho}_t$ is the radius associated with the sequence $\tilde{\alpha}_t$.
		%Thus it suffices to prove that $\rho_t$ is non-decreasing in $t$ given $\tilde{\alpha}_t = -\alpha_t$.\footnote{Similar arguments will be repeatedly used in the rest of the proof.}
		
		\medskip 
		
		We consider a few cases based on $\alpha_t$ and $\beta_t$, considering only $\alpha_t \geq 0$. 
		Since $\alpha_t$ is an integer, also divide into   $\alpha_t \geq 1$ and   $\alpha_t = 0$:
		% we consider the following cases of either $\alpha_t \ge 1$ or $\alpha_t = 0$ and $\beta_t \ge 0$.
		% Note again that every other case is a mirror image of these due to Lemma~\ref{lem:fictitious-symmetry}.
		% \sscomment{how about $\alpha_t \ge 0$ but $\beta_t < 0$? Doesn't Lemma 8 guarantees the mirroring of both $\alpha$ and $\beta$ simultaneously?}
		
		\begin{description}
			\item[$\;\;$ $(\alpha_t \geq 1$ and $\beta_t > 0)$:] By Lemma \ref{lem:fictitious-basic-dynamics} Alice will cut at 1 and Bob will pick L. Then by Lemma \ref{lem:fictitious-play-actions-affect-alpha-beta}, we have $\alpha_{t+1} = \alpha_t - 1$ and $\beta_{t+1} = \beta_t + 1$. Therefore:
			\begin{align}
				\rho_{t+1} = |\alpha_{t+1}| + |\beta_{t+1}| & = \alpha_t - 1 + \beta_t + 1 \explain{Since $\alpha_{t+1} \geq 0$} \\
				&= |\alpha_t| + |\beta_t| \explain{Because $\alpha_t, \beta_t \geq 0$} \\
				&= \rho_t\,.
			\end{align}
			%        which proves (a) and (b) in this case.
			% , regardless of whether $t$ is axis-crossing. \rpcomment{In fact $t$ is not axis-crossing, but proving that would require another sentence or so.}
			\item[$\;\;$ $(\alpha_t \geq 1$ and $\beta_t = 0)$:] Then Alice will cut at 1 and Bob will pick whichever piece he likes. Therefore  $\alpha_{t+1} = \alpha_t \pm 1$ and $\beta_{t+1} = \beta_t + 1 = 1$. Since $\beta_t \leq 0$ but $\beta_{t+1} > 0$, round $t$ is axis-crossing. What remains  is to show that $\rho_{t+1} \geq \rho_t$ in this case:
			\begin{align}
				\rho_{t+1} = |\alpha_{t+1}| + |\beta_{t+1}| & \geq |\alpha_t| - 1 + 1 \explain{Because $|x \pm y| \geq |x| - |y|$ for all $x, y$} \\
				&= |\alpha_t| + |\beta_t| \explain{Since $\beta_t = 0$} \\
				&= \rho_t \,.
			\end{align}
			\item[$\;\;$ $(\alpha_t \geq 1$ and $-1 < \beta_t < 0)$:] Then Alice will cut at 1 and Bob will pick R. Therefore, $\alpha_{t+1} = \alpha_t + 1$ and $\beta_{t+1} = \beta_t + 1$. Since $\beta_t < 0$ but $\beta_{t+1} > 0$, round $t$ is axis-crossing. What remains  is to show that 
			$\rho_{t+1} \geq \rho_t$:
			\begin{align}
				\rho_{t+1} = |\alpha_{t+1}| + |\beta_{t+1}| &= \alpha_t + 1 + |\beta_t + 1| \notag \\
				&\geq \alpha_t + 1 + |\beta_t| - 1 \explain{Since $|x+y| \geq |x|-|y|$ for all $x, y$} \\
				&= \rho_t \,.
			\end{align}
			
			\item[$\;\;$ $(\alpha_t \geq 1$ and $\beta_t \leq -1)$:] Then Alice will cut at 1 and Bob will pick R. Therefore, $\alpha_{t+1} = \alpha_t + 1$ and $\beta_{t+1} = \beta_t + 1$. In this case, $\rho_{t+1} = \rho_t$:
			\begin{align}
				\rho_{t+1} = |\alpha_{t+1}| + |\beta_{t+1}| 
				&= \alpha_t + 1 - (\beta_t + 1) \explain{Because $\beta_t + 1 \leq 0$} \\
				&= |\alpha_t| + |\beta_t| \explain{Since $\beta_t \leq 0$} \\
				&= \rho_t \,.
			\end{align}
			
			\item[$\;\;$ $(\alpha_t = 0$ and $\beta_t \geq 0)$:] Then $t$ is an axis-crossing round. Bob could pick either L or R, but either way $|\alpha_{t+1}| = 1$. Alice could cut anywhere, but since $|\beta_{t+1} - \beta_t| \leq 1$ we can still conclude $|\beta_t| - |\beta_{t+1}| \leq 1$. Therefore:
			\begin{align}
				\rho_{t+1} = |\alpha_{t+1}| + |\beta_{t+1}| & \geq 1 + |\beta_t| - 1 \notag \\
				&= |\alpha_t| + |\beta_t| \explain{Since $\alpha_t = 0$} \\
				&= \rho_t \,.
			\end{align}
			\item[$\;\;$ $(\alpha_t = 0$ and $\beta_t < 0)$:] Then we can use the symmetry of Lemma \ref{lem:fictitious-symmetry} to consider $\beta_t > 0$ instead, which has already been covered.
		\end{description}
		
		In all cases, properties (a) and (b) must hold.
		
		\medskip 
		
		Now we show that the properties (c) and (d) hold. 
		Property (c) follows  from $\alpha_0 = \beta_0 = 0$. Because $\rho_{t+1} \geq \rho_t$ for all $t$, property  (d) would follow from showing $\rho_1 \geq 1$, which can be seen true from the following inequality:
		\begin{align}
			\rho_1 = |\alpha_1| + |\beta_1| 
			&= 1 + |\beta_1| \explain{Since $\alpha_0 = 0$,  so $\alpha_1 = \pm 1$} \\
			&\geq 1 \,.
		\end{align}
		This finishes the proof.
	\end{proof}
	
	\begin{lemma}
		\label{lem:fictitious-play-axis-crossing-spacing}
		Suppose $t-1$ is an axis-crossing round and $\tau > t-1$ is the next axis-crossing round after $t-1$. 
		Then, there exists at least $\rho_t -2$  and at most $\rho_t$ rounds between them, \ie
		\begin{align}\label{ineq:02070208}
			\rho_t - 2 \leq \tau - t \leq \rho_t\,.
		\end{align}
	\end{lemma}
	
	\begin{proof}
		By Lemma \ref{lem:fictitious-symmetry}, we can assume $\alpha_{t-1} \geq 0$. Then it suffices to  consider only four types that the axis-crossing round $t-1$ could have: 
		\begin{enumerate}[(i)]
			\item  $\alpha_{t-1} \geq 1$, $\beta_{t-1} \leq 0$, and $\beta_t > 0$; 
			\item $\alpha_{t-1} \geq 1$, $\beta_{t-1} \geq 0$, and $\beta_t < 0$; 
			\item $\alpha_{t-1} = 0$ and $\beta_{t-1} \geq 1$;  
			\item  $\alpha_{t-1} =  0$ and $0 \leq \beta_{t-1} < 1$. 
		\end{enumerate}
		We show separately for each of the types \textbf{(i)}-\textbf{(iv)}.
		\begin{description}
			\item[\bf  (i) $\alpha_{t-1} \geq 1$, $\beta_{t-1} \leq 0$, and $\beta_t > 0$.]
			Because $\alpha_{t-1} > 0$, Alice will cut at 1 in round $t$ by Lemma~\ref{lem:fictitious-basic-dynamics}, so $\beta_t = \beta_{t-1} + 1$. Since $\beta_{t-1}$ and $\beta_t$ have different signs, it must be that $-1 < \beta_{t-1} \leq 0$ and $0 < \beta_t \leq 1$. 
			As long as $\alpha$ remains positive, Alice will keep cutting at 1 and increasing $\beta$ by Lemma~\ref{lem:fictitious-basic-dynamics} and Lemma~\ref{lem:fictitious-play-actions-affect-alpha-beta}, so the next axis-crossing round $\tau$ cannot be one where $\beta$ changes sign.
			Thus, it must be the one satisfying $\alpha_{\tau} = 0$. Until then, Bob will keep picking L, so $\alpha$ will decrease by 1 every round. Therefore, this implies that $\tau = t + \alpha_t$. To show $\tau - t \leq \rho_t$ as required by ~\eqref{ineq:02070208}, we have
			\begin{align*}
				\tau - t = \alpha_t 
				\leq |\alpha_t| + |\beta_t| \le \rho_t.
			\end{align*}
			
			To prove $\rho_t -2 \le \tau - t$, observe that $\alpha_t \geq \alpha_{t-1} - 1$. Since $\alpha_{t-1} \geq 1$, we have $\alpha_t \geq 0$. Thus 
			\begin{align}
				\tau - t &= \alpha_t \notag \\
				&\geq |\alpha_t| + |\beta_t| - 1 \explain{Since $0 < \beta_t \leq 1$ and $\alpha_t \geq 0$} \\
				& = \rho_t - 1 > \rho_t - 2\,.
			\end{align}
			
			\item[\bf (ii) $\alpha_{t-1} \geq 1$, $\beta_{t-1} \geq 0$, and $\beta_t < 0$.]
			Because $\alpha_{t-1} > 0$, Alice will cut at $1$ in round $t$, so $\beta_t = \beta_{t-1} + 1$. But then $\beta_t > \beta_{t-1} \geq 0$, contradicting $\beta_t < 0$. Therefore, this case cannot happen.
			
			\item[\bf (iii) $\alpha_{t-1} = 0$ and $\beta_{t-1} \geq 1$.]
			By Lemma~\ref{lem:fictitious-basic-dynamics}, Alice can cut wherever she likes, but Bob will pick L. Wherever Alice cuts, we will have $\alpha_t = -1$ and $\beta_t \geq 0$. In order to return to $\alpha = 0$, Bob must start picking R, but he cannot do so until $\beta \leq 0$ by Lemma~\ref{lem:fictitious-basic-dynamics}. Therefore, the next axis-crossing round $\tau$ will be the one where $\beta_{\tau+1} < 0$ and $\beta_{\tau} \geq 0$. Until then, Alice will keep cutting at 0, so $\beta$ will decrease by 1 every round. Therefore, $\tau = t + \lfloor \beta_t \rfloor$, and this implies that
			\begin{align*}
				\tau - t &= \lfloor \beta_t \rfloor 
				\leq |\alpha_t| +|\beta_t| = \rho_t \,.
			\end{align*}
			
			Again to show $\rho_t -2 \le \tau - t$, observe that
			\begin{align*}
				\tau - t &= \lfloor \beta_t \rfloor \\
				&\geq |\alpha_t| - 1 + |\beta_t| - 1 \explain{Since $\alpha_t = -1$ and $\beta_t \geq 0$} \\
				&= \rho_t - 2\,.
			\end{align*}
			
			\item[\bf  (iv) $\alpha_{t-1} = 0$ and $0 \leq \beta_{t-1} < 1$.]
			Under these constraints, we have $\rho_{t-1} = |\alpha_{t-1}| + |\beta_{t-1}| < 1$, so  part (d) of Lemma~\ref{lem:fictitions-phase-space-radius-nondecreasing} implies that  $t-1 = 0$. Therefore, we have $\alpha_{t-1} = \beta_{t-1} = 0$. Accounting for all possible choices Alice and Bob can make, it must be the case that $\alpha_t = \pm 1$ and $-1 \leq \beta_t \leq 1$, which implies that $1 \leq \rho_t \leq 2$. 
			Thus it suffices to show that $\tau - t \leq 1$. We have a few cases:
			\begin{itemize}
				\item If $\alpha_t = 1$ and $\beta_t > 0$, then Bob will pick L in round $t+1$ by Lemma~\ref{lem:fictitious-basic-dynamics}. Therefore, $\alpha_{t+1} = 0$, so $\tau = t + 1$.
				\item If $\alpha_t = 1$ and $-1 < \beta_t \leq 0$, then Alice will cut at 1 in round $t+1$ by Lemma~\ref{lem:fictitious-basic-dynamics}. Therefore, $\beta_{t+1} > 0$, so round $t$ is axis-crossing and $\tau - t = 0$.
				\item If $\alpha_t = 1$ and $\beta_t = -1$, then Alice will cut at 1 and Bob will pick R in round $t+1$ by Lemma~\ref{lem:fictitious-basic-dynamics}. That will lead to $\alpha_{t+1} = 2$ and $\beta_{t+1} = 0$, so round $t+1$ is axis-crossing. Therefore, $\tau - t = 1$.
				\item If $\alpha_t = -1$, we can reduce to the $\alpha_t = 1$ case by Lemma~\ref{lem:fictitious-symmetry}.
			\end{itemize}
		\end{description}
		Thus for all types \textbf{(i)}-\textbf{(iv)}, it follows that $\rho_t - 2 \leq \tau - t \leq \rho_t$ as desired.
		% (Sketch) Because $t-1$ is an axis-crossing round, it must lie within 1 unit clockwise of one of the axes. Apply a rotation change of coordinates by a multiple of $90\deg$ under which $(\alpha_{t-1}, \beta_{t-1})$ are mapped to $(x_{t-1}, y_{t-1})$ where $x_{t-1} > 0$ and $-1 < y_{t-1} \leq 0$. The reason this still works is that, for the rounds up to and including $\tau$, the dynamics will be ones that decrease the $x$'s by 1 in each round and maintain $y > 0$.
		% Even more strongly, $\tau$ will be the first round where $x_{\tau+1} < 0$. Therefore, the number of non-axis-crossing rounds between $t-1$ and $\tau$ is exactly $\lfloor x_t \rfloor$. Because $0 < y_t \leq 1$, we therefore have:
		% \begin{align*}
		%     \tau - t &= \lfloor x_t \rfloor \\
		%     &\geq (x_t - 1) + (y_t - 1) \\
		%     &= \rho_t - 2
		% \end{align*}
		% And:
		% \begin{align*}
		%     \tau - t &= \lfloor x_t \rfloor \\
		%     &\leq x_t + y_t \\
		%     &= \rho_t
		% \end{align*}
	\end{proof}
	
	%Rather than separately bounding Alice and Bob's payoffs, we can bound the \emph{total} payoff to both of them by $T + \Theta(\sqrt{T})$, then bound Bob's payoff by $\frac{1}{2}T + \Theta(\sqrt{T})$.
	The next two lemmas show different conditions under which $\rho$ must increase. The first (Lemma \ref{lem:fictitious-play-fractional-beta-increases-rho}) is more technical in nature, while the second (Lemma \ref{lem:fictitious-play-alice-middle-cut-bound}) is key to bounding the players' total payoff.
	
	\begin{lemma}
		\label{lem:fictitious-play-fractional-beta-increases-rho}
		% \sscomment{no comma rule}
		Suppose there exists a round $t$ such that   $\beta_t \not \in \mathbb{Z}$. Let $\tau > t$ be the first round after $t$ such that $\beta_t \beta_{\tau} \leq 0$, \ie $\beta_{\tau}$ is zero or has the opposite sign of $\beta_{t}$. Then the following inequality holds:
		\[
		\rho_{\tau} \geq \lfloor \rho_t \rfloor + 1.
		\]
	\end{lemma}
	\begin{proof}
		Again by Lemma \ref{lem:fictitious-symmetry}, we can assume without loss of generality that $\beta_t \geq 0$. Since $\beta_t \not \in \mathbb{Z}$, we can further assume $\beta_t > 0$.
		
		Suppose that $\beta_{\tau} \not \in \mathbb{Z}$. Because $\tau > t$ is the first round after $t$ with $\beta_{\tau} \leq 0$, 
		we must have $\beta_{\tau-1} > 0$. Since $\beta_{\tau} \not \in \mathbb{Z}$, we also have $\beta_{\tau} < 0$.  
		Combining $|\beta_{\tau} - \beta_{\tau-1}| \leq 1$ and  $\beta_{\tau} < 0$ yields  $0 < \beta_{\tau-1} < 1$. Since $\beta_{\tau-1} > 0$, Bob picked L in round $\tau$ by Lemma~\ref{lem:fictitious-basic-dynamics}, so we have $\alpha_{\tau} = \alpha_{\tau-1} - 1$. As $\beta_{\tau} < \beta_{\tau-1}$, Alice must  have not cut at 1 in round $\tau$ by Lemma~\ref{lem:fictitious-play-actions-affect-alpha-beta}. This implies that $\alpha_{\tau-1} \leq 0$. 
		Finally, we obtain
		\begin{align*}
			\rho_{\tau} = |\alpha_{\tau}| + |\beta_{\tau}|
			&\geq |\alpha_{\tau-1} - 1| + 0 \explain{Since $\alpha_\tau =\alpha_{\tau-1} - 1$}\\
			&= 1 - \alpha_{\tau-1} \explain{Since $\alpha_{\tau-1} \le 0$} \\
			&= 1 + |\alpha_{\tau-1}| + \lfloor |\beta_{\tau-1}| \rfloor \explain{Since  $\alpha_{\tau-1} \le 0$ and $0 < \beta_{\tau-1} < 1$}\\
			&= 1 + \lfloor \rho_{\tau-1} \rfloor \explain{Since $\alpha_{\tau-1} \in \mathbb{Z}$} \\
			&\geq 1 + \lfloor \rho_t \rfloor \,. \explain{By Lemma \ref{lem:fictitions-phase-space-radius-nondecreasing}}
		\end{align*}
		
		On the other hand, suppose $\beta_{\tau} \in \mathbb{Z}$. By Lemma \ref{lem:fictitions-phase-space-radius-nondecreasing}, we have $\rho_{\tau} \geq \rho_t$. Since $\beta_t \not \in \mathbb{Z}$ but $\alpha_t \in \mathbb{Z}$ and $\alpha_{\tau} \in \mathbb{Z}$, we have  $\rho_{\tau} \in \mathbb{Z}$ but $\rho_t \not \in \mathbb{Z}$. Therefore, $\lfloor \rho_t \rfloor < \rho_t \leq \rho_{\tau}$. Since both $\lfloor \rho_t \rfloor \in \mathbb{Z}$ and $\rho_{\tau} \in \mathbb{Z}$, we have $\rho_{\tau} \geq \lfloor \rho_t \rfloor + 1$.
	\end{proof}
	
	\begin{lemma}
		\label{lem:fictitious-play-alice-middle-cut-bound}
		Let $t$ be a round in which Alice cuts at  $a_t \in (0,1)$. 
		Let  $\tau-1$ be the first axis-crossing round strictly after $t-1$. Then 
		% If $\tau$ is the round after the next axis-crossing round after $t-1$, we have:
		$
		\rho_{\tau} \geq \lfloor \rho_{t-1} \rfloor + 1.
		$
	\end{lemma}
	
	\begin{proof}
		By Lemma~\ref{lem:fictitious-basic-dynamics}, Alice will cut at 0 or 1 if $\alpha_{t-1} \neq 0$.
		As Alice does not cut at $0$ or $1$ in round $t$, this implies that $\alpha_{t-1} = 0$.
		%Then, we can break into cases based on $\beta_{t-1}$.
		First, if $\beta_{t-1} = 0$, then  Lemma \ref{lem:fictitions-phase-space-radius-nondecreasing} implies  $t-1=0$, so $\rho_{\tau} \geq 1 = \lfloor \rho_{t-1} \rfloor + 1$ as required.
		
		Otherwise,  by Lemma \ref{lem:fictitious-symmetry}, we can assume without loss of generality that $\beta_{t-1} > 0$. Because $\rho_{t-1} \geq |\beta_{t-1}| > 0$,   Lemma \ref{lem:fictitions-phase-space-radius-nondecreasing} implies  $t-1 \geq 1$. 
		By part (d) of Lemma \ref{lem:fictitions-phase-space-radius-nondecreasing}, we further have 
		\begin{align*}
			1 \leq \rho_{t-1}  = |\alpha_{t-1}| + |\beta_{t-1}| = \beta_{t-1}  \,. 
		\end{align*}
		Since Alice does not cut at 0 in round $t$, we have $\beta_t > \beta_{t-1} - 1$ by Lemma~\ref{lem:fictitious-basic-dynamics}, which implies that $\beta_t > 0$. 
		As $\beta_{t-1} > 0$, Bob picked L in round $t$ by Lemma~\ref{lem:fictitious-basic-dynamics}, so we have $\alpha_t < 0$. Again by Lemma~\ref{lem:fictitious-basic-dynamics}, for $\alpha$ to return to 0, Bob would have to start picking R, but he won't until $\beta$ stops being positive. Therefore, the next axis-crossing round after $t$ will be one where $\beta$ crosses into being non-positive from positive, so $\beta_{\tau} < 0$.
		
		Let $s > t-1$ be the first round after $t-1$ such that $\beta_s \leq 0$. Because $\beta_{\tau} < 0$, we have $s \leq \tau$. Because $\beta_t > 0$, we have $s > t$.
		
		If $\beta_{t-1} \not \in \mathbb{Z}$, then applying Lemma \ref{lem:fictitious-play-fractional-beta-increases-rho} to $t-1$ yields $\rho_s \geq \lfloor \rho_{t-1} \rfloor + 1$. By Lemma \ref{lem:fictitions-phase-space-radius-nondecreasing}, we have $\rho_{\tau} \geq \rho_s$, and so $\rho_{\tau} \geq \rho_s \geq \lfloor \rho_{t-1} \rfloor + 1$ as required.
		
		If $\beta_{t} \not \in \mathbb{Z}$, then applying Lemma \ref{lem:fictitious-play-fractional-beta-increases-rho} to $t$ yields $\rho_s \geq \lfloor \rho_t \rfloor + 1$. By Lemma \ref{lem:fictitions-phase-space-radius-nondecreasing}, we have $\rho_{\tau} \geq \rho_s$ and $\rho_t \geq \rho_{t-1}$, and so $\rho_{\tau} \geq \rho_s \geq \lfloor \rho_t \rfloor + 1 \geq \lfloor \rho_{t-1} \rfloor + 1$ as required.
		
		Else, we have $\beta_{t-1}, \beta_t \in \mathbb{Z}$. Since Alice did not cut at 0 or 1 in round $t$, by Lemma~\ref{lem:fictitious-play-actions-affect-alpha-beta} we have $|\beta_{t} - \beta_{t-1}| < 1$. Since $\beta_t, \beta_{t-1} \in \mathbb{Z}$, we get $\beta_{t-1} = \beta_t$. 
		Moreover, since $\alpha_{t-1} = 0$, we have $|\alpha_t| = |\alpha_{t-1} \pm 1| = 1$.
		% But $\alpha_{t-1} = 0$ and $\alpha_t \neq 0$, and since $\alpha_t$ is an integer, $|\alpha_t| \geq 1$. 
		This implies that
		\begin{align}
			\rho_{\tau} &\geq \rho_t \explain{By Lemma \ref{lem:fictitions-phase-space-radius-nondecreasing}} \\
			&= \rho_{t-1} + 1 \explain{Since $\beta_t = \beta_{t-1}$ and  $|\alpha_t| = |\alpha_{t-1}| + 1$} \\
			&\geq \lfloor \rho_{t-1} \rfloor + 1\,. \notag
		\end{align}
		This completes the proof.
	\end{proof}
	
	The following lemma bounds Bob's total payoff using the radius $\rho_t$.
	\begin{lemma}
		\label{lem:fictitious-bob-payoff-bounded-by-rho}
		For every round $t \geq 0$, the following inequalities hold:
		\begin{align}
			-\rho_t \leq \sum_{i=1}^t (2u_B^i - 1) \leq \rho_t.
			\label{ineq:02021757}
		\end{align}
	\end{lemma}
	
	\begin{proof}
		We will first show the following stronger set of inequalities by induction on $t$:
		\begin{equation}
			\label{eq:fictitious-bob-payoff-bounded-by-rho-induction}
			0 \leq |\alpha_t| + \sum_{i=1}^t (2u_B^i - 1) \leq \rho_t.
		\end{equation}
		As a base case, consider $t=0$. Before anything has happened, all three sides of \eqref{eq:fictitious-bob-payoff-bounded-by-rho-induction} are zero, so the inequalities trivially hold.
		
		Now assume that \eqref{eq:fictitious-bob-payoff-bounded-by-rho-induction} holds for some $t \geq 0$, and we will prove that it still holds for round $t+1$. 
		We consider three cases separately in what follows, depending on the values of $\alpha_t$ and $\beta_t$. 
		Note that again by Lemma \ref{lem:fictitious-symmetry}, it suffices to only consider cases where $\alpha_t \geq 0$.
		
		\begin{itemize}
			\item If $\alpha_t > 0$, Alice will cut at 1 in round $t+1$ by Lemma~\ref{lem:fictitious-basic-dynamics}. If Bob picks $L$, then he will receive a payoff of 1 and $\alpha_t$ will decrease by 1 due to Lemma~\ref{lem:fictitious-play-actions-affect-alpha-beta}. If Bob picks $R$, then he will receive a payoff of 0 and $\alpha_t$ will increase by 1 due to Lemma~\ref{lem:fictitious-play-actions-affect-alpha-beta}. In either case, the changes to $|\alpha_t| + \sum_{i=1}^t (2u_B(i) - 1)$ cancel out, which implies that
			\begin{equation}
				\label{eq:fictitious-alpha-t-positive-inductive-step}
				|\alpha_t| + \sum_{i=1}^t (2u_B^i - 1) = |\alpha_{t+1}| + \sum_{i=1}^{t+1} (2u_B^i - 1).
			\end{equation}
			
			Further, by Lemma~\ref{lem:fictitions-phase-space-radius-nondecreasing} we have $\rho_{t+1} \ge \rho_t$.
			Together with the induction hypothesis \eqref{eq:fictitious-bob-payoff-bounded-by-rho-induction},
			this concludes
			% As $0$ stays constant and $\rho_t$ doesn't decrease by Lemma \ref{lem:fictitions-phase-space-radius-nondecreasing}, the inductive hypothesis \eqref{eq:fictitious-bob-payoff-bounded-by-rho-induction} can be combined with \eqref{eq:fictitious-alpha-t-positive-inductive-step} to show:
			\begin{equation*}
				0 \leq |\alpha_{t+1}| + \sum_{i=1}^{t+1} (2u_B^i - 1) \leq \rho_{t+1},
			\end{equation*}
			and thus the induction holds for the first case.
			
			\item If $\alpha_t = 0$ and $\beta_t \geq 0$, then regardless of whether Bob picks L or R in round $t+1$ we have $|\alpha_{t+1}| = |0 \pm 1| = 1$ as $\alpha$ necessarily changes by $1$.
			Therefore, we have
			\begin{align}
				\left(|\alpha_{t+1}| + \sum_{i=1}^{t+1} (2u_B^i - 1)\right) - \left(|\alpha_t| + \sum_{i=1}^t (2u_B^i - 1)\right) &= 1 + 2u_B^{t+1} - 1 \notag \\
				\label{eq:fictitious-bob-alpha-zero-case-payoff-change}
				&= 2u_B^{t+1}
			\end{align}
			
			Also, the change in $\rho$ can be bounded as follows:
			\begin{align}
				\rho_{t+1} - \rho_t &= |\alpha_{t+1}| + |\beta_{t+1}| - |\alpha_t| - |\beta_t| \notag \\
				&= 1 + \left|\sum_{i=1}^{t+1} (2V_B([0, a_i]) - 1)\right| - 0 - \beta_t \explain{$\beta_t \geq 0$} \\
				&= 1 + \left|(2V_B([0, a_{t+1}]) - 1) + \beta_t \right| - \beta_t \notag \\
				&\geq 1 + 2V_B([0, a_{t+1}]) - 1 \explain{Removing the absolute value} \\
				\label{eq:fictitious-bob-alpha-zero-case-radius-weak-change}
				&= 2V_B([0, a_{t+1}])
			\end{align}
			
			If Bob picked $L$ in round $t+1$, this is exactly the same as \eqref{eq:fictitious-bob-alpha-zero-case-payoff-change}. If Bob picked $R$, then by Lemma~\ref{lem:fictitious-basic-dynamics} and the assumption that $\beta_t \geq 0$ we must have $\beta_t = 0$. The change in the radius can be bounded as follows:
			\begin{align}
				\rho_{t+1} - \rho_t &= |\alpha_{t+1}| + |\beta_{t+1}| - |\alpha_t| - |\beta_t| \notag \\
				&= 1 + |2V_B([0, a_{t+1}]) - 1| - 0 - 0 \explain{$\beta_t = 0$} \\
				&= 1 + |1 - 2V_B([a_{t+1}, 1])| \explain{$V_B([0, a_{t+1}]) + V_B([a_{t+1}, 1]) = 1$} \\
				&= 1 + \left|-\left(1 - 2V_B([a_{t+1}, 1])\right)\right| \notag \\
				&\geq 1 - 1 + 2V_B([a_{t+1}, 1]) \explain{Removing the absolute value} \\
				&= 2u_B^{t+1} \explain{Since Bob picked $R$}
			\end{align}
			
			In either case, the radius increased by at least as much as the middle of \eqref{eq:fictitious-bob-payoff-bounded-by-rho-induction}. 
			More precisely, by the induction hypothesis \eqref{eq:fictitious-bob-payoff-bounded-by-rho-induction}, we obtain
			\begin{align*}
				|\alpha_{t+1}| + \sum_{i=1}^t (2u_B^i-1)
				&= |\alpha_t| + \sum_{i=1}^t (2u_B^i-1) + 2u_B^{t+1}-1 + |\alpha_{t+1}| - |\alpha_t|
				\\
				&\ge 0 + 2u_B^{t+1} \explain{$\alpha_t =0$ and $|\alpha_{t+1}| = 1$}
				\\
				&\ge 0. \explain{$u_B^{t+1} \ge 0$}
			\end{align*}
			Note further that
			\begin{align*}
				|\alpha_{t+1}| + \sum_{i=1}^t (2u_B^i-1)
				&= |\alpha_t| + \sum_{i=1}^t (2u_B^i-1) + 2u_B^{t+1}-1 + |\alpha_{t+1}| - |\alpha_t|
				\\
				&\le \rho_t + 2u_B^{t+1} \explain{By the induction hypothesis}
				\\
				&= \rho_{t+1}. \explain{$\rho_{t+1}-\rho_t = 2u_B^{t+1}$}
			\end{align*}
			This finishes the proof of the induction for the second case.
			
			\item If $\alpha_t = 0$ and $\beta_t < 0$, then by Lemma \ref{lem:fictitious-symmetry} we can consider $\alpha_t = 0$ and $\beta_t > 0$ instead, which has already been shown above.
		\end{itemize}
		
		In all cases, the inductive step holds. Therefore, by induction principle, \eqref{eq:fictitious-bob-payoff-bounded-by-rho-induction} holds for all $t$. Subtracting $|\alpha_t|$ from all three sides of it, we obtain
		\begin{equation*}
			-|\alpha_t| \leq \sum_{i=1}^t (2u_B^i - 1) \leq |\beta_t|.
		\end{equation*}
		Since $\rho_t = |\alpha_t| + |\beta_t|$, this immediately implies the desired bounds.
	\end{proof}

	Now that we have bounds on the relevant events in terms of $\alpha$, $\beta$, and $\rho$, we can bound them as functions of $T$ to obtain our final result.
	
	\begin{lemma}
		\label{lem:fictitious-play-n-squared-rho-increases}
		Let $n \geq 7$. Let $t_1, t_2, \dots, t_n$ be a sequence of rounds such that for all $i \in [n-1]$ the following inequality holds:
		\begin{align}
			\rho_{t_{i+1}} \geq \left\lfloor \rho_{t_i} \right\rfloor + 1.\label{ineq:02070842}
		\end{align}
		Then $t_n - t_1 > \frac{1}{10}n^2$.
	\end{lemma}
	
	%\rpcomment{The proof shows a much stronger bound of asymptotically $\frac{1}{2}n^2$, but with some extra lower-order terms tacked on. They're more annoying to deal with than it's worth if the goal is just to show that the payoffs are within $O(\sqrt{T})$ of $1/2$.}
	
	\begin{proof}
		First, we will show by induction that, for $i \in [n]$, we have $\rho_{t_i} \geq i-1$.
		For the base case, $\rho_{t_1} \geq \rho_0 = 0$, so the inequality trivially holds.
		
		Now assume $\rho_{t_i} \geq i-1$ for some $i \in [n-1]$ and we will prove that the induction step holds for the case $i+1$. 
		By~\eqref{ineq:02070842}, observe that
		\begin{align*}
			\rho_{t_{i+1}} &\geq \left\lfloor \rho_{t_i} \right\rfloor + 1 \\
			&\geq \lfloor i-1 \rfloor + 1 \\
			&= (i+1)-1.
		\end{align*}
		Thus, by the induction principle we have
		\begin{equation}
			\label{eq:fictitious-play-rho-t-i-greater-than-i}
			\rho_{t_i} \geq i-1 \text{ for $i = 1, \dots, n$}
		\end{equation}
		
		Now consider an arbitrary $i \in [n-1]$. Note that $\rho_{t_{i+1}} \geq \left\lfloor \rho_{t_i} \right\rfloor + 1$ implies $\rho_{t_{i+1}} > \rho_{t_i}$. Therefore, by Lemma \ref{lem:fictitions-phase-space-radius-nondecreasing}, there must be an axis-crossing round $c_i$ among the interval $[t_i, t_{i+1})$.
		This is because if this is not true, we have $\rho_{t_{i+1}} = \rho_{t_i}$ which contradicts $\rho_{t_{i+1}} > \rho_{t_i}$.
		Repeating the same argument for each $i$, we conclude that there exists a sequence of rounds $c_1, c_2, \dots, c_{n-1}$ each of which is axis-crossing, and satisfies the following inequality for every $i \in [n-1]$:
		\begin{equation}
			\label{eq:fictitious-play-ci-between-ti}
			t_i \leq c_i < t_{i+1}
		\end{equation}
		In addition, for any $i \in [n-1]$, we observe that
		\begin{align*}
			c_{i+1} - (c_i+1) &\geq \rho_{c_{i}+1} - 2 \explain{By Lemma \ref{lem:fictitious-play-axis-crossing-spacing}, with $\tau = c_{i+1}$ and $t=c_{i}+1$} \\
			&\geq \rho_{t_i} - 2 \explain{By \eqref{eq:fictitious-play-ci-between-ti}} \\
			&\geq i - 3 \explain{By \eqref{eq:fictitious-play-rho-t-i-greater-than-i}}
		\end{align*}
		Slightly rearranging, we obtain
		\begin{align}
			\label{eq:fictitious-play-c-gap-large}
			c_{i+1} - c_i \geq i-2 \; \; \forall i \in [n-1]\,.
		\end{align}
		%    for $i\in [n-1]$.
		
		Combining~\eqref{eq:fictitious-play-ci-between-ti} and~\eqref{eq:fictitious-play-c-gap-large}, we finally obtain
		\begin{align*}
			t_n - t_1 &> c_{n-1} - c_1 \explain{By \eqref{eq:fictitious-play-ci-between-ti}} \\
			&= \sum_{i=1}^{n-2} (c_{i+1} - c_i) \\
			&\geq \sum_{i=1}^{n-2} (i-2) \explain{By \eqref{eq:fictitious-play-c-gap-large}} \\
			&= \frac{1}{2}n^2 - \frac{7}{2}n + 5.
		\end{align*}
		For $n\geq 7$, we have $\frac{1}{2}n^2 - \frac{7}{2}n + 5 \geq  \frac{1}{10}n^2$,\footnote{We omit the elementary calculus.} and so
		$t_n - t_1 > n^2/10$, as required.
	\end{proof}
	
	The following lemma will finally be combined with Lemma~\ref{lem:fictitious-bob-payoff-bounded-by-rho} to obtain the desired bound for Bob's payoff.
	\begin{lemma}
		\label{lem:fictitious-play-rho-T-upper-bound}
		For $T \geq 5$, the final radius $\rho_T$ satisfies the following:
		\[
		\rho_T \leq 2\sqrt{10T}.
		\]
	\end{lemma}
	
	\begin{proof}
		Let $t_1=0$, and for $i \geq 2$ recursively define  $t_i$ be the first round after $t_{i-1}$ satisfying $\rho_{t_i} \geq \lfloor \rho_{t_{i-1}} \rfloor + 1$.
		Let $n$ be the last index of such $t_i$ given the time horizon $T$.
		
		As an immediate corollary of Lemma \ref{lem:fictitious-play-n-squared-rho-increases}, we have
		\begin{equation*}
			n \leq \max\{7, \sqrt{10T}\}.
		\end{equation*}
		For $T \ge 5$, this implies that
		% Which, for $T \geq 5$, eliminates the first case to give:
		\begin{equation}
			\label{eq:fictitious-play-n-less-than-sqrt-10T}
			n \leq \sqrt{10T}\,.
		\end{equation}
		
		By Lemma \ref{lem:fictitions-phase-space-radius-nondecreasing}, we also know that $\rho$ can increase by at most $2$ per round. 
		Furthermore, since $t_n$ is the last element in the sequence $\{t_i\}_{i \in [n]}$, we have $\rho_T < \lfloor \rho_{t_n} \rfloor + 1$. Relaxing this slightly, we have
		\begin{equation}
			\label{eq:fictitious-play-rho-T-less-than-rho-t-n-plus-2}
			\rho_T \leq \rho_{t_n}+2\,.
		\end{equation}
		
		Therefore, we obtain the following inequalities:
		\begin{align}
			\rho_T &\leq 2+\rho_{t_n} \explain{By \eqref{eq:fictitious-play-rho-T-less-than-rho-t-n-plus-2}} \\
			&= 2 + \sum_{i=1}^{n-1} (\rho_{i+1} - \rho_i) \notag \\
			&\leq 2 + 2(n-1) \explain{By Lemma \ref{lem:fictitions-phase-space-radius-nondecreasing}}\\
			\label{eq:fictitious-play-rho-T-less-than-2n}
			&= 2n \,.
		\end{align}
		
		Putting \eqref{eq:fictitious-play-n-less-than-sqrt-10T} and \eqref{eq:fictitious-play-rho-T-less-than-2n} together, we obtain
		\begin{align*}
			\rho_T &\leq 2\sqrt{10T},
		\end{align*}
		which finishes the proof of the lemma.
	\end{proof}
	
	Using the above bound on the radius together with Lemma~\ref{lem:fictitious-bob-payoff-bounded-by-rho}, Bob's payoff can now be bounded as follows.
	\begin{lemma}
		\label{lem:fictitious-play-bob-payoff-bound}
		For $T \geq 5$, Bob's total payoff satisfies:
		\[
		\frac{T}{2} - \sqrt{10T} \leq \sum_{t=1}^T u_B^t \leq \frac{T}{2} + \sqrt{10T}\,.
		\]
	\end{lemma}
	
	\begin{proof}
		We start from the inequality~\eqref{ineq:02021757}.
		Halving all three sides of~~\eqref{ineq:02021757} for $t=T$ and adding $T/2$, we obtain
		\[
		\frac{T}{2} - \frac{1}{2}\rho_T \leq \sum_{t=1}^T u_B^t \leq \frac{T}{2} + \frac{1}{2}\rho_T.
		\]
		Using the upper bound $\rho_T \leq 2\sqrt{10T}$ from Lemma~\ref{lem:fictitious-play-rho-T-upper-bound} gives the desired bounds.
	\end{proof}

	Combined with Lemma~\ref{lem:fictitious-play-bob-payoff-bound}, Alice's payoff can eventually be bounded using the following lemma, which effectively bounds the summation of Alice and Bob's total payoff.
	\begin{lemma}
		\label{lem:fictitious-play-total-payoff-bound}
		For $T \geq 5$, the summation of total payoff to Alice and Bob satisfies the following:
		\[
		T - \sqrt{10T} \leq \sum_{t=1}^T \big(u_A^t + u_B^t\big) \leq T + \sqrt{10T}\,.
		\]
	\end{lemma}
	
	%\rpcomment{The $T \geq 5$ condition can be removed, since the bound is meaningless anyway for $T \leq 10$ (superseded by the trivial $0$ and $2T$). However, putting that in adds a little bit of casework that probably won't be useful anyway, so I'll leave it out for now.}
	
	\begin{proof}
		Given the time horizon $T$, let $s_1, s_2, \dots, s_k$ be the rounds in which Alice cuts at a point other than $0$ or $1$, where $k$ denotes the number of such rounds.
		These are the only rounds in which the total payoff $u_A^t + u_B^t$ is not necessarily $1$.
		The total payoff in these rounds can be bounded as
		\begin{align*}
			0 \le u_A^{s_i} + u_B^{s_i} \le 2 \qquad  \forall i \in [k]\,.
		\end{align*}
		%for any $i \in [k]$.
		Thus, we obtain
		\begin{align*}
			T- k \le \sum_{t=1}^T u_A^t + u_B^t \le T+k.
		\end{align*}
		Therefore, it suffices to prove that $k  \le \sqrt{10 T}$.
		
		If $k < 7$, the proof follows as $k < 7 \le \sqrt{10T}$.
		
		Otherwise, we have $k \ge 7$.
		% These rounds are important because they are the only ones in which the total payoff $u_A(t) + u_B(t)$ is not necessarily $1$. In rounds $t_1, \dots, t_k$, the total payoff may be as low as $0$ or as high as $2$ -- a deviation of $1$ in either direction. As there are $k$ such rounds:
		% \begin{equation}
		%     \label{eq:fictitious-play-total-payoff-in-terms-of-k}
		%     T - k \leq \sum_{t=1}^T u_A(t) + u_B(t) \leq T + k
		% \end{equation}
		% So the lemma will follow immediately from showing $k \leq \sqrt{10T}$. We divide into two cases. For $k < 7$, we can immediately conclude:
		% \begin{align*}
		%     k &< 7 \\
		%     &\leq \sqrt{10T} \explain{$T \geq 5$}
		% \end{align*}
		% As desired. For $k \geq 7$, we define additional useful rounds. 
		For each $s_i$, let $\tau_i$ be the round after the next axis-crossing round after $s_i-1$.
		Consider $s_i$ for an arbitrary $i \in [k]$.
		Due to Lemma~\ref{lem:fictitious-basic-dynamics}, if $\alpha_{s_i -1 } \neq 0$ then Alice cuts at $0$ or $1$ in round $s_i$, which contradicts our definition of round $s_i$.
		Thus we have $\alpha_{s_i -1} = 0$.
		% It must be that $\alpha_{t_i-1} = 0$, as if $\alpha_{t_i-1} \neq 0$ the fictitious play dynamics force Alice to cut at $0$ or $1$ in round $t_i$ which violates our definition of the round $t_i$.
		This argument holds for arbitrary $i \in [k]$, so both $s_i-1$ and $s_{i+1}-1$ are axis-crossing rounds. 
		Moreover, since both $\alpha_{s_i-1}=0$ and $\alpha_{s_{i+1}-1} = 0$, there must have been some rounds between $s_i-1$ and $s_{i+1}-1$ where Bob picked $L$ and some where he picked $R$. Therefore, $\beta$ must have changed sign at least once, so there is another axis-crossing round in between $s_i-1$ and $s_{i+1}-1$ where $\beta$ changed sign.
		This implies that for every $i=1,\ldots, k-1$, the next axis crossing round $\tau_i$ after $s_i-1$ should exist at least before $s_{i+1}-1$, \ie 
		\begin{equation}
			\label{eq:fictitious-play-alice-middle-times-apart}
			\tau_i \leq s_{i+1}-1\,.
		\end{equation}
		Due to the monotonicity of $\rho$ by Lemma~\ref{lem:fictitions-phase-space-radius-nondecreasing}, for any $i \in \{1, \dots, k-1\}$ we have
		\begin{align*}
			\rho_{s_{i+1}-1} &\geq \rho_{\tau_i} \explain{By \eqref{eq:fictitious-play-alice-middle-times-apart}} \\
			&\geq \left\lfloor \rho_{s_i-1} \right\rfloor + 1\,. \explain{By Lemma \ref{lem:fictitious-play-alice-middle-cut-bound}}
		\end{align*}
		By Lemma \ref{lem:fictitious-play-n-squared-rho-increases}, we have:
		\begin{equation}
			\label{eq:fictitious-play-tk-minus-to-unsimplified}
			(s_k-1) - (s_1-1) > \frac{1}{10}k^2 \,.
		\end{equation}
		Re-arranging \eqref{eq:fictitious-play-tk-minus-to-unsimplified} and using the fact that $s_k - s_1 \leq T$, we obtain
		%\begin{equation*}
		%    \label{eq:fictitious-play-tk-minus-to-simplified}
		$k < \sqrt{10T}$.
		%\end{equation*}
		This finishes the proof of the case $k \ge 7$, which completes the lemma.
	\end{proof}

\end{document}